\documentclass[a4paper,11pt,]{article}
\usepackage{graphicx}
\usepackage{tabularx}
 \usepackage{amssymb}
 \usepackage{amsmath}
\usepackage{multirow}

\newtheorem{theorem}{Theorem}[section]

\newtheorem{LMA}[theorem]{Lemma}

\global \count10 = 0

\global \count11 = 1

\global \count12 = 1

\global \count13 = 1

\newcommand{\com}[1]{\ifnum\count13<1 #1 \fi}

\global \count14 = 0

\def\squarebox#1{\hbox to #1{\hfill\vbox to #1{\vfill}}}
\def\qed{\hspace*{\fill}%
        \vbox{\hrule\hbox{\vrule\squarebox{.667em}\vrule}\hrule}\smallskip}
\newenvironment{proof}{\begin{trivlist}
\item[\hspace{\labelsep}{\em\noindent Proof.~}]}{\qed\end{trivlist}}
\def\squarebox#1{\hbox to #1{\hfill\vbox to #1{\vfill}}}
\def\qed{\hspace*{\fill}%
        \vbox{\hrule\hbox{\vrule\squarebox{.667em}\vrule}\hrule}\smallskip}

\begin{document}

\title{Tight Analysis of Priority Queuing\\ for Egress Traffic}

\author{
	Jun Kawahara$^{1}$, 
	Koji M. Kobayashi$^{2}$, and 
	Tomotaka Maeda$^{3}$
	\\   
	{\footnotesize 
		$^{1}$Graduate School of Information Science, Nara Institute of Science and Technology
	}
	\\
	{\footnotesize 
		$^{2}$National Institute of Informatics, 
	}
	\\   
	{\footnotesize 
		$^{3}$Academic Center for Computing and Media Studies, Kyoto University
	}
}

\date{}

\setlength{\baselineskip}{4.85mm}

\maketitle

\pagestyle{plain}
\thispagestyle{plain}

\begin{abstract}
Recently, 
the problems of evaluating performances of switches and routers have been formulated as online problems, 
and a great amount of results have been presented. 
In this paper, 
we focus on managing outgoing packets (called {\em egress traffic}) on switches that support Quality of Service (QoS), 
and analyze the performance of one of the most fundamental scheduling policies {\em Priority Queuing} ($PQ$) using competitive analysis. 
We formulate the problem of managing egress queues as follows:
An output interface is equipped with $m$ queues, each of which has a buffer of size $B$. 
The size of a packet is unit, and each buffer can store up to $B$ packets simultaneously. 
Each packet is associated with one of $m$ priority values $\alpha_{j}$ ($1 \leq j \leq m$), 
where $\alpha_{1} \leq \alpha_{2} \leq \cdots \leq \alpha_{m}$, $\alpha_{1} = 1$, and $\alpha_{m} = \alpha$ 
and the task of an online algorithm is to select one of $m$ queues at each scheduling step. 
The purpose of this problem is to maximize the sum of the values of the scheduled packets. 
For any $B$ and any $m$, 
we show that the competitive ratio of $PQ$ is exactly 
$2 - \min_{x \in [1, m-1] } \{ \frac{ \alpha_{x+1} }{ \sum_{j = 1}^{x+1} \alpha_{j} } \}$. 
That is, 
we conduct a complete analysis of the performance of $PQ$ using worst case analysis. 
Moreover, 
we show that no deterministic online algorithm can have a competitive ratio smaller than 
$1 + \frac{ \alpha^3 + \alpha^2 +  \alpha }{ \alpha^4 + 4 \alpha^3 + 3 \alpha^2 + 4 \alpha + 1 }$.
\end{abstract}
%
%



\section{Introduction}
In recent years, 
the Internet has provided a rich variety of applications, such as
teleconferencing, video streaming, IP telephone, mainly thanks to the
rapid growth of the broadband technology.  To enjoy such services, the
demand for the Quality of Service (QoS) guarantee is crucial.  For
example, usually there is little requirement for downloading programs
or picture images, whereas real-time services, such as distance meeting,
require constant-rate packet transmission.  One possible way of
supporting QoS is differentiated services (Diffserv)~\cite{SB98}. 
In DiffServ, 
a value is assigned to each packet according
to the importance of the packet.  
Then, switches that support QoS (QoS switches) decide the order of packets to be processed, 
based on the value of packets. 
In such a mechanism, 
one of the main issues in designing algorithms is how to treat packets depending on the priority in buffering or scheduling.
This kind of problems was recently modeled as an {\em online problem},
and the {\em competitive analysis}~\cite{AB98,DS85} of algorithms has been done.

Aiello et~al.\ \cite{WA00} was the first to attempt this study, in which
they considered a model with only one First In First Out (FIFO) queue.
This model mainly focuses on the buffer management issue of the input
port of QoS switches: There is one FIFO queue of size $B$, meaning that
it can store up to $B$ packets. 
An input is a sequence of events. 
An event is either an {\em arrival event}, at which a packet with a specified priority value arrives, or a {\em scheduling event}, 
at which the packet at the head of the queue will be transmitted. 
The task of an online (buffer management) algorithm is to decide, when a packet arrives
at an arrival event, whether to accept or to reject it (in order to keep
a room for future packets with higher priority).  The purpose of the
problem is to maximize the sum of the values of the transmitted packets.
Aiello et~al.\ analyzed the competitiveness of the Greedy Policy, the
Round Robin Policy, the Fixed Partition Policy, etc.

After the publication of this seminal paper, 
more and more complicated models have been introduced and studied, 
some of which are as follows:
Azar et~al.\ \cite{YA03} considered the {\em multi-queue switch model},
which formulates the buffering problem of one input port of the switch.
In this problem, 
an input port has $N$ input buffers connected to a common output buffer. 
The task of an online algorithm is now not only buffer management but also scheduling. 
At each scheduling event, 
an algorithm selects one of $N$ input buffers, 
and the packet at the head of the selected buffer is transmitted to the inside of the switch through the output buffer. 
\com{（■）}
There are some formulations that model not only one port but the entire switch.  
For example, Kesselman et~al.\ \cite{AKS03} introduced the {\em Combined Input and Output Queue (CIOQ) switch model}. 
In this model, a switch consists of $N$ input
ports and $N$ output ports, where each port has a buffer.  At an {\em
arrival phase}, a packet (with the specified destination output port)
arrives at an input port.  The task of an online algorithm is buffer
management as mentioned before.  At a {\em transmission phase}, all the
packets at the top of the nonempty buffers of output ports are
transmitted. 
Hence, there is no task of an online algorithm. 
At a {\em scheduling phase},  
packets at the top of the buffers of input ports are transmitted to the buffers of the output ports. 
Here, 
an online algorithm computes a matching between input ports and output ports. 
According to this matching, the packets in the input ports will be
transmitted to the corresponding output ports. 
Kesselman et~al.\ \cite{AKB08} considered the {\em crossbar switch model}, 
which models the scheduling phase of the CIOQ switch model more in detail. 
In this model, 
there is also a buffer for each pair of an input port and an output port. 
Thus, there arises another buffer management problem at scheduling phases.
In some real implementation (e.g., \cite{QoSSwitch}), 
additional buffers are equipped with each output port of a QoS switch to control the outgoing packets (called {\em egress traffic}). 
Assume that there are $m$ priority values of packets $\alpha_{1}, \alpha_{2}, \ldots, \alpha_{m}$
such that $\alpha_{1} \leq \alpha_{2} \leq \cdots \leq \alpha_{m}$.
Then, $m$ FIFO queues $Q^{(1)}, Q^{(2)}, \ldots, Q^{(m)}$ are introduced
for each output port, and a packet with the value $\alpha_{i}$ arriving
at this output port is stored in the queue $Q^{(i)}$. 
Usually, 
this buffering policy is greedy, namely, when a packet arrives, 
it is rejected if the corresponding queue is full, and accepted otherwise. 
The task of an algorithm is to decide which queue to transmit a packet at
each scheduling event.

Several practical algorithms, such as Priority Queuing ($PQ$), 
Weighted Round-Robin ($WRR$)~\cite{MK91}, and Weighted Fair Queuing ($WFQ$)~\cite{AD90}, 
are currently implemented in network switches.
$PQ$ is the most fundamental algorithm, which selects the highest priority non-empty queue. 
This policy is implemented in many switches by default. 
(e.g., Cisco's Catalyst 2955 series~\cite{cat2955})
In the $WRR$ algorithm, 
queues are selected according to the round robin policy based on the weight of packets corresponding to queues, 
i.e., the rate of selecting $Q^{(i)}$ in one round is
proportional to $\alpha_{i}$ for each $i$. 
This algorithm is implemented in Cisco's Catalyst 2955 series~\cite{cat2955} and so on. 
In the $WFQ$ algorithm, length of packets, as well as the priority values, are taken
into consideration so that shorter packets are more likely to be
scheduled. 
This algorithm is implemented in Cisco's Catalyst 6500 series~\cite{cat6500} and so on. 

In spite of intensive studies on online buffer management and scheduling algorithms, 
to the best of our knowledge, 
there have been no research on the egress traffic control, which we focus on in this paper. 
Our purpose is to evaluate the performances of actual scheduling algorithms for egress queues. 
\noindent
{\bf Our Results.}~~~
We formulate this problem as an online problem, 
and provide a tight analysis of the performance of $PQ$ using competitive analysis. 
Specifically, 
for any $B$, 
we show that the competitive ratio of $PQ$ is exactly 
$2 - \min_{x \in [1, m-1] } \{ \frac{ \alpha_{x+1} }{ \sum_{j = 1}^{x+1} \alpha_{j} } \}$. 
$PQ$ is trivial to implement, 
and has a lower computational load than the other policies, such as $WRR$ and $WFQ$. 
Hence, it is meaningful to analyze the exact performance of $PQ$. 
Moreover, 
we present a lower bound of $1 + \frac{ \alpha^3 + \alpha^2 +  \alpha }{ \alpha^4 + 4 \alpha^3 + 3 \alpha^2 + 4 \alpha + 1 }$ 
on the competitive ratio of any deterministic algorithm. 
\noindent
\noindent {\bf Related Work.}~~~
Independently of our work, 
Al-Bawani and Souza \cite{KA11} have very recently considered much the same model. 
$PQ$ is called the greedy algorithm in their paper.
They consider the case where $0 < \alpha_{1} < \alpha_{2} < \cdots < \alpha_{m}$. 
Also, they assume that for any $j(\in [1, m])$, the $j$th queue can store at most $B_{j} (\in [1, B])$ packets at a time. 
In the case of $B_{j} = B$, that is, in the same setting as ours, 
they showed that the competitive ratio of $PQ$ is at most $2 - \min_{j \in [1, m-1]} \{ \frac{ \alpha_{j+1} - \alpha_{j} }{ \alpha_{j+1} } \}$ for any $m$ and $B$. 
%
When comparing our result and their upper bound, 
we have 
$2 - \min_{x \in [1, m-1] } \{ \frac{ \alpha_{x+1} }{ \sum_{j = 1}^{x+1} \alpha_{j} } \} < 2 - \min_{j \in [1, m-1]} \{ \frac{ \alpha_{j+1} - \alpha_{j} }{ \alpha_{j+1} } \}$ by elementary calculation (see Sec.~\ref{ap.sec:0} in Appendix). 
Note that $2 - \min_{j \in [1, m-1]} \{ \frac{ \alpha_{j+1} - \alpha_{j} }{ \alpha_{j+1} } \}$ is equal to $2$ when there exists some $z$ such that $\alpha_{z+1} = \alpha_{z}$. 
In general practical switches, 
the sizes of any two egress queues attached to the same output port are equivalent by default. 
Since we focus on evaluating the performance of algorithms in a more practical setting (which might be less generalized), 
we assume that the size of each queue is $B$. 
Moreover, 
our analysis in this paper does not depend on the maximum numbers of packets stored in buffers, 
and instead it depends on whether buffers are full of packets. 
Thus, the exact competitive ratio of $PQ$ would be derived for the setting where for any $j$, 
the size of the $j$th queue is $B_{j}$ in the same way as this paper. 
(If we apply our method in their setting, Lemma~\ref{LMA:3.2.6} in Sec.~\ref{sec:3.2.2} has to be fixed slightly.
However the competitive ratio obtained in this setting seems to be a more complicated value including some $\min$s or $\max$es.)
%
%
%
%

%
As mentioned earlier, 
there are a lot of studies concentrating on evaluating performances of functions of switches and routers, 
such as queue management and packet scheduling.
The most basic one is the model consisting of single FIFO queue by Aiello et~al.\ \cite{WA00} mentioned above. 
In their model, each packet can take one of two values 1 or $\alpha (> 1)$. 
Andelman et~al.\ \cite{NACQ03} generalized the values of packets to any value between $1$
and $\alpha$. 
Another generalization is to allow {\em preemption}, namely, 
one may drop a packet that is already stored in a queue. 
Results of the competitiveness on this model are given in \cite{WA00,AKB01,MS01,AKI02,NACQ03,NACM03,NA05,ME06}. 
The multi-queue switch model \cite{YA03,YAZ04,KK09} consists of $m$ FIFO queues. 
In this model, the task of an algorithm is to manage its buffers and to schedule packets. 
The problem of designing only a scheduling algorithm in multi-queue switches is considered in \cite{SA04,YAM04,MB08,KK08,MB10}. 
Moreover, 
Albers and Jacobs \cite{SA07} performed an experimental study for the first time on several online scheduling algorithms for this model. 
Also, the overall performance of several switches, such as shared-memory
switches \cite{EH01,AKH01,KK07}, 
CIOQ switches \cite{AKS03,YA06,AKI08,AKC08}, 
and crossbar switches \cite{AKP08,AKB08},
are extensively studied. 

Fleischer and Koga \cite{HK04} and Bar-Noy et~al.\ \cite{AB03} studied the online problem of minimizing the length of the longest queue in a switch, 
in which the size of each queue is unbounded. 
In \cite{HK04} and \cite{AB03}, 
they showed that the competitive ratio of any online algorithm is $\Omega(\log m)$, 
where $m$ is the number of queues in a switch. 
Fleischer and Koga \cite{HK04} presented a lower bound of $\Omega(m)$ for the round robin policy. 
In addition, 
in \cite{HK04} and \cite{AB03}, 
the competitive ratio of a greedy algorithm called Longest Queue First is $O(\log m)$. 
Recently, Kogan et~al.\ \cite{KK13} studied a multi-queue switch where packets with different required processing times arrive. 
(In the other settings mentioned above, the required processing times of all packets are equivalent.) 
\section{Model Description}\label{sec:2}
In this section, we formally define the problem studied in this paper.
Our model consists of $m$ queues, each with a buffer of size $B$. 
The size of a packet is unit, 
which means that  each buffer can store up to $B$ packets simultaneously. 
Each packet is associated with one of $m$ values $\alpha_{i}$ ($1\leq i\leq m$), 
which represents the priority of this packet where a packet with larger value is of higher priority.
Without loss of generality, 
we assume that $\alpha_{1} = 1$, $\alpha_{m} = \alpha$, 
and $\alpha_{1} \leq \alpha_{2} \leq \cdots \leq \alpha_{m}$. 
The $i$th queue is denoted $Q^{(i)}$ and is also associated with its priority value $\alpha_{i}$. 
An arriving packet with the value $\alpha_{i}$ is stored in $Q^{(i)}$.

An input for this model is a sequence of {\em events}. 
Each event is an {\em arrival event} or a {\em scheduling event}. 
At an arrival event, 
a packet arrives at one of $m$ queues, and the packet is {\em accepted} to the buffer when the corresponding queue has free space. 
Otherwise, it is {\em rejected}. 
If a packet is accepted, 
it is stored at the tail of the corresponding queue. 
At a scheduling event, 
an online algorithm selects one non-empty queue and transmits the packet at the head of the selected queue.
We assume that any input contains enough scheduling events to transmit all the arriving packets in it. 
That is, any algorithm can certainly transmit a packet stored in its queue. 
Note that this assumption is common in the buffer management problem. 
(See e.g.~\cite{MG10}.)
The {\em gain} of an algorithm is the sum of the values of transmitted packets. 
Our goal is to maximize it. 
The gain of an algorithm $ALG$ for an input $\sigma$ is denoted by $V_{ALG}(\sigma)$. 
If $V_{ALG}(\sigma) \geq V_{OPT}(\sigma)/c$ for an arbitrary input $\sigma$, we say that $ALG$ is
$c$-{\em competitive}, where $OPT$ is an optimal offline algorithm for $\sigma$.
%
%
%

%
\section{Analysis of Priority Queuing}\label{sec:3}
\subsection{Priority Queuing}\label{sec:3.1}
\ifnum \count10 > 0
$PQ$は、グリーディなアルゴリズムである。
すなわち、
パケットを保持する最もpriorityの高いキューを選択し、その先頭のパケットをscheduleする。
tie breakの方法は、最も大きい番号のキューを選択するとする。
なお、解析を簡単にするため、
$OPT$はpacketをrejectしないと仮定する。
この仮定は競合比の解析に影響を与えない。
（Appendix~\ref{ap.sec:1}のLemma~\ref{LMA:a.1}）
\fi
\ifnum \count11 > 0
\com{（■英語）}
$PQ$ is a greedy algorithm. 
At a scheduling event,
$PQ$ selects the non-empty queue with the largest index. 
For analysis, we assume that $OPT$ does not reject an
arriving packet. 
This assumption does not affect the analysis of the competitive ratio.
(See Lemma~\ref{LMA:a.1} in Appendix~\ref{ap.sec:1}.)
\fi
\subsection{Overview of the Analysis}\label{sec:3.2.1}
\ifnum \count10 > 0
\com{（■日本語）\\}
$OPT$のみがacceptするパケットを{\em extra packet}と呼ぶ。
$PQ$の競合比を評価するためにextra packetの価値の総和を評価する。
そこで、解析のために幾つか定義を与える。
任意の入力$\sigma$に対して、
${k}_{j}(\sigma)$を, $Q^{(j)}$に到着するextra packetの数とする。
また、
extra packetが1つ以上arriveするqueueを{\em good queue}と呼び、
任意の入力$\sigma$に対して、
$n(\sigma)$は、good queueの数を表す。
更に、任意の入力$\sigma$、任意の$i (\in [1, n(\sigma)])$に対して、
$q_{i}(\sigma)$は、$i$番目に小さいgood queueの番号を表す。
すなわち、$1 \leq q_{1}(\sigma) < q_{2}(\sigma) < \cdots < q_{n(\sigma)}(\sigma) \leq m$が成立するとする。
また、
$q_{n(\sigma)+1}(\sigma) = m$と定義する。
更に、
任意の入力$\sigma$に対して、
$s_{j}(\sigma)$は$PQ$が$Q^{(j)}$からtransmitするpacketの数を表す。
ただし、
これらの変数では、
文脈から明らかな場合は$\sigma$を省略する。
このとき、
${V}_{PQ}(\sigma) = \sum_{j = 1}^{m} \alpha_{j} s_{j}$
かつ
${V}_{OPT}(\sigma) = {V}_{PQ}(\sigma) + \sum_{i = 1}^{n} \alpha_{q_{i}} k_{q_{i}}$
が成立する。
補題~\ref{LMA:3.2.1}において、
$k_{m} = 0$、
すなわち、
$q_{n} + 1 \leq m$、
が成立することを示す。
（■）
ある入力集合${\cal S}^{*}$（定義は以下で与える）を、
補題~\ref{LMA:3.2.3}から補題~\ref{LMA:3.2.8}にかけて、
幾つかの$PQ$に対する敵対者の戦略を用いて逐次的に構成する。
さらに、
補題~\ref{LMA:3.2.9}において、
$\frac{V_{OPT}(\sigma)}{V_{PQ}(\sigma)}$を最大にする様な入力$\sigma$が
ある性質を満たす入力集合${\cal S}^{*}$に含まれていることを示す。
すなわち、
$PQ$の競合比を得ることができる入力$\sigma^{*}$が${\cal S}^{*}$に含まれていることを示す。
正確に言うと、
（■）
入力集合${\cal S}^{*}$を次の5つの条件をみたす入力$\sigma'$の集合とする。
(i) 
任意の$i (\in [1, n(\sigma')-1])$に対して、$q_{i}(\sigma') + 1 = q_{i+1}(\sigma')$が成立する、
(ii) 
任意の$i (\in [1, n(\sigma')])$に対して、$k_{q_{i}(\sigma')}(\sigma') = B$が成立する、
(iii) 
任意の$j (\in [q_{1}(\sigma'), q_{n(\sigma')}(\sigma')+1])$に対して、$s_{j}(\sigma') = B$が成立する、
(iv) 
$q_{1}(\sigma') -1 \geq 1$が成立するならば、
任意の$j (\in [1, q_{1}(\sigma')-1])$に対して、$s_{j}(\sigma') = 0$が成立する、
(v) 
$q_{n(\sigma')}(\sigma') + 2 \leq m$が成立するならば、
任意の$j (\in [q_{n(\sigma')}(\sigma')+2, m])$に対して、$s_{j}(\sigma') = 0$が成立する。
このとき、
補題~\ref{LMA:3.2.9}において、
ある入力$\sigma^{*} \in {\cal S}^{*}$が存在して、
$\max_{\sigma''} \{ \frac{V_{OPT}(\sigma'')}{V_{PQ}(\sigma'')} \} = \frac{V_{OPT}(\sigma^{*})}{V_{PQ}(\sigma^{*})}$
が成立することを示す。
上記の事実から
次の様に$PQ$の競合比を得ることができる。
簡単のため、
任意の$i (\in [1, m])$に対して、
$s_{i}(\sigma^{*})$を$s_{i}^{*}$と表記し、
$n(\sigma^{*})$を$n^{*}$と表記し、
任意の$i (\in [1, n^{*}])$に対して、
$q_{i}(\sigma^{*})$を$q_{i}^{*}$と表記する。
任意の$j (\in [1, m])$に対して、
$k_{i}(\sigma^{*})$を$k_{i}^{*}$と表記する。
よって、
$
\frac{V_{OPT}(\sigma^{*})}{V_{PQ}(\sigma^{*})} 
	= \frac{ V_{PQ}(\sigma^{*}) + \sum_{i = 1}^{n^{*}} \alpha_{q_{i}^{*}} k_{q_{i}^{*} }^{*} }{ V_{PQ}(\sigma^{*}) }
	= 1 + \frac{ B \sum_{ j = q_{1}^{*} }^{ q_{n^{*}}^{*} } \alpha_{j} }{ B \sum_{j = q_{1}^{*} }^{ q_{n^{*}}^{*}+1 } \alpha_{j} }
	\leq 1 + \frac{ \sum_{j = 1}^{q_{n^{*}}^{*}} \alpha_{j} }{ \sum_{j = 1}^{q_{n^{*}}^{*}+1} \alpha_{j} }
	= 2 - \frac{ \alpha_{q_{n^{*}+1}} }{ \sum_{j = 1}^{q_{n^{*}}+1} \alpha_{j} }
$
が成立する。
最後の不等式は、任意の$x \geq 2$と$y \geq x$に対して、
$
\frac{ \sum_{j = x-1}^{y} \alpha_{j} }{ \sum_{j = x-1}^{y+1} \alpha_{j} } - \frac{ \sum_{j = x}^{y} \alpha_{j} }{ \sum_{j = x}^{y+1} \alpha_{j} }
	= ( \sum_{j = x-1}^{y} \alpha_{j} \sum_{j = x}^{y+1} \alpha_{j} - \sum_{j = x}^{y} \alpha_{j} \sum_{j = x-1}^{y+1} \alpha_{j} ) / (\sum_{j = x-1}^{y+1} \alpha_{j} \sum_{j = x}^{y+1} \alpha_{j} )
	= ( \alpha_{x-1} \alpha_{y+1} ) / ( \sum_{j = x-1}^{y+1} \alpha_{j} \sum_{j = x}^{y+1} \alpha_{j} )
	> 0
$
が成立することから導ける。
これは$PQ$の競合比の上限を示している。
一方、
補題~\ref{LMA:lower}において、
入力$\hat{\sigma}$が存在して、
$
\frac{V_{OPT}(\hat{\sigma})}{V_{PQ}(\hat{\sigma})} 
	= 2 - \min_{x \in [1, m-1] } \{ \frac{ \alpha_{x+1} }{ \sum_{j = 1}^{x+1} \alpha_{j} } \}
$
が成立することを示す。
これは$PQ$の競合比の下限を示している。
それゆえ、以下の定理が成立する: 
\fi
\ifnum \count11 > 0
\com{（■英語）}
We define an {\em extra packet} as a packet which is accepted by $OPT$ but rejected by $PQ$. 
In the following analysis, we evaluate the sum of the values of extra packets to obtain the competitive ratio of $PQ$.
We introduce some notation for our analysis. 
For any input $\sigma$, 
${k}_{j}(\sigma)$ denotes the number of extra packets arriving at $Q^{(j)}$ when treating $\sigma$. 
We call a queue at which at least one extra packet arrives a {\em good queue} when treating $\sigma$. 
$n(\sigma)$ denotes the number of good queues for $\sigma$. 
Moreover, for any input $\sigma$ and any $i (\in [1, n(\sigma)])$, 
$q_{i}(\sigma)$ denotes the good queue with the $i$th minimum index. 
That is, 
$1 \leq q_{1}(\sigma) < q_{2}(\sigma) < \cdots < q_{n(\sigma)}(\sigma) \leq m$. 
Also, we define $q_{n(\sigma)+1}(\sigma) = m$. 
In addition, 
for any input $\sigma$, 
$s_{j}(\sigma)$ denotes the number of packets which $PQ$ transmits from $Q^{(j)}$. 
We drop the input $\sigma$ from the notation when it is clear. 
Then, 
${V}_{PQ}(\sigma) = \sum_{j = 1}^{m} \alpha_{j} s_{j}$, and 
${V}_{OPT}(\sigma) = {V}_{PQ}(\sigma) + \sum_{i = 1}^{n} \alpha_{q_{i}} k_{q_{i}}$. 
(The equality follows from Lemma~\ref{LMA:a.1}.)
First, 
we show that $k_{m} = 0$, 
that is, $q_{n} + 1 \leq m$, in Lemma~\ref{LMA:3.2.1}. 
We will gradually construct some input set ${\cal S}^{*}$ (defined below) from Lemma~\ref{LMA:3.2.3} to Lemma~\ref{LMA:3.2.8} using some adversarial strategies against $PQ$. 
Moreover, 
in Lemma~\ref{LMA:3.2.9}, 
we prove that the set ${\cal S}^{*}$ includes an input $\sigma$ such that the ratio $\frac{V_{OPT}(\sigma)}{V_{PQ}(\sigma)}$ is maximized. 
That is, 
we show that there exists an input $\sigma^{*}$ in the set ${\cal S}^{*}$ to get the competitive ratio of $PQ$ in the lemma. 
More formally, 
we define the set ${\cal S}^{*}$ of the inputs $\sigma'$ satisfying the following five conditions: 
(i) 
for any $i (\in [1, n(\sigma')-1])$, $q_{i}(\sigma') + 1 = q_{i+1}(\sigma')$, 
(ii) 
for any $i (\in [1, n(\sigma')])$, $k_{q_{i}(\sigma')}(\sigma') = B$, 
(iii) 
for any $j (\in [q_{1}(\sigma'), q_{n(\sigma')}(\sigma')+1])$, $s_{j}(\sigma') = B$, 
(iv) 
for any $j (\in [1, q_{1}(\sigma')-1])$, $s_{j}(\sigma') = 0$ if $q_{1}(\sigma') -1 \geq 1$, and 
(v) 
for any $j (\in [q_{n(\sigma')}(\sigma')+2, m])$, $s_{j}(\sigma') = 0$ 
if $q_{n(\sigma')}(\sigma') + 2 \leq m$. 
Then, 
we show that there exists an input $\sigma^{*} \in {\cal S}^{*}$ such that 
$\max_{\sigma''} \{ \frac{V_{OPT}(\sigma'')}{V_{PQ}(\sigma'')} \} = \frac{V_{OPT}(\sigma^{*})}{V_{PQ}(\sigma^{*})}$ 
in Lemma~\ref{LMA:3.2.9}. 
By the above lemmas, 
we can obtain the competitive ratio of $PQ$ as follows: 
For ease of presentation, 
we write $s_{i}(\sigma^{*})$, $n(\sigma^{*})$, $q_{i}(\sigma^{*})$ and $k_{i}(\sigma^{*})$ as $s_{i}^{*}$, $n^{*}$, $q_{i}^{*}$ and $k_{i}^{*}$, respectively. 
Thus, 
$
\frac{V_{OPT}(\sigma^{*})}{V_{PQ}(\sigma^{*})} 
	= \frac{ V_{PQ}(\sigma^{*}) + \sum_{i = 1}^{n^{*}} \alpha_{q_{i}^{*}} k_{q_{i}^{*} }^{*} }{ V_{PQ}(\sigma^{*}) }
	= 1 + \frac{ B \sum_{ j = q_{1}^{*} }^{ q_{n^{*}}^{*} } \alpha_{j} }{ B \sum_{j = q_{1}^{*} }^{ q_{n^{*}}^{*}+1 } \alpha_{j} }
	\leq 1 + \frac{ \sum_{j = 1}^{q_{n^{*}}^{*}} \alpha_{j} }{ \sum_{j = 1}^{q_{n^{*}}^{*}+1} \alpha_{j} }
	= 2 - \frac{ \alpha_{q_{n^{*}+1}} }{ \sum_{j = 1}^{q_{n^{*}}+1} \alpha_{j} }
$. 
The last inequality follows from 
$
\frac{ \sum_{j = x-1}^{y} \alpha_{j} }{ \sum_{j = x-1}^{y+1} \alpha_{j} } - \frac{ \sum_{j = x}^{y} \alpha_{j} }{ \sum_{j = x}^{y+1} \alpha_{j} }
	= ( \sum_{j = x-1}^{y} \alpha_{j} \sum_{j = x}^{y+1} \alpha_{j} - \sum_{j = x}^{y} \alpha_{j} \sum_{j = x-1}^{y+1} \alpha_{j} ) / (\sum_{j = x-1}^{y+1} \alpha_{j} \sum_{j = x}^{y+1} \alpha_{j} )
	= ( \alpha_{x-1} \alpha_{y+1} ) / ( \sum_{j = x-1}^{y+1} \alpha_{j} \sum_{j = x}^{y+1} \alpha_{j} )
	> 0
$. 
This gives an upper bound on the competitive ratio of $PQ$. 
On the other hand, 
we show that 
there exists some input $\hat{\sigma}$ such that 
$
\frac{V_{OPT}(\hat{\sigma})}{V_{PQ}(\hat{\sigma})} 
	= 2 - \min_{x \in [1, m-1] } \{ \frac{ \alpha_{x+1} }{ \sum_{j = 1}^{x+1} \alpha_{j} } \}
$
in Lemma~\ref{LMA:lower}, 
which presents a lower bound for $PQ$. 
Therefore, 
we have the following theorem:
\fi

\begin{theorem} \label{thm:1}
	\ifnum \count10 > 0
	$PQ$の競合比は
	$2 - \min_{x \in [1, m-1] } \{ \frac{ \alpha_{x+1} }{ \sum_{j = 1}^{x+1} \alpha_{j} } \}$
	が成立する。
	\fi
	\ifnum \count11 > 0
	The competitive ratio of $PQ$ is exactly 
	$2 - \min_{x \in [1, m-1] } \{ \frac{ \alpha_{x+1} }{ \sum_{j = 1}^{x+1} \alpha_{j} } \}$. 
	\fi
\end{theorem}
\subsection{Competitive Analysis of $PQ$} \label{sec:3.2.2}
\ifnum \count10 > 0
\com{（■日本語）\\}
定義を与える。
For ease of presentation, 
an {\em event time} denotes a moment when an event happens, and 
and any other moment is called a {\em non-event time}. 
We assign index numbers $1$ through $B$ to each position of a queue from
the head to the tail in increasing order. 
The $j$th position of $Q^{(i)}$ is called the $j$th {\em cell}.
任意のnon-event time $t$に対して、
時刻$t$に、
$PQ$は$Q^{(i)}$の$j$番目のcellにpacketを保持しているが、
$OPT$は$Q^{(i)}$の$j$番目のcell $c$にpacketを保持していないとする。
このとき、$c$を{\em free} cellと呼ぶ。
extra packetは必ずfree cellにおいて受理されることに注意せよ。
任意のnon-event time $t$に対して、
アルゴリズム$ALG$が$t$において、$Q^{(j)}$に保持するパケットの数を$h_{ALG}^{(j)}(t)$であらわすことにする。
以下の補題を最初に示す。
\fi
\ifnum \count11 > 0
\com{（■英語）}
We give some definitions. 
For ease of presentation, 
an {\em event time} denotes a time when an event happens, 
and any other moment is called a {\em non-event time}. 
We assign index numbers $1$ through $B$ to each position of a queue from
the head to the tail in increasing order. 
The $j$th position of $Q^{(i)}$ is called the $j$th {\em cell}.
For any non-event time $t$, 
suppose that the $j$th cell in $Q^{(i)}$ of $PQ$ holds a packet at $t$ but 
the $j$th cell $c$ in $Q^{(i)}$ of $OPT$ does not at $t$. 
Then, we call $c$ a {\em free} cell at $t$. 
Note that any extra packet is accepted at a free cell.
For any non-event time $t$, 
let $h_{ALG}^{(j)}(t)$ denote the number of packets which an algorithm $ALG$ stores in $Q^{(j)}$ at $t$. 
We first prove the following lemma.
(The lemma is similar to Lemma~2.3 in \cite{KA11}.)
\fi
%
\ifnum \count14 > 0
%
Due to page limitations, 
we omit most of the proofs of the following lemmas. 
They are included in Appendix~\ref{ap.sec:3}.
(The full version of this paper is available at 
{\tt \footnotesize http://www-lsm.naist.jp/\textasciitilde jkawahara/paper/egress.pdf} .)
%
\fi
%

\begin{LMA} \label{LMA:3.2.1}
	\ifnum \count10 > 0
	${k}_{m} = 0$
	が成立する。
	\fi
	\ifnum \count11 > 0
	\com{（■英語）}
	${k}_{m} = 0$. 
	\fi
\end{LMA}
%
\ifnum \count14 < 1
%
\begin{proof}
	\ifnum \count10 > 0
	（■日本語）
	\com{（■英語と内容が違うんです）}\\
	$PQ$はその定義より、
	常に優先度の高いキューを選択する。
	よって、
	任意のnon-event time $t$に、
	$h_{PQ}^{(m)}(t) \leq h_{OPT}^{(m)}(t)$
	が成立する。
	また、
	Lemma~\ref{LMA:a.1}より、
	$OPT$はパケットをrejectしないので、
	$OPT$は$Q^{(m)}$にarriveするパケットを全て受理する。
	よって、
	$PQ$も$Q^{(m)}$にarriveするパケットを全て受理することが出来る。
	よって、
	${k}_{m} = 0$
	が成立する。
	\fi
	\ifnum \count11 > 0
	\com{（■英語）}
	By the definition of $PQ$, $PQ$ selects the non-empty queue with the	highest priority. 
	Thus, $h_{PQ}^{(m)}(t) \leq h_{OPT}^{(m)}(t)$ holds at any non-event time $t$. 
	Therefore, there is no free cell in $Q^{(m)}$ of $OPT$ at any time. 
	Since any extra packet is accepted to a free cell, 
	${k}_{m} = 0$.
	\fi
\end{proof}
%
\fi
%

%
%
\ifnum \count10 > 0
\com{（■日本語）\\}
（■未修正）
以下では、
各$Q^{(q_i)} \hspace{1mm} (i \in [1, n])$において受理されるextra packetの数を評価するために、
後で定義するmatching rourineに従って各extra packetとある$PQ$のpacketの間にmatchingを構成する。
ただし、
extra packet $p$と$PQ$のpacket $p'$がmatchされている場合、
$i < i'$が成立する様にmatchingを構成する。
なお、$p$は$Q^{(i)}$から、$p'$は$Q^{(i')}$からscheduleされるものとする。
routineにおいてmatchingを構成する方針について説明する。
時刻の経過とともにmatchingを構成する。
しかし、
extra packetと$PQ$がscheduleするpacketを直接マッチさせることは難しいので、
2段階に分けてマッチングを構成する。
すなわち、
まず任意のfree cell $c$に対して、
$c$がfreeになったevent timeに$PQ$のあるpacket $p$をmatchさせ、
$c$にextra packet $p'$がacceptされたevent timeに、
$p$と$p'$をmatchさせる。
この様なmatchのやり方を実現するために、
matching routineを導入する前に、
まず各eventにおけるfree cellの変化と数の増減を概観する。
そのために幾つか定義を与える。
任意のevent time $t$に対して、
$t-$は、
$t$より前であり、かつ1つ前のイベントが発生した後のnon-event timeを表す。
同様に、
$t+$は、
$t$より後であり、かつ1つ後のイベントが発生する前のnon-event timeを表すとする。
任意のnon-event time $t$に対して、
$t$における$Q^{(j)}$におけるfree cellの数を
${f}^{(j)}(t)$であらわすことにする。
すなわち、
${f}^{(j)}(t) = \max\{ h_{PQ}^{(j)}(t) - h_{OPT}^{(j)}(t), 0 \}$
が成立する。
ただし、Lemma~\ref{LMA:a.1}より、
$OPT$がpacketは非受理しないことに注意せよ。
よって、任意のnon-event time $t$に対して、
$\sum_{j = 1}^{m} h_{PQ}^{(j)}(t) > 0$. 
が成立するならば、
$\sum_{j = 1}^{m} h_{OPT}^{(j)}(t) > 0$することに注意せよ。
\fi
\ifnum \count11 > 0
\com{（■英語）}
Next, 
in order to evaluate the total number of extra packets accepted at each $Q^{(q_i)} \hspace{1mm} (i \in [1, n])$, 
we construct some matching between extra packets and $PQ$'s packets according to the matching routine defined later. 
(Note that evaluating the number of extra packets is related to the property (ii) of ${\cal S}^{*}$.) 
Suppose that extra packet $p$ is matched with $PQ$'s packet $p'$ such that $p$ and $p'$ are transmitted from $Q^{(i)}$ and $Q^{(i')}$, respectively. 
Then, the routine constructs this matching where $i < i'$. 
Let us explain how to construct the matching. 
We match extra packet one by one with time. 
However, 
it is difficult to match an extra packet with $PQ$'s packet in a direct way. 
Thus, the matching is formed in two stages. 
That is, 
at first, for any free cell $c$, 
we match $c$ with some $PQ$'s packet $p$ when $c$ becomes free at an event time. 
At a later time, 
we rematch the extra packet $p'$ accepted into $c$ with $p$ at an event time when $OPT$ accepts $p'$. 
In order to realize such matching, 
we first verify a change in the number of free cells at each event 
before introducing our matching routine. 
We give some definitions for that reason. 
For any event time $t$, 
$t-$ denotes the non-event time before $t$ and after the previous event time. 
Also, 
$t+$ denotes the non-event time after $t$ and before the next event time. 
Let ${f}^{(j)}(t)$ denote the number of free cells in $Q^{(j)}$ at a non-event time $t$,
that is, ${f}^{(j)}(t) = \max\{ h_{PQ}^{(j)}(t) - h_{OPT}^{(j)}(t), 0 \}$.
Note that $OPT$ does not reject any packet by our assumption (Lemma~\ref{LMA:a.1} in Appendix~\ref{ap.sec:1}).
Thus, for any non-event time $t$, 
$\sum_{j = 1}^{m} h_{OPT}^{(j)}(t) > 0$ if $\sum_{j = 1}^{m} h_{PQ}^{(j)}(t) > 0$. 
\fi
\ifnum \count10 > 0
\com{（■日本語)\\}
\noindent
{\bf Arrival event: }\\
event time $t$に、packet $p$が$Q^{(x)}$に到着するとする。
\noindent
{\bf Case A1:$PQ$と$OPT$が$p$をacceptし、
$h_{PQ}^{(x)}(t-) - h_{OPT}^{(x)}(t-) > 0$が成立する場合: }
$h_{PQ}^{(x)}(t+) = h_{PQ}^{(x)}(t-) + 1$、
$h_{OPT}^{(x)}(t+) = h_{OPT}^{(x)}(t-) + 1$より、
$h_{PQ}^{(x)}(t+) - h_{OPT}^{(x)}(t+) > 0$
が成立するので、
$Q^{(x)}$の$h_{OPT}^{(x)}(t-) + 1$の位置のcellに代わって、
$Q^{(x)}$の$h_{PQ}^{(x)}(t-) + 1$の位置のcellがfree cellとなる。
また、
${f}^{(x)}(t+) = {f}^{(x)}(t-)$
が成立する。
\noindent
{\bf Case A2:$PQ$と$OPT$が$p$をacceptし、
$h_{PQ}^{(x)}(t-) - h_{OPT}^{(x)}(t-) \leq 0$が成立する場合: }
$h_{PQ}^{(x)}(t+) = h_{PQ}^{(x)}(t-) + 1$、
$h_{OPT}^{(x)}(t+) = h_{OPT}^{(x)}(t-) + 1$、
$h_{PQ}^{(x)}(t+) - h_{OPT}^{(x)}(t+) \leq 0$
が成立する。
よって、
全てのfree cellはeventの前後で変化しないので、
${f}^{(x)}(t+) = {f}^{(x)}(t-)$
が成立する。
\noindent
{\bf Case A3:$PQ$が$p$をreject, $OPT$が$p$をacceptする場合: }
$OPT$のみが$p$を受理するので、$p$はextra packetである。
$Q^{(x)}$の$h_{OPT}^{(x)}(t-) + 1$の位置のfree cellに$p$が受理され、
$h_{PQ}^{(x)}(t+) = h_{PQ}^{(x)}(t-) = B$、
$h_{OPT}^{(x)}(t+) = h_{OPT}^{(x)}(t-) + 1$
が成立し
${f}^{(x)}(t+) = {f}^{(x)}(t-) - 1$
が成立する。

\vspace{3mm}
\noindent
{\bf Scheduling event: }\\
$PQ$($OPT$, respectively)のバッファ内にpacketがある場合、
$PQ$($OPT$, respectively)が$Q^{(y)}$($Q^{(z)}$, respectively)からpacketをtransmitするとする。
\noindent
{\bf Case S: 
	$\sum_{j = 1}^{m} h_{PQ}^{(j)}(t-) > 0$かつ$\sum_{j = 1}^{m} h_{OPT}^{(j)}(t-) > 0$: }
\noindent
\hspace{2mm}
{\bf\boldmath Case S1: $y = z$の場合:}
\hspace{4mm}
\noindent
\hspace{4mm}
{\bf\boldmath Case S1.1: $h_{PQ}^{(y)}(t-) - h_{OPT}^{(y)}(t-) > 0$の場合:}\\
\hspace{6mm}
	$h_{PQ}^{(y)}(t+) = h_{PQ}^{(y)}(t-) - 1$、
	$h_{OPT}^{(y)}(t+) = h_{OPT}^{(y)}(t-) - 1$
	より、
	$h_{PQ}^{(y)}(t+) - h_{OPT}^{(y)}(t+) > 0$
	が成立するので、
	$Q^{(y)}$の$h_{PQ}^{(y)}(t-)$の位置のcellに代わって、
	$Q^{(y)}$の$h_{OPT}^{(y)}(t-)$の位置のcellがfree cellとなる。
	また、
	${f}^{(y)}(t+) = {f}^{(y)}(t-)$
	が成立する。
\noindent
\hspace{4mm}
{\bf\boldmath Case S1.2: $h_{PQ}^{(y)}(t-) - h_{OPT}^{(y)}(t-) \leq 0$の場合:}\\
\hspace{6mm}
	$h_{PQ}^{(y)}(t+) = h_{PQ}^{(y)}(t-) - 1$、
	$h_{OPT}^{(y)}(t+) = h_{OPT}^{(y)}(t-) - 1$
	より、
	$h_{PQ}^{(y)}(t+) - h_{OPT}^{(y)}(t+) \leq 0$
	が成立するので、
	free cellはeventの前後で変化しない。
\noindent
\hspace{2mm}
{\bf\boldmath Case S2: $y > z$の場合:}
\hspace{4mm}
\noindent
\hspace{4mm}
{\bf\boldmath Case S2.1: $h_{PQ}^{(z)}(t-) - h_{OPT}^{(z)}(t-) < 0$の場合:}\\
\hspace{6mm}
	$h_{PQ}^{(z)}(t+) = h_{PQ}^{(z)}(t-)$
	と
	$h_{OPT}^{(z)}(t+) = h_{OPT}^{(z)}(t-) - 1$より、
	$h_{PQ}^{(z)}(t+) \leq h_{OPT}^{(z)}(t+)$
	が成立するので、
	$Q^{(z)}$のfree cellはeventの前後で変化しない。
\noindent
\hspace{6mm}
{\bf\boldmath Case S2.1.1: $h_{PQ}^{(y)}(t-) - h_{OPT}^{(y)}(t-) > 0$の場合:}\\
\hspace{8mm}
	$h_{PQ}^{(y)}(t+) = h_{PQ}^{(y)}(t-) - 1$
	かつ
	$h_{OPT}^{(y)}(t+) = h_{OPT}^{(y)}(t-)$
	が成立するので、
	${f}^{(y)}(t+) = {f}^{(y)}(t-) - 1$
	が成立し、すなわち、
	$Q^{(y)}$のfree cellの数は減少する。
\noindent
\hspace{6mm}
	{\bf\boldmath Case S2.1.2: $h_{PQ}^{(y)}(t-) - h_{OPT}^{(y)}(t-) \leq 0$の場合:}\\
\hspace{8mm}
	$h_{PQ}^{(y)}(t+) = h_{PQ}^{(y)}(t-) - 1$
	かつ
	$h_{OPT}^{(y)}(t+) = h_{OPT}^{(y)}(t-)$
	が成立するので、
	$h_{PQ}^{(y)}(t+) < h_{OPT}^{(z)}(t+)$
	が成立するので、
	$Q^{(y)}$のfree cellはeventの前後で変化しない。
\noindent
\hspace{4mm}
{\bf\boldmath Case S2.2: $h_{PQ}^{(z)}(t-) - h_{OPT}^{(z)}(t-) \geq 0$の場合:}\\
\hspace{6mm}
$h_{PQ}^{(z)}(t+) = h_{PQ}^{(z)}(t-)$,  
$h_{OPT}^{(z)}(t+) = h_{OPT}^{(z)}(t-) - 1$
が成立する。
ゆえに、
$Q^{(z)}$の$h_{OPT}^{(z)}(t-)$の位置のcellがfree cellとなり、
${f}^{(z)}(t+) = {f}^{(z)}(t-) + 1$
が成立する。
\noindent
\hspace{6mm}
	{\bf\boldmath Case S2.2.1: $h_{PQ}^{(y)}(t-) - h_{OPT}^{(y)}(t-) > 0$の場合: }\\
\hspace{8mm}
	$h_{PQ}^{(y)}(t+) = h_{PQ}^{(y)}(t-) - 1$
	かつ
	$h_{OPT}^{(y)}(t+) = h_{OPT}^{(y)}(t-)$
	より、
	${f}^{(y)}(t+) = {f}^{(y)}(t-) - 1$
	が成立する。
\noindent
\hspace{6mm}
{\bf\boldmath Case S2.2.2: $h_{PQ}^{(y)}(t-) - h_{OPT}^{(y)}(t-) \leq 0$の場合:}\\
\hspace{8mm}
	$h_{PQ}^{(y)}(t+) = h_{PQ}^{(y)}(t-) - 1$
	かつ
	$h_{OPT}^{(y)}(t+) = h_{OPT}^{(y)}(t-)$
	より、
	$h_{PQ}^{(y)}(t+) < h_{OPT}^{(z)}(t+)$
	が成立するので、
	$Q^{(y)}$のfree cellはeventの前後で変化しない。
\noindent
\hspace{2mm}
{\bf\boldmath Case S3: $y < z$の場合:}\\
\hspace{4mm}
	$PQ$の定義より、
	$h_{PQ}^{(z)}(t+) = h_{PQ}^{(z)}(t-) = 0$が成立するので、
	$Q^{(z)}$に新しいfree cellは生じない。
\noindent
\hspace{4mm}
{\bf\boldmath Case S3.1: $h_{PQ}^{(y)}(t-) - h_{OPT}^{(y)}(t-) > 0$の場合:}\\
\hspace{6mm}
	$h_{PQ}^{(y)}(t+) = h_{PQ}^{(y)}(t-) - 1$
	かつ
	$h_{OPT}^{(y)}(t+) = h_{OPT}^{(y)}(t-)$
	より、
	${f}^{(y)}(t+) = {f}^{(y)}(t-) - 1$
	が成立する。
\noindent
\hspace{4mm}
{\bf\boldmath Case S3.2: $h_{PQ}^{(y)}(t-) - h_{OPT}^{(y)}(t-) \leq 0$の場合:}\\
\hspace{6mm}
	$h_{PQ}^{(y)}(t+) = h_{PQ}^{(y)}(t-) - 1$
	かつ
	$h_{OPT}^{(y)}(t+) = h_{OPT}^{(y)}(t-)$
	より、
	$h_{PQ}^{(y)}(t+) < h_{OPT}^{(z)}(t+)$
	が成立するので、
	$Q^{(y)}$のfree cellはeventの前後で変化しない。
\noindent
{\bf Case \={S}: 
$\sum_{j = 1}^{m} h_{PQ}^{(j)}(t-) = 0$かつ$\sum_{j = 1}^{m} h_{OPT}^{(j)}(t-) > 0$の場合: }
\hspace{2mm}
	$PQ$のバッファは空なので、
	全てのキューにfree cellは存在しない。

\fi
\ifnum \count11 > 0
\com{（■英語）\\}
\vspace{3mm}
\noindent
{\bf Arrival event:} 
Let $p$ be the packet arriving at $Q^{(x)}$ at an event time $t$. 
\noindent
{\bf\boldmath Case A1: Both $PQ$ and $OPT$ accept $p$, and $h_{PQ}^{(x)}(t-) - h_{OPT}^{(x)}(t-) > 0$:} 
Since $h_{PQ}^{(x)}(t+) = h_{PQ}^{(x)}(t-) + 1$ and $h_{OPT}^{(x)}(t+) = h_{OPT}^{(x)}(t-) + 1$,
$h_{PQ}^{(x)}(t+) - h_{OPT}^{(x)}(t+) > 0$.  Thus, the
$(h_{PQ}^{(x)}(t-) + 1)$st cell of $Q^{(x)}$ becomes free in place of the $(h_{OPT}^{(x)}(t-) + 1)$st cell of $Q^{(x)}$. 
Hence ${f}^{(x)}(t+) = {f}^{(x)}(t-)$.
\noindent 
{\bf\boldmath Case A2: Both $PQ$ and $OPT$ accept $p$, and $h_{PQ}^{(x)}(t-) - h_{OPT}^{(x)}(t-) \leq 0$:}
Since $h_{PQ}^{(x)}(t+)
= h_{PQ}^{(x)}(t-) + 1$ and $h_{OPT}^{(x)}(t+) = h_{OPT}^{(x)}(t-) + 1$,
$h_{PQ}^{(x)}(t+) - h_{OPT}^{(x)}(t+) \leq 0$. 
Since the states of all the free cells do not change before and after $t$, 
${f}^{(x)}(t+) = {f}^{(x)}(t-)$. 
\noindent 
{\bf\boldmath Case A3: $PQ$ rejects $p$, but $OPT$ accepts $p$:}
$p$ is an extra packet since only $OPT$ accepts $p$. 
$p$ is accepted into the $(h_{OPT}^{(x)}(t-) + 1)$st cell, which is free at $t-$, of $Q^{(x)}$. 
$h_{PQ}^{(x)}(t+) = h_{PQ}^{(x)}(t-) = B$, and $h_{OPT}^{(x)}(t+) = h_{OPT}^{(x)}(t-) + 1$, 
which means that 
${f}^{(x)}(t+) = {f}^{(x)}(t-) - 1$.
\vspace{2mm}
\noindent
{\bf\boldmath Scheduling event:}\\
If $PQ$ ($OPT$, respectively) has at least one non-empty queue, 
suppose that $PQ$ ($OPT$, respectively) transmits a packet from $Q^{(y)}$ ($Q^{(z)}$, respectively) at $t$.  
\noindent
{\bf\boldmath Case S: 
	$\sum_{j = 1}^{m} h_{PQ}^{(j)}(t-) > 0$ and $\sum_{j = 1}^{m} h_{OPT}^{(j)}(t-) > 0$: } 
\noindent
\hspace{2mm}
{\bf\boldmath Case S1: $y = z$:}
\hspace{4mm}
\noindent
\hspace{4mm}
{\bf\boldmath Case S1.1: $h_{PQ}^{(y)}(t-) - h_{OPT}^{(y)}(t-) > 0$:}\\
\hspace{6mm}
	Since $h_{PQ}^{(y)}(t+) = h_{PQ}^{(y)}(t-) - 1$ and $h_{OPT}^{(y)}(t+) = h_{OPT}^{(y)}(t-) - 1$, 
	$h_{PQ}^{(y)}(t+) - h_{OPT}^{(y)}(t+) > 0$ holds. 
	Thus, the $h_{OPT}^{(y)}(t-)$th cell of
	$Q^{(y)}$ becomes free in place of the $h_{PQ}^{(y)}(t-)$th cell of $Q^{(y)}$. 
	Hence ${f}^{(y)}(t+) = {f}^{(y)}(t-)$.
\noindent
\hspace{4mm}
{\bf\boldmath Case S1.2: $h_{PQ}^{(y)}(t-) - h_{OPT}^{(y)}(t-) \leq 0$:}\\
\hspace{6mm}
Since $h_{PQ}^{(y)}(t+) = h_{PQ}^{(y)}(t-) - 1$ and $h_{OPT}^{(y)}(t+) = h_{OPT}^{(y)}(t-) - 1$ hold, 
$h_{PQ}^{(y)}(t+) - h_{OPT}^{(y)}(t+) \leq 0$. 
Hence the states of all the free cells do not change before and after $t$. 
\noindent
\hspace{2mm}
{\bf\boldmath Case S2: $y > z$:}
\hspace{4mm}
\noindent
\hspace{4mm}
{\bf\boldmath Case S2.1: $h_{PQ}^{(z)}(t-) - h_{OPT}^{(z)}(t-) < 0$:}\\
\hspace{6mm}
Since $h_{PQ}^{(z)}(t+) = h_{PQ}^{(z)}(t-)$ and $h_{OPT}^{(z)}(t+) = h_{OPT}^{(z)}(t-) - 1$, 
$h_{PQ}^{(z)}(t+) \leq h_{OPT}^{(z)}(t+)$. 
Thus, the states of all the free cells of $Q^{(z)}$ do not change before and after $t$.
\noindent
\hspace{6mm}
{\bf\boldmath Case S2.1.1: $h_{PQ}^{(y)}(t-) - h_{OPT}^{(y)}(t-) > 0$:}\\
\hspace{8mm}
	Since $h_{PQ}^{(y)}(t+) = h_{PQ}^{(y)}(t-) - 1$ and $h_{OPT}^{(y)}(t+) = h_{OPT}^{(y)}(t-)$, 
	${f}^{(y)}(t+) = {f}^{(y)}(t-) - 1$ holds.
\noindent
\hspace{6mm}
	{\bf\boldmath Case S2.1.2: $h_{PQ}^{(y)}(t-) - h_{OPT}^{(y)}(t-) \leq 0$:}\\
\hspace{8mm}
		Since $h_{PQ}^{(y)}(t+) = h_{PQ}^{(y)}(t-) - 1$ and $h_{OPT}^{(y)}(t+) = h_{OPT}^{(y)}(t-)$, 
		$h_{PQ}^{(y)}(t+) < h_{OPT}^{(y)}(t+)$. 
		Hence, the states of all the free cells of $Q^{(y)}$ do not change before and after $t$.
\noindent
\hspace{4mm}
{\bf\boldmath Case S2.2: $h_{PQ}^{(z)}(t-) - h_{OPT}^{(z)}(t-) \geq 0$:}\\
\hspace{6mm}
	$h_{PQ}^{(z)}(t+) = h_{PQ}^{(z)}(t-)$ and $h_{OPT}^{(z)}(t+) = h_{OPT}^{(z)}(t-) - 1$. 
	Thus, the $h_{OPT}^{(z)}(t-)$th cell of $Q^{(z)}$ becomes free, which means that 
	${f}^{(z)}(t+) = {f}^{(z)}(t-) + 1$ holds.
\noindent
\hspace{6mm}
	{\bf\boldmath Case S2.2.1: $h_{PQ}^{(y)}(t-) - h_{OPT}^{(y)}(t-) > 0$: }\\
\hspace{8mm}
		Since $h_{PQ}^{(y)}(t+) = h_{PQ}^{(y)}(t-) - 1$ and $h_{OPT}^{(y)}(t+) = h_{OPT}^{(y)}(t-)$,
		${f}^{(y)}(t+) = {f}^{(y)}(t-) - 1$.
\noindent
\hspace{6mm}
{\bf\boldmath Case S2.2.2: $h_{PQ}^{(y)}(t-) - h_{OPT}^{(y)}(t-) \leq 0$:}\\
\hspace{8mm}
	Since $h_{PQ}^{(y)}(t+) = h_{PQ}^{(y)}(t-) - 1$ and $h_{OPT}^{(y)}(t+) = h_{OPT}^{(y)}(t-)$, 
	$h_{PQ}^{(y)}(t+) < h_{OPT}^{(y)}(t+)$, which means that 
	the states of all the free cells of $Q^{(y)}$ do not change before and after $t$.
\noindent
\hspace{2mm}
{\bf\boldmath Case S3: $y < z$:}\\
\hspace{4mm}
	Since $h_{PQ}^{(z)}(t+) = h_{PQ}^{(z)}(t-) = 0$ by the definition of $PQ$, 
	no new free cell arises in $Q^{(z)}$. 
\noindent
\hspace{4mm}
{\bf\boldmath Case S3.1: $h_{PQ}^{(y)}(t-) - h_{OPT}^{(y)}(t-) > 0$:}\\
\hspace{6mm}
	Since $h_{PQ}^{(y)}(t+) = h_{PQ}^{(y)}(t-) - 1$ and $h_{OPT}^{(y)}(t+) = h_{OPT}^{(y)}(t-)$, 
	${f}^{(y)}(t+) = {f}^{(y)}(t-) - 1$ holds.
\noindent
\hspace{4mm}
{\bf\boldmath Case S3.2: $h_{PQ}^{(y)}(t-) - h_{OPT}^{(y)}(t-) \leq 0$:}\\
\hspace{6mm}
	Since $h_{PQ}^{(y)}(t+) = h_{PQ}^{(y)}(t-) - 1$ and $h_{OPT}^{(y)}(t+) = h_{OPT}^{(y)}(t-)$, 
	$h_{PQ}^{(y)}(t+) < h_{OPT}^{(y)}(t+)$ holds. 
	Hence, 
	the states of all the free cells of $Q^{(y)}$ do not change before and after $t$. 
\noindent
{\bf Case \={S}: 
$\sum_{j = 1}^{m} h_{PQ}^{(j)}(t-) = 0$ and $\sum_{j = 1}^{m} h_{OPT}^{(j)}(t-) > 0$: }\\
\hspace{2mm}
	Since the buffer of $PQ$ is empty, 
	there does not exist any free cell in it.
\fi
\vspace{3mm}
\ifnum \count10 > 0
上のfree cellの状態の変化を踏まえて、
表~\ref{tab:rotuine_j}のmatching routineにしたがって、
各extra packetに$PQ$がtransmitするpacketをmatchする。
ただし、ルーチン内のcaseの名前は、上記のfree cellの変化の場合分けの名前に対応している。
matching routineの概要を説明する。
直感的に言えば、
ルーチンは過去に構成したマッチングを保持しつつ、
free cellが新しく生じる場合（Cases A1, S1.1, S2.2）と
extra packetが$OPT$がacceptする場合（Case A3）に、
matchingを新しく構成、もしくは修正する。
それ以外の場合（Cases A2, S1.2, S2.1, S3, \={S}）は、何もしない。
具体的に言えば、
Cases A1では、
$OPT$と$PQ$が同じキューにおいてパケットを受理し、
Case S1.1では同じキューからパケットを送信する。
結果として、これらの場合、free cellの状況は変化するが、その総数は変化しないので、
matchingの相手を変更してmatchingを更新する。
Case S2は、
$OPT$と$PQ$が異なるキューからpacketを送信しており、
両者のキュー内のpacketの数などの条件（S2.2）により、
$OPT$のキューにfree cellが生じる。
そのcellに対して、$PQ$がこのeventにおいて送信するpacketをmatchする。
Case A3では、
free cell $c$にextra packet $p$を受理する。
$c$には$PQ$のpacket $p'$が過去にmatchされている（帰納的に証明される）ので、
$p'$の相手を$c$から$p$に付け替える。
extra packetが一旦、matchされると、
それ以降、そのmatchの相手が変更されることはない。
\fi
\ifnum \count11 > 0
\com{（■英語）}
\com{（■infocomの後修正。）}
Based on a change in the state of free cells, 
we match each extra packet with a packet transmitted by $PQ$ according to the matching routine in Table~\ref{tab:rotuine}. 
(All the names of the cases in the routine correspond to the names of cases in the above sketch about free cells.)
We outline the matching routine. 
Roughly speaking, 
the routine either adds a new edge to a tentative matching if a new free cell arises (Cases A1, S1.1, S2.2), 
or fixes some edge if $OPT$ accepts an extra packet (Case A3), 
while keeping edges constructed before.
In the other cases (Cases A2, S1.2, S2.1, S3, \={S}), the routine does nothing. 
Specifically, 
both $OPT$ and $PQ$ accept arriving packets at the same queue in Case A1, and 
they transmit packets from the same queue in Case S1.1. 
Since the total numbers of free cells do not change in these cases but the states of free cells do, 
the routine updates an edge in a tentative matching, 
namely removes an edge between $PQ$'s packet $p$ and a cell that became non-free and adds a new edge between $p$ and a new free cell. 
When the routine executes Case S2.2, 
the queue where $OPT$ transmits a packet is different from that of $PQ$. 
By the conditions of the numbers of packets in their queues and so on (see the condition of Case S2.2), 
a cell of $OPT$'s queue becomes free. 
The routine matches the cell with the packet transmitted by $PQ$ at this event. 
In Case A3, 
an extra packet is accepted into a free cell $c$. 
Since $c$ has been already matched with some $PQ$'s packet $p'$, which can be proven inductively in Lemma~\ref{LMA:3.2.2}, 
the routine replaces the partner of $p'$ from $c$ to $p$. 
Once an extra packet is matched, the partner of the packet never changes. 
\fi
\ifnum \count10 > 0
\begin{table*}[h]
\caption{Matching routine}
\begin{center}
\begin{tabularx}{165mm}{|X|}
	\hline
\noindent
{\bf Matching routine:}
時刻$t$をevent timeとする.\\
\noindent
{\bf Arrival event: }\\
$t$にpacket $p$が$Q^{(x)}$に到着するとする。
以下の3つのCaseのうち1つを実行する。
\noindent
{\bf\boldmath Case A1: $PQ$と$OPT$が$p$をacceptし、
$h_{PQ}^{(x)}(t-) - h_{OPT}^{(x)}(t-) > 0$が成立する場合: } \\ 
\hspace{2mm}
	時刻$t-$に$Q^{(x)}$の$h_{OPT}^{(x)}(t-) + 1$の位置のfree cellとmatchしているパケット$q$が必ず存在する。
	（補題~\ref{LMA:3.2.2}参照。）
	ここで、$q$をunmatchし、
	$q$を$h_{PQ}^{(x)}(t-) + 1$の位置の新しいfree cellとmatchさせる。
\noindent
{\bf\boldmath Case A2:$PQ$と$OPT$が$p$をacceptし、
$h_{PQ}^{(x)}(t-) - h_{OPT}^{(x)}(t-) \leq 0$が成立する場合: }\\
\hspace{2mm}
	何もしない. 
\noindent
{\bf\boldmath Case A3: $PQ$が$p$をreject, $OPT$が$p$をacceptする場合: }\\
\hspace{2mm}
	$c$を$Q^{(x)}$の$OPT$の$(h_{OPT}^{(x)}(t-) + 1)$番目のcellとする。
	すなわち、$c$にextra packetである$p$が受理される。
	$c$は$t-$において、freeであり、$t+$においてfreeでなくなることに注意せよ。
	補題~\ref{LMA:3.2.2}より、
	$t-$において$c$にmatchされているpacket $q$が存在する。
	そこで、
	$q$のmatchの相手を$c$から$p$に変更する。
\vspace{2mm}
\noindent
{\bf Scheduling event:} 
$t-$に
$PQ$($OPT$, respectively)のバッファ内にpacketがある場合、
$t$に$PQ$($OPT$, respectively)が$Q^{(y)}$($Q^{(z)}$, respectively)からpacketをtransmitするとする。
以下の3つのCaseのうち1つを実行する。
\noindent
{\bf\boldmath Case S1.1: {\footnotesize $\sum_{j = 1}^{m} h_{PQ}^{(j)}(t-) > 0$かつ$\sum_{j = 1}^{m} h_{OPT}^{(j)}(t-) > 0$かつ$y = z$かつ$h_{PQ}^{(y)}(t-) - h_{OPT}^{(y)}(t-) > 0$が成立する場合}:}\\
\hspace{2mm}
	時刻$t-$に$Q^{(y)}$の$h_{PQ}^{(y)}(t-)$の位置のfree cellとmatchしているパケット$q$が必ず存在する。
	（補題~\ref{LMA:3.2.2}参照。）
	$q$をunmatchし、
	$q$を$h_{OPT}^{(y)}(t-)$の位置の新しいfree cellとmatchさせる。
\noindent
{\bf\boldmath Case S2.2: {\footnotesize $\sum_{j = 1}^{m} h_{PQ}^{(j)}(t-) > 0$かつ$\sum_{j = 1}^{m} h_{OPT}^{(j)}(t-) > 0$かつ$y > z$かつ$h_{PQ}^{(z)}(t-) - h_{OPT}^{(z)}(t-) \geq 0$が成立する場合}:}\\
\hspace{2mm}
	$c$を$Q^{(z)}$の$h_{OPT}^{(z)}(t-)$の位置に新しく発生したfree cellとする。
	$t$において$PQ$に$Q^{(y)}$からtransmitされるpacket $p$は何にもmatchされていないので、
	（補題~\ref{LMA:3.2.2}参照。）
	$p$のpacketを$c$にmatchする。
\noindent
{\bf Otherwise (Cases S1.2, S2.1, S3, \={S}): }\\
\hspace{2mm}
	何もしない。 
\\
	\hline 
\end{tabularx}
\end{center}
	\label{tab:rotuine_j}
\end{table*}
\fi
\ifnum \count11 > 0
\com{（■英語）\\}
\begin{table*}[h]
\caption{Matching routine}
\begin{center}
\begin{tabularx}{165mm}{|X|}
	\hline
\noindent
{\bf Matching routine:}
Let $t$ be an event time. \\
\noindent
{\bf Arrival event:}
Suppose that the packet $p$ arrives at $Q^{(x)}$ at $t$. 
Execute one of the following three cases at $t$. 
\noindent
{\bf\boldmath Case A1: Both $PQ$ and $OPT$ accept $p$, and $h_{PQ}^{(x)}(t-) - h_{OPT}^{(x)}(t-) > 0$:} \\ 
\hspace{2mm}
	Let $c$ be $OPT$'s $(h_{OPT}^{(x)}(t-) + 1)$st cell of $Q^{(x)}$, 
	which is free at $t-$ but not at $t+$. 
	Let $c'$ be $OPT$'s $(h_{PQ}^{(x)}(t-) + 1)$st cell which is not free at $t-$ but is free at $t+$. 
	There exists the packet $q$ matched with $c$ at $t-$. 
	(The existence of such $q$ is guaranteed by Lemma~\ref{LMA:3.2.2}.)
	Change the matching partner of $q$ from $c$ to $c'$. 
\noindent
{\bf\boldmath Case A2: Both $PQ$ and $OPT$ accept $p$, and $h_{PQ}^{(x)}(t-) - h_{OPT}^{(x)}(t-) \leq 0$:}\\
\hspace{2mm}
	Do nothing. 
\noindent
{\bf\boldmath Case A3: $PQ$ rejects $p$, but $OPT$ accepts $p$:}\\
\hspace{2mm}
	Let $c$ be $OPT$'s $(h_{OPT}^{(x)}(t-) + 1)$st cell of $Q^{(x)}$, 
	that is, the cell to which the extra packet $p$ is now stored. 
	Note that $c$ is free at $t-$ but is not at $t+$. 
	There exists the packet $q$ matched with $c$ at $t-$. 
	(See Lemma~\ref{LMA:3.2.2}.)
	Change the partner of $q$ from $c$ to $p$. 
\vspace{2mm}
\noindent
{\bf Scheduling event:} 
If $PQ$ ($OPT$, respectively) has at least one non-empty queue at $t-$, 
suppose that $PQ$ ($OPT$, respectively) transmits a packet from $Q^{(y)}$ ($Q^{(z)}$, respectively) at $t$. 
Execute one of the following three cases at $t$. 
%

%
%
%
\noindent
{\bf\boldmath Case S1.1: {\footnotesize $\sum_{j = 1}^{m} h_{PQ}^{(j)}(t-) > 0$, $\sum_{j = 1}^{m} h_{OPT}^{(j)}(t-) > 0$, $y = z$, and $h_{PQ}^{(y)}(t-) - h_{OPT}^{(y)}(t-) > 0$}:}\\
\hspace{2mm}
	Let $c$ be $OPT$'s $h_{PQ}^{(y)}(t-)$th cell of $Q^{(y)}$, which is free at $t-$ but is not free at $t+$. 
	Let $c'$ be $OPT$'s $h_{OPT}^{(y)}(t-)$th cell of $Q^{(y)}$, which is not free at $t-$ but is free at $t+$. 
	There exists the packet $q$ matched with $c$ at $t-$. 
	(See Lemma~\ref{LMA:3.2.2}.)
	Change the matching partner of $q$ from $c$ to $c'$. 
\noindent
{\bf\boldmath Case S2.2: {\footnotesize $\sum_{j = 1}^{m} h_{PQ}^{(j)}(t-) > 0$, $\sum_{j = 1}^{m} h_{OPT}^{(j)}(t-) > 0$, $y > z$, and $h_{PQ}^{(z)}(t-) - h_{OPT}^{(z)}(t-) \geq 0$}:}\\
\hspace{2mm}
	Let $c$ be $OPT$'s $h_{OPT}^{(z)}(t-)$th cell of $Q^{(z)}$, 
	which becomes free at $t+$. 
	Since the packet $p$ transmitted from $Q^{(y)}$ by $PQ$ is not matched with anything (see Lemma~\ref{LMA:3.2.2}), 
	match $p$ with $c$. 
\noindent
{\bf Otherwise (Cases S1.2, S2.1, S3, \={S}): }
%
%
	Do nothing. 
\\
	\hline 
\end{tabularx}
\end{center}
	\label{tab:rotuine}
\end{table*}
\fi
\ifnum \count10 > 0
定義を与える。
任意のpacket $p$に対して、
$g(p)$は$p$が到着するキューの番号を表す。
また、
任意のcell $c$に対して、
$g(c)$は$c$が存在するキューの番号を表す。
以下で、このroutineの実行可能性を示す。
\fi
\ifnum \count11 > 0
\com{（■英語）\\}
We give some definitions. 
For any packet $p$, 
$g(p)$ denotes the index of the queue at which $p$ arrives.
Also, 
for any cell $c$, 
$g(c)$ denotes the index of the queue including $c$. 
We now show the feasibility of the routine. 
\fi
%

%
\begin{LMA} \label{LMA:3.2.2}
	\ifnum \count10 > 0
	任意のnon-event time $t'$、
	$t'$より前に到着する任意のextra packet $p$と
	$t'$より前に$PQ$にtransmitされるあるpacket $p'$に対して、
	$t'$において$p$には$p'$がmatchされており、
	$g(p) < g(p')$
	が成立している。
	また、
	$t'$における任意のfree cell $c$、
	$t'$より前に$PQ$にtransmitされるあるpacket $p''$に対して、
	$t'$において$c$には$p''$がmatchされており、
	$g(c) < g(p'')$
	が成立している。
	\fi
	\ifnum \count11 > 0
	\com{（■英語）}
	For any non-event time $t'$, and any extra packet $p$ which arrives before $t'$, 
	there exists some packet $p'$ such that $PQ$ transmits $p'$ before $t'$, $g(p) < g(p')$ and $p$ is matched with $p'$ at $t'$. 
	Moreover, 
	for any free cell $c$ at $t'$, 
	there exists some packet $p''$ such that $PQ$ transmits $p''$ before $t'$, $g(c) < g(p'')$, and $c$ is matched with $p''$ at $t'$. 
	\fi
\end{LMA}
\begin{proof}
	\ifnum \count10 > 0
	event timeに関する帰納法により示す。
	base caseは明らかに成立する。
	任意のevent time $t$の直前まで成立すると仮定し、
	$t$の直後においても成立することを示す。
	最初に
	$t$にroutineがCase A1、もしくはCase S1.1を実行する場合を考える。
	$c$を$t$に新しくfreeになるcellとし、
	$c'$を$t-$にfreeであり、$t+$にfreeでなくなるcellとする。
	また、
	帰納法の仮定より、
	$t$より前に$PQ$にtransmitされるpacket $p$が$c'$に$t-$おいてmatchされている。
	Case A1とCase S1.1の定義より、
	$t$にroutineは、$p$をunmatchして、
	$p$を$c$とmatchする。
	このとき、
	$g(c) = g(c')$が明らかに成立し、
	帰納法の仮定より、
	$g(c) < g(p)$が成立するので、
	$t+$においても題意が満たされる。
	次に
	$t$にroutineがCase A3を実行する場合を考える。
	packet $p'$を、$t$において$OPT$にacceptされるextra packetとする。
	$c$を$p'$が$t$にacceptされるfree cellとする。
	帰納法の仮定より、
	$t$より前に$PQ$にtransmitされるpacket $p$が$c$に$t-$においてmatchされている。
	Case A3の定義より、
	$t$にroutineは、$p$をunmatchして、
	$p$を$p'$とmatchする。
	$t+$に$p$がextra packet $p'$にmatchする。
	$g(c) = g(p')$が成立している。
	また、帰納法の仮定より、
	$g(c) < g(p)$
	するので、
	$g(p') < g(p)$が成立し、
	題意が満たされる。
	第三に、
	時刻$t$にroutineがCase S2.2を実行する場合を考える。
	$p$を、$t$に$PQ$にtransmitされるpacketとし、
	$c$を、$t$に新しく生じるfree cellとする。
	帰納法の仮定より、
	$PQ$のmatchしているpacketは$t$より前にtransmitされているので、
	$p$は$t-$において何にもmatchされていない。
	よって、$t$に
	routineは$p$を$c$にmatchすることが出来る。
	また、
	Case S2.2の条件より、
	$g(c) < g(p)$が成立するので、帰納法の仮定より題意が満たされる。
	最後に、
	それ以外の場合は、
	新しくmatchingは生じないので、
	帰納法の仮定より、題意が満たされる。
	\fi
	\ifnum \count11 > 0
	\com{（■英語）}
	The proof is by induction on the event time. 
	The base case is clear. 
	Let $t$ be any event time. 
	We assume that the statement is true at $t-$, 
	and prove that it is true at $t+$. 
	First, 
	we discuss the case where the routine executes Case A1 or S1.1 at $t$. 
	Let $c$ be the cell which becomes free at $t$. 
	Also, 
	let $c'$ be the cell which is free at $t-$ and not free at $t+$. 
	By the induction hypothesis, 
	a packet $p$ which is transmitted by $PQ$ before $t-$ is matched with $c'$ at $t-$. 
	Then, 
	the routine unmatches $p$, and matches $p$ with $c$ by the definitions of Cases A1 and S1.1. 
	$g(c) = g(c')$ clearly holds. 
	Also, since $g(c') < g(p)$ by the induction hypothesis, 
	the statement is true at $t+$. 
	Next, 
	we consider the case where the routine executes Case A3 at $t$. 
	Let $p'$ be the extra packet accepted by $OPT$ at $t$. 
	Also, 
	let $c$ be the free cell into which $OPT$ accepts $p'$ at $t$. 
	By the induction hypothesis, 
	a packet $p$ which is transmitted by $PQ$ before $t-$ is matched with $c$ at $t-$. 
	Then, 
	by the definition of Case A3, 
	the routine unmatches $p$, and matches $p$ with $p'$. 
	$g(c) = g(p')$ holds by definition. 
	In addition, $g(c) < g(p)$ by the induction hypothesis. 
	Thus, 
	$g(p') < g(p)$, which means that 
	the statement holds at $t+$. 
	Third, 
	we investigate the case where the routine executes Case S2.2 at $t$. 
	Suppose that $PQ$ transmits a packet $p$ at $t$, 
	and the new free cell $c$ arises at $t$. 
	By the induction hypothesis, 
	any $PQ$'s packet which is matched with a free cell or an extra packet is transmitted before $t$. 
	Hence, $p$ is not matched with anything at $t-$. 
	Thus, the routine can match $p$ with $c$ at $t$. 
	Moreover, 
	$g(c) < g(p)$ by the condition of Case S2.2. 
	By the induction hypothesis, the statement is true at $t+$. 
	In the other cases, 
	a new matching does not arise. 
	Therefore, 
	the statement is clear by the induction hypothesis, which completes the proof. 
	\fi
\end{proof}
\ifnum \count10 > 0
（■未修正）
以下の補題では、
解析対象となる入力集合を逐次的に作成しつつ、
最終的に${\cal S}^{*}$を構成する。
それと同時に各queueに到着するextra packetの数を評価する。
\fi
\ifnum \count11 > 0
\com{（■英語）}
In the next lemma, 
we obtain part of the properties of the set ${\cal S}^{*}$. 
%
\fi
%

%
\begin{LMA} \label{LMA:3.2.3}
	\ifnum \count10 > 0
	$\sigma$を、
	$s_{u}(\sigma) > B$が成立する様な$u (\in [1, m])$が存在する任意の入力とする。
	このとき、必ず次の様な入力$\hat{\sigma}$が存在する。
	(i) 各$j (\in [1, m])$に対して、$s_{j}(\hat{\sigma}) \leq B$が成立し、
	(ii) $\frac{V_{OPT}(\sigma)}{V_{PQ}(\sigma)} < \frac{V_{OPT}(\hat{\sigma})}{V_{PQ}(\hat{\sigma})}$が成立する。
	\fi
	\ifnum \count11 > 0
	\com{（■英語）}
	Let $\sigma$ be an input such that for some $u (\in [1, m])$, $s_{u}(\sigma) > B$. 
	Then, 
	there exists an input $\hat{\sigma}$ such that 
	for each $j (\in [1, m])$, $s_{j}(\hat{\sigma}) \leq B$, and 
	$\frac{V_{OPT}(\sigma)}{V_{PQ}(\sigma)} < \frac{V_{OPT}(\hat{\sigma})}{V_{PQ}(\hat{\sigma})}$. 
	\fi
\end{LMA}
%
\ifnum \count14 < 1
%
\begin{proof}
	\ifnum \count10 > 0
	$z$を、
	$s_{z}(\sigma) > B$が成立する最小の整数とする。
	このとき、次の3つの条件をみたす様な3つのevent time $t_{1}, t_{2} (> t_{1})$と$t_{3} (> t_{2})$が存在する。
	(i) $t_{2}$は、$Q^{(z)}$において$PQ$がacceptする$B+1$個目のpacketが到着するevent timeである。
	(ii) 
	時間$(t_{1}, t_{2})$の間に、$OPT$は$Q^{(z)}$からpacketをtransmitしない。
	ただし、$t_{1}$は、$OPT$が$Q^{(z)}$からpacketをtransmitするevent timeである。
		（仮定より、$OPT$は全てのpacketを受理するので、$t_{2}$より前に必ず$Q^{(z)}$からpacketをscheduleする。）
	(iii) 
	時間$(t_{2}, t_{3})$の間に、$PQ$は$Q^{(z)}$からpacketをtransmitしない。
	ただし、$t_{3}$は、$PQ$が$Q^{(z)}$からpacketをtransmitするevent timeである。
	このとき、$\sigma$から$t_{1}$と$t_{2}$に発生するeventを除いて$\sigma'$を構成する。
	もし、
	$\frac{V_{OPT}(\sigma)}{V_{PQ}(\sigma)} < \frac{V_{OPT}(\sigma')}{V_{PQ}(\sigma')}$
	が成立するならば、
	$\{ x \mid s_{x}(\sigma) > B \}$に含まれるキューの番号$j$の昇順に、
	$Q^{(j)}$に関するeventを取り除けば、
	\com{（■書き方うーむ。）}
	各$j (\in [1, m])$に対して、
	$s_{j}(\hat{\sigma}) \leq B$
	かつ
	$\frac{V_{OPT}(\sigma)}{V_{PQ}(\sigma)} < \frac{V_{OPT}(\hat{\sigma})}{V_{PQ}(\hat{\sigma})}$
	が成立する入力$\hat{\sigma}$を構成することが出来る。
	これは題意がみたされる。
	よって、以下では、
	$\frac{V_{OPT}(\sigma)}{V_{PQ}(\sigma)} < \frac{V_{OPT}(\sigma')}{V_{PQ}(\sigma')}$
	が成立することを示す。
	まず、
	$\sigma'$に対して$OPT$が得る価値を評価する。
	$ALG$を次の様なオフラインアルゴリズムとする。
	$\sigma'$の各event timeにおいて、$\sigma$において$OPT$が選ぶキューと同じキューを選ぶオフラインアルゴリズムとする。
	このとき、
	時間$(t_{1}, t_{3})$の間における$\sigma'$に対する$ALG$のバッファ内のパケット数について考える。
	任意のnon-event time $t (\in (t_{1}, t_{3}))$、任意の$y (\ne z)$に対して、
	$h_{ALG}^{(y)}(t) = h_{OPT}^{(y)}(t)$
	が成立する。
	任意のnon-event time $t (\in (t_{1}, t_{2}))$に対して、
	$h_{ALG}^{(z)}(t) = h_{OPT}^{(z)}(t) + 1$
	が成立し、また、
	任意のnon-event time $t (\in (t_{2}, t_{3}))$に対して、
	$h_{ALG}^{(z)}(t) = h_{OPT}^{(z)}(t)$
	が成立する。	
	以上より
	$V_{OPT}(\sigma') \geq V_{ALG}(\sigma') = V_{OPT}(\sigma) - \alpha_{z}$
	が成立する。
	次に、$\sigma'$に対して$PQ$が得る価値を評価する。
	簡単のため、
	$\sigma'$に対する$PQ$を$PQ'$と表記する。
	まず、
	時間$(t_{1}, t_{3})$の間に、
	$PQ$が受理するが、$PQ'$が非受理するpacketが存在しない場合を考える。
	この場合の$PQ'$の利得を評価するために
	時間$t_{1}$後の$PQ$と$PQ'$のバッファ内のpacket数について論じる。
	任意のnon-event time $t (\in (t_{1}, t_{2}))$に対して、
	$\sum_{j = 1}^{m} h_{PQ'}^{(j)}(t) = \sum_{j = 1}^{m} h_{PQ}^{(j)}(t) + 1$
	が成立する。
	任意のnon-event time $\hat{t}$に対して、
	$w(\hat{t}) = \arg \max\{ j \mid h_{PQ'}^{(j)}(\hat{t}) > 0 \}$と定義する。
	具体的には、
	$h_{PQ'}^{(w(t))}(t) = h_{PQ}^{(w(t))}(t) + 1$
	が成立する。
	（性質(a)と呼ぶ。）
	また、
	任意のnon-event time $t (\in (t_{2}, t_{3}))$に対して、
	$\sum_{j = 1}^{m} h_{PQ'}^{(j)}(t) = \sum_{j = 1}^{m} h_{PQ}^{(j)}(t)$
	が成立する。
	ただし、
	$w(t) > z$ならば、
	$h_{PQ'}^{(w(t))}(t) = h_{PQ}^{(w(t))}(t) + 1$
	が成立し、
	$h_{PQ'}^{(z)}(t) = h_{PQ}^{(z)}(t) - 1$
	が成立する。
	$w(t) = z$ならば、
	任意の$j (\in [1, m])$に対して、
	$h_{PQ'}^{(j)}(t) = h_{PQ}^{(j)}(t)$
	が成立する。
	任意のnon-event time $t (> t_{3})$、
	任意の$j (\in [1, m])$に対して、
	$h_{PQ'}^{(j)}(t) = h_{PQ}^{(j)}(t)$
	が成立する。
	以上より、
	$V_{PQ}(\sigma') = V_{PQ}(\sigma) - \alpha_{z}$
	が成立する。
	次に
	$PQ$がacceptし、$PQ'$がrejectするpacketが少なくとも1つは到着する場合を考える。
	$t'$を、
	$PQ$がacceptし、$PQ'$がrejectする様なpacket $p$が初めてarriveするevent timeとする。
	$t' \in (t_{1}, t_{2})$が成立する場合を考える。
	$z$の定義より、そのpacketは$z' \geq z$が成立する$Q^{(z')}$に到着するpacketである。
	性質(a)より、
	任意の$j (\in [1, m])$に対して、
	$h_{PQ'}^{(j)}(t'+) = h_{PQ}^{(j)}(t'+)$
	が成立する。
	よって、
	時間$(t', t_{2})$の間に、$PQ$が受理するpacketは、$PQ'$も受理することが出来る。
	$\sigma'$の定義より、
	$t_{2}$において$PQ$のみが$Q^{(z)}$に到着するpacketを受理するので、
	$h_{PQ'}^{(z)}(t_{2}+) = h_{PQ}^{(z)}(t_{2}+) - 1$
	が成立し、
	任意の$j (\in [1, m])$（ただし、$j \ne z$）に対して、
	$h_{PQ'}^{(j)}(t_{2}+) = h_{PQ}^{(j)}(t_{2}+)$
	が成立する。
	（性質(b)と呼ぶ。）
	$t_{2}$の後に、$PQ$と$PQ'$が受理するpacketが全て同じであれば、
	$V_{PQ}(\sigma') = V_{PQ}(\sigma) - \alpha_{z} - \alpha_{z'}$
	が成立する。
	次に、
	$t_{2}$の後に、
	$PQ$がrejectし、$PQ'$がacceptするpacket $p'$が存在する場合を考える。
	$PQ$の定義と性質(b)より、
	任意のnon-event time $t (> t_{2})$、任意の$z' (\geq z+1)$に対して、
	$h_{PQ'}^{(z')}(t) = h_{PQ}^{(z')}(t)$
	が成立する。
	よって、
	$p'$はある$z'' \leq z$が成立する$Q^{(z'')}$に到着するpacketである。
	$p'$が到着するevent timeを$t''$とする。
	任意の$j (\in [1, m])$に対して、
	$h_{PQ'}^{(j)}(t''+) = h_{PQ}^{(j)}(t''+)$
	が成立する。
	したがって、
	$t''$より後に$PQ$と$PQ'$が受理するpacketは一緒である。
	したがって、
	$V_{PQ}(\sigma') = V_{PQ}(\sigma) - \alpha_{z} - \alpha_{z'} + \alpha_{z''} \leq V_{PQ}(\sigma) - \alpha_{z}$
	が成立する。
	最後に
	$t' \in (t_{2}, t_{3})$が成立する場合を考える。
	上の場合と全く同じ議論が成立する。
	具体的には、
	$t'$の後に、
	$PQ$がrejectし、$PQ'$がacceptする様なpacketは高々1つしか存在しない。
	また、そのpacketはある$z''' \leq z$が成立する$Q^{(z''')}$に到着するpacketである。
	よって、
	$V_{PQ}(\sigma') = V_{PQ}(\sigma) - \alpha_{z} - \alpha_{z'} + \alpha_{z'''} \leq V_{PQ}(\sigma) - \alpha_{z}$
	が成立する。
	以上の議論より、
	$\frac{V_{OPT}(\sigma')}{V_{PQ}(\sigma')} 
		\geq \frac{V_{ALG}(\sigma')}{V_{PQ}(\sigma')} 
		\geq \frac{V_{OPT}(\sigma) - \alpha_{z}}{V_{PQ}(\sigma) - \alpha_{z}}
		> \frac{V_{OPT}(\sigma)}{V_{PQ}(\sigma)}$
	が成立する。
	\fi
	\ifnum \count11 > 0
	\com{（■英語）}
	Let $z$ be the minimum index such that $s_{z}(\sigma) > B$. 
	Then, 
	there exist the three event times $t_{1}, t_{2} (> t_{1})$ and $t_{3} (> t_{2})$ satisfying the following three conditions: 
	(i) $t_{2}$ is the arrival event time when the $(B+1)$st packet which $PQ$ accepts at $Q^{(z)}$ arrives, 
	(ii) 
	$OPT$ does not transmit any packet from $Q^{(z)}$ during time $(t_{1}, t_{2})$, 
	where $t_{1}$ is the event time when $OPT$ transmits a packet from $Q^{(z)}$, 
	(Since $OPT$ accepts any arriving packet by our assumption, $OPT$ certainly transmits at least one packet from $Q^{(z)}$ before $t_{2}$.) 
	and 
	(iii) 
	$PQ$ does not transmit any packet from $Q^{(z)}$ during time $(t_{2}, t_{3})$, 
	where $t_{3}$ is the event time when $PQ$ transmits a packet from $Q^{(z)}$. 
	We construct $\sigma'$ by removing the events at $t_{1}$ and $t_{2}$ from $\sigma$. 
	Suppose that $\frac{V_{OPT}(\sigma)}{V_{PQ}(\sigma)} < \frac{V_{OPT}(\sigma')}{V_{PQ}(\sigma')}$. 
	If we remove some events corresponding to $Q^{(j)}$ in ascending order of index $j$ in $\{ x \mid s_{x}(\sigma) > B \}$, 
	then we can construct an input $\hat{\sigma}$ such that 
	for each $j (\in [1, m])$, 
	$s_{j}(\hat{\sigma}) \leq B$, 
	and 
	$\frac{V_{OPT}(\sigma)}{V_{PQ}(\sigma)} < \frac{V_{OPT}(\hat{\sigma})}{V_{PQ}(\hat{\sigma})}$, which completes the proof. 
	Hence, 
	we next show that $\frac{V_{OPT}(\sigma)}{V_{PQ}(\sigma)} < \frac{V_{OPT}(\sigma')}{V_{PQ}(\sigma')}$. 
	First, 
	we discuss the gain of $OPT$ for $\sigma'$. 
	Let $ALG$ be the offline algorithm for $\sigma'$ such that 
	for each scheduling event $e$ in $\sigma'$, 
	$ALG$ selects the queue which $OPT$ selects at $e$ in $\sigma$. 
	We consider the number of packets in $ALG$'s buffer during time $(t_{1}, t_{3})$ for $\sigma'$. 
	For any non-event time $t (\in (t_{1}, t_{3}))$, and any $y (\ne z)$, 
	$h_{ALG}^{(y)}(t) = h_{OPT}^{(y)}(t)$. 
	For any non-event time $t (\in (t_{1}, t_{2}))$, 
	$h_{ALG}^{(z)}(t) = h_{OPT}^{(z)}(t) + 1$. 
	Also, for any non-event time $t (\in (t_{2}, t_{3}))$, 
	$h_{ALG}^{(z)}(t) = h_{OPT}^{(z)}(t)$. 
	By the above argument,
	$V_{OPT}(\sigma') \geq V_{ALG}(\sigma') = V_{OPT}(\sigma) - \alpha_{z}$. 
	Next, we evaluate the gain of $PQ$ for $\sigma'$. 
	For notational simplicity, 
	we describe $PQ$ for $\sigma'$ as $PQ'$. 
	First, 
	we consider the case where there does not exist any packet which $PQ$ accepts but $PQ'$ rejects during time $(t_{1}, t_{3})$. 
	To evaluate the gain of $PQ'$ in this case, 
	we discuss the numbers of packets which $PQ$ and $PQ'$ store in their buffers after $t_{1}$. 
	For any non-event time $t (\in (t_{1}, t_{2}))$, 
	$\sum_{j = 1}^{m} h_{PQ'}^{(j)}(t) = \sum_{j = 1}^{m} h_{PQ}^{(j)}(t) + 1$. 
	For any non-event time $\hat{t}$, 
	we define $w(\hat{t}) = \arg \max\{ j \mid h_{PQ'}^{(j)}(\hat{t}) > 0 \}$. 
	Specifically, 
	$h_{PQ'}^{(w(t))}(t) = h_{PQ}^{(w(t))}(t) + 1$. 
	(We call this fact the property (a).)
	Moreover, 
	for any non-event time $t (\in (t_{2}, t_{3}))$, 
	$\sum_{j = 1}^{m} h_{PQ'}^{(j)}(t) = \sum_{j = 1}^{m} h_{PQ}^{(j)}(t)$. 
	However, 
	if $w(t) > z$, 
	then $h_{PQ'}^{(w(t))}(t) = h_{PQ}^{(w(t))}(t) + 1$. 
	Also, 
	$h_{PQ'}^{(z)}(t) = h_{PQ}^{(z)}(t) - 1$. 
	If $w(t) = z$, 
	then for any $j (\in [1, m])$, $h_{PQ'}^{(j)}(t) = h_{PQ}^{(j)}(t)$. 
	For any non-event time $t (> t_{3})$ and 
	any $j (\in [1, m])$, 
	$h_{PQ'}^{(j)}(t) = h_{PQ}^{(j)}(t)$. 
	By the above argument, 
	$V_{PQ}(\sigma') = V_{PQ}(\sigma) - \alpha_{z}$ holds. 
	Secondly, 
	we consider the case where there exists at least one packet which $PQ$ accepts but $PQ'$ rejects. 
	Let $t'$ be the first event time when the packet $p$ which $PQ$ accepts but $PQ'$ rejects arrives. 
	Then, suppose that $t' \in (t_{1}, t_{2})$. 
	By the definition of $z$, 
	$p$ arrives at $Q^{(z')}$ such that $z' \geq z$. 
	By the property (a), 
	for $j (\in [1, m])$, 
	$h_{PQ'}^{(j)}(t'+) = h_{PQ}^{(j)}(t'+)$. 
	Thus, 
	packets accepted by $PQ$ during time $(t', t_{2})$ can be accepted by $PQ'$. 
	Only $PQ$ accepts the packet arriving at $Q^{(z)}$ at $t_{2}$ by the definition of $\sigma'$. 
	Hence, 
	$h_{PQ'}^{(z)}(t_{2}+) = h_{PQ}^{(z)}(t_{2}+) - 1$, 
	and 
	for any $j (\in [1, m])$ such that $j \ne z$, 
	$h_{PQ'}^{(j)}(t_{2}+) = h_{PQ}^{(j)}(t_{2}+)$. 
	(We call this fact the property (b).)
	If all the packets which $PQ$ accepts after $t_{2}$ are the same as those accepted by $PQ'$ after $t_{2}$, 
	$V_{PQ}(\sigma') = V_{PQ}(\sigma) - \alpha_{z} - \alpha_{z'}$. 
	Then, 
	we consider the case where there exists at least one packet $p'$ which $PQ$ rejects but $PQ'$ accepts after $t_{2}$. 
	By the greediness of $PQ$ and the property (b), 
	for any non-event time $t (> t_{2})$ and any $z' (\geq z+1)$, 
	$h_{PQ'}^{(z')}(t) = h_{PQ}^{(z')}(t)$. 
	Hence, 
	$p'$ arrives at $Q^{(z'')}$ for some $z'' (\leq z)$. 
	Let $t''$ be the event time when $p'$ arrives. 
	For any $j (\in [1, m])$, 
	$h_{PQ'}^{(j)}(t''+) = h_{PQ}^{(j)}(t''+)$, 
	which means that 
 	all the packets accepted by $PQ$ are equal to those accepted by $PQ'$ after $t''$. 
	Thus, 
	$V_{PQ}(\sigma') = V_{PQ}(\sigma) - \alpha_{z} - \alpha_{z'} + \alpha_{z''} \leq V_{PQ}(\sigma) - \alpha_{z}$. 
	Finally, 
	we consider the case where $t' \in (t_{2}, t_{3})$. 
	By the same argument as the case of $t' \in (t_{1}, t_{2})$, 
	we can prove this case. 
	Specifically, 
	the number of packets which $PQ$ rejects but $PQ'$ accepts after $t'$ is exactly one. 
	This packet arrives at $Q^{(z''')}$, where some $z''' \leq z$. 
	Therefore, 
	$V_{PQ}(\sigma') = V_{PQ}(\sigma) - \alpha_{z} - \alpha_{z'} + \alpha_{z'''} \leq V_{PQ}(\sigma) - \alpha_{z}$. 
	By the above argument, 
	$\frac{V_{OPT}(\sigma')}{V_{PQ}(\sigma')} 
		\geq \frac{V_{ALG}(\sigma')}{V_{PQ}(\sigma')} 
		\geq \frac{V_{OPT}(\sigma) - \alpha_{z}}{V_{PQ}(\sigma) - \alpha_{z}}
		> \frac{V_{OPT}(\sigma)}{V_{PQ}(\sigma)}$. 
	\fi
\end{proof}
%
\fi
%

%
\ifnum \count10 > 0
（■）
定義を与える。
任意の$j (\in [1, m])$に対して、
$s_{j}(\sigma) \leq B$
が成立する様な入力$\sigma$の集合を${\cal S}_{1}$で表す。
補題~\ref{LMA:3.2.3}より、
以下では、
${\cal S}_{1}$に含まれる入力についてのみ解析を行う。	
\fi
\ifnum \count11 > 0
\com{（■英語）}
We give the notation. 
${\cal S}_{1}$ denotes the set of inputs $\sigma$ such that 
for any $j (\in [1, m])$, 
$s_{j}(\sigma) \leq B$. 
In what follows, 
we analyze only inputs in ${\cal S}_{1}$ by Lemma~\ref{LMA:3.2.3}. 
Next, we evaluate the number of extra packets arriving at each good queue using Lemma~\ref{LMA:3.2.2}.
\fi
%

%
\begin{LMA} \label{LMA:3.2.4}
	\ifnum \count10 > 0
	任意の$x (\in [1, n])$に対して、
	$\sum_{i = x}^{n} k_{q_{i}} \leq \sum_{j = q_{x}+1}^{m} s_{j}$
	が成立する。
	\fi
	\ifnum \count11 > 0
	\com{（■英語）}
	For any $x (\in [1, n])$, 
	$\sum_{i = x}^{n} k_{q_{i}} \leq \sum_{j = q_{x}+1}^{m} s_{j}$. 
	\fi
\end{LMA}
%
\ifnum \count14 < 1
%
\begin{proof}
	\ifnum \count10 > 0
	補題~\ref{LMA:3.2.2}より、
	入力が終了した時点において、
	各extra packet $p$は$PQ$がtransmitするpacket $p'$に必ずmatchされており、
	また、
	extra packet $p$と$PQ$のpacket $p'$がmatchされている場合、
	$g(p) < g(p')$
	が成立している。
	よって、
	$k_{q_{n}} \leq \sum_{j = q_{n}+1}^{m} s_{j}$
	かつ
	$k_{q_{n-1}} \leq (\sum_{j = q_{n-1}+1}^{m} s_{j}) - k_{q_{n}}$
	かつ
	…
	$k_{q_{1}} \leq (\sum_{j = q_{1}+1}^{m} s_{j}) - \sum_{i = 2}^{n} k_{q_{i}}$
	が成立する。
	よって、
	任意の$x (\in [1, n])$に対して、
	$\sum_{i = x}^{n} k_{q_{i}} \leq \sum_{j = q_{x}+1}^{m} s_{j}$
	が成立する。
	\fi
	\ifnum \count11 > 0
	\com{（■英語）}
	By Lemma~\ref{LMA:3.2.2}, 
	each extra packet $p$ is matched with a packet $p'$ transmitted by $PQ$ at the end of the input. 
	In addition, 
	$g(p) < g(p')$ 
	if an extra packet $p$ is matched with a packet $p'$ of $PQ$. 
	Thus, 
	$k_{q_{n}} \leq \sum_{j = q_{n}+1}^{m} s_{j}$, 
	$k_{q_{n-1}} \leq (\sum_{j = q_{n-1}+1}^{m} s_{j}) - k_{q_{n}}$, $\cdots$, and 
	$k_{q_{1}} \leq (\sum_{j = q_{1}+1}^{m} s_{j}) - \sum_{i = 2}^{n} k_{q_{i}}$. 
	Therefore, 
	for any $x (\in [1, n])$, 
	$\sum_{i = x}^{n} k_{q_{i}} \leq \sum_{j = q_{x}+1}^{m} s_{j}$. 
	\fi
\end{proof}
%
\fi
%

%
\ifnum \count10 > 0
（■未修正）
次のLemma~\ref{LMA:3.2.5}では、
ある good queue $Q^{(q_{i})}$ ($i \in [1, n]$)について、どのようにeventが発生すると、
$Q^{(q_{i})}$に到着する extra packet の数が最大になるかについて議論し、
それは
$k_{q_{i}} = \sum_{j = q_{i} + 1}^{q_{i+1}} s_{j}$
が成立する場合であることを示す。
具体的には、
この式が成立するのは、
$PQ$が各$Q^{(j)}$ ($j \in [q_{i}+1, q_{i+1}]$)から$s_{j}$個のpacketを送信する際に、
$OPT$が必ず$Q^{(q_{i})}$から$\sum_{j = q_{i} + 1}^{q_{i+1}} s_{j}$個のpacketを送信し、
かつ各scheduling eventの直後に$Q^{(q_{i})}$にpacketが到着するという状況である。
この場合、
$\sum_{j = q_{i} + 1}^{q_{i+1}} s_{j}$個のextra packetが$Q^{(q_{i})}$においてacceptされる。
以下の全ての補題の証明中では、
簡単のため、
任意の$j (\in [1, m])$に対して、
$s_{j}(\sigma)$を$s_{j}$と表記し、
$n(\sigma)$を$n$と表記し、
任意の$i (\in [1, n])$に対して、
$q_{i}(\sigma)$を$q_{i}$と表記し、
任意の$j (\in [1, m])$に対して、
$k_{j}(\sigma)$を$k_{j}$と表記する。
\fi
\ifnum \count11 > 0
\com{（■英語）}
Now we gradually gain all the properties of ${\cal S}^{*}$ in the following lemmas 
while proving ${\cal S}^{*}$ contains inputs $\sigma$ such that $\frac{V_{OPT}(\sigma)}{V_{PQ}(\sigma)}$ is maximized. 
Specifically, for $i = 1, \ldots, 4$, 
we construct some subset ${\cal S}_{i+1}$ from the set ${\cal S}_i$ in each of the following lemmas, 
and eventually we can gain ${\cal S}^{*}$ from ${\cal S}_{5}$. 
(We have already obtained ${\cal S}_{1}$ in Lemma~\ref {LMA:3.2.3}.) 
It is difficult to show all the properties of ${\cal S}^{*}$ in one lemma, 
and thus we progressively give the definitions of the ${\cal S}_{i+1}$ that has more restrictive properties than ${\cal S}_{i}$. 
Next in Lemma~\ref{LMA:3.2.5}, 
we discuss the condition of events where the number of extra packets accepted into a good queue $Q^{(q_{i})}$ ($i \in [1, n]$) is maximized, and show that it is true when $k_{q_{i}} = \sum_{j = q_{i} + 1}^{q_{i+1}} s_{j}$. 
%
%
%
Throughout the proofs of all the following lemmas,
we drop $\sigma$ from $s_{j}(\sigma)$, $n(\sigma)$, $q_{i}(\sigma)$ and $k_{j}(\sigma)$. 
\fi
%

%
\begin{LMA} \label{LMA:3.2.5}
	\ifnum \count10 > 0
	任意の入力$\sigma (\in {\cal S}_1)$に対して、
	必ず次の様な入力$\hat{\sigma} (\in {\cal S}_1)$が存在する。
	(i) 各$i (\in [1, n(\hat{\sigma})])$に対して、
		$k_{q_{i}(\hat{\sigma})}(\hat{\sigma}) = \sum_{j = q_{i}(\hat{\sigma}) + 1}^{q_{i+1}(\hat{\sigma})} s_{j}(\hat{\sigma})$が成立する。
	(ii) $q_{1}(\hat{\sigma}) - 1 \geq 1$の場合、任意の$j (\in [1, q_{1}(\hat{\sigma})-1])$に対して、$s_{j}(\hat{\sigma}) = 0$が成立し、
	(iii) $\frac{V_{OPT}(\sigma)}{V_{PQ}(\sigma)} \leq \frac{V_{OPT}(\hat{\sigma})}{V_{PQ}(\hat{\sigma})}$が成立する。
	\fi
	\ifnum \count11 > 0
	\com{（■英語）}
	For any input $\sigma \in {\cal S}_1$, 
	there exists an input $\hat{\sigma} (\in {\cal S}_1)$ such that 
	(i) for any $i (\in [1, n(\hat{\sigma})])$, 
		$k_{q_{i}(\hat{\sigma})}(\hat{\sigma}) = \sum_{j = q_{i}(\hat{\sigma}) + 1}^{q_{i+1}(\hat{\sigma})} s_{j}(\hat{\sigma})$, 
	(ii) for any $j (\in [1, q_{1}(\hat{\sigma})-1])$, $s_{j}(\hat{\sigma}) = 0$ if $q_{1}(\hat{\sigma}) - 1 \geq 1$, and 
	(iii) $\frac{V_{OPT}(\sigma)}{V_{PQ}(\sigma)} \leq \frac{V_{OPT}(\hat{\sigma})}{V_{PQ}(\hat{\sigma})}$. 
	\fi
\end{LMA}
%
\ifnum \count14 < 1
%
\begin{proof}
	\ifnum \count10 > 0
	任意の入力$\sigma \in {\cal S}_1$に対して、
	$\sigma$から次の様な入力$\sigma'$を構成する。
	まず、
	時間$(0, 1)$の間に各$Q^{(j)} \hspace{1mm} (j \in [q_{1}, m])$にpacketが到着する$s_{j}$回のarrival eventが発生する。
	${\cal S}_1$の定義より、
	$s_{j} \leq B$が成立するので、$PQ$はそれらの到着するpacketを全て受理する。
	時刻1以降に、
	$PQ$がacceptできない$\sum_{i = 1}^{n} k_{q_{i}}$個のpacketだけがarriveする。
	詳しく言うと、
	任意の$i (\in [1, n])$に対して、
	$a_{i} = \sum_{j = q_{n+1-i}+1}^{q_{n+2-i}} s_{j}$と定義し、
	$a_{0} = 0$
	と定義する。
	任意の$x (\in [0, n-1])$に対して、
	整数時間$t = (\sum_{j = 0}^{x} a_{j} ) + 1, \ldots, \sum_{j = 0}^{x+1} a_{j}$に、
	scheduling eventが発生し、
	時刻$t + \frac{1}{2}$において、
	$Q^{(q_{n-x})}$にpacketが到着する様なarrival eventが発生する。
	時刻$(\sum_{j = 0}^{n} a_{j})+1$の後は、
	十分な数のscheduling eventが発生する。
	このとき、
	$PQ$は$t$に、$j (\in [q_{n-x}+1, q_{n-x+1}])$が成立する$Q^{(j)}$からpacketをtransmitする。
	また、
	あるオフラインアルゴリズム$ALG$を考える。
	$ALG$は$t$に、$Q^{(q_{n-x})}$からpacketをtransmitする。
	このとき、
	任意の$i (\in [1, n])$に対して、
	$Q^{(q_{i})}$はextra packetが到着するので、
	$s_{q_{i}} = B$が成立している。
	よって、
	任意の$i (\in [1, n])$に対して、
	$h_{PQ}^{(q_{i})}(1-) = B$が成立しているので、
	$PQ$は各時刻$t + \frac{1}{2}$において、到着する全てのpacketをacceptできない。
	しかし、$ALG$は全てのpacketを受理することが出来る。
	すなわち、$ALG$は最適なオフラインアルゴリズムの1つである。
	このとき、
	$n(\sigma') = n$
	かつ
	任意の$i (\in [1, n])$に対して、
	$q_{i}(\sigma') = q_{i}$
	が成立する。
	以上より、
	$V_{PQ}(\sigma') = V_{PQ}(\sigma) - \sum_{j = 1}^{q_{1}-1} \alpha_{j} s_{j}$
	が成立する。
	また、各$i (\in [1, n])$に対して、
	$k_{q_{i}}(\sigma') = \sum_{j = q_{i}+1}^{q_{i+1}} s_{j}$
	が成立する。
	以上の式より、
	$V_{ALG}(\sigma') 
		= V_{PQ}(\sigma') + \sum_{i = 1}^{n} \alpha_{q_{i}} k_{q_{i}}(\sigma')
		= V_{PQ}(\sigma) + \sum_{i = 1}^{n} \alpha_{q_{i}} ( \sum_{j = q_{i}+1}^{q_{i+1}} s_{j} ) - \sum_{j = 1}^{q_{1}-1} \alpha_{j} s_{j}$
	が成立する。
	補題~\ref{LMA:3.2.4}より、
	$\sum_{i = x}^{n} k_{q_{i}} \leq \sum_{j = q_{x}+1}^{m} s_{j}$
	が成立するので、
	上の不等式から各$s_{j}$を消去すると、
	$V_{ALG}(\sigma') 
		\geq V_{PQ}(\sigma) + \sum_{i = 1}^{n} \alpha_{q_{i}} k_{q_{i}} - \sum_{j = 1}^{q_{1}-1} \alpha_{j} s_{j}
		= V_{OPT}(\sigma) - \sum_{j = 1}^{q_{1}-1} \alpha_{j} s_{j}$
	が成立する。
	以上より、
	$\frac{V_{OPT}(\sigma')}{V_{PQ}(\sigma')} = \frac{V_{ALG}(\sigma')}{V_{PQ}(\sigma')}
		\geq \frac{V_{OPT}(\sigma) - \sum_{j = 1}^{q_{1}-1} \alpha_{j} s_{j} }{V_{PQ}(\sigma) - \sum_{j = 1}^{q_{1}-1} \alpha_{j} s_{j}}
		\geq \frac{V_{OPT}(\sigma)}{V_{PQ}(\sigma)}
	$
	が成立する。
	また、$\sigma'$の定義より、
	$\sigma'$は条件(ii)を満たし、${\cal S}_1$に含まれる。
	\com{（■S)}
	\fi
	\ifnum \count11 > 0
	\com{（■英語）}
	For any input $\sigma \in {\cal S}_1$, 
	we construct $\sigma'$ from $\sigma$ according to the following steps. 
	First, 
	for each $j (\in [q_{1}, m])$, 
	$s_{j}$ events at which $s_{j}$ packets arrive at $Q^{(j)}$ occur during time $(0, 1)$. 
	Since $s_{j} \leq B$ by the definition of ${\cal S}_1$, 
	$PQ$ accepts all the packets which arrive at these events. 
	$\sum_{i = 1}^{n} k_{q_{i}}$ packets arrive after time 1, and $PQ$ cannot accept them.  
	Specifically, 
	for any $i (\in [1, n])$, 
	we define $a_{i} = \sum_{j = q_{n+1-i}+1}^{q_{n+2-i}} s_{j}$ and 
	$a_{0} = 0$. 
	Then, 
	for each $x (\in [0, n-1])$, 
	a scheduling event occurs at each integer time $t = (\sum_{j = 0}^{x} a_{j} ) + 1, \ldots, \sum_{j = 0}^{x+1} a_{j}$, and 
	an arrival event where a packet arrives at $Q^{(q_{n-x})}$ occurs at each time $t + \frac{1}{2}$. 
	After time $(\sum_{j = 0}^{n} a_{j}) + 1$, 
	sufficient scheduling events to transmit all the arriving packets occur. 
	For these scheduling events, 
	$PQ$ transmits a packet from $Q^{(j)}$ at $t$, 
	where $j$ is an integer between $q_{n-x}+1$ and $q_{n-x+1}$. 
	Also, 
	let $ALG$ be an offline algorithm. 
	$ALG$ transmits a packet from $Q^{(q_{n-x})}$ at $t$. 
	Since for any $i (\in [1, n])$, at least one extra packet arrives at $Q^{(q_{i})}$, 
	$s_{q_{i}} = B$ holds. 
	Hence, 
	since for any $i (\in [1, n])$, 
	$h_{PQ}^{(q_{i})}(1-) = B$, 
	$PQ$ cannot accept the packet which arrives at each $t + \frac{1}{2}$. 
	However, 
	$ALG$ can accept all these packets, which means that 
	$ALG$ is an optimal offline algorithm.
	Then, 
	$n(\sigma') = n$, and 
	for any $i (\in [1, n])$, 
	$q_{i}(\sigma') = q_{i}$. 
	By the above argument, 
	$V_{PQ}(\sigma') = V_{PQ}(\sigma) - \sum_{j = 1}^{q_{1}-1} \alpha_{j} s_{j}$. 
	Furthermore, for each $i (\in [1, n])$, 
	$k_{q_{i}}(\sigma') = \sum_{j = q_{i}+1}^{q_{i+1}} s_{j}$. 
	By these equalities, 
	$V_{ALG}(\sigma') 
		= V_{PQ}(\sigma') + \sum_{i = 1}^{n} \alpha_{q_{i}} k_{q_{i}}(\sigma')
		= V_{PQ}(\sigma) + \sum_{i = 1}^{n} \alpha_{q_{i}} ( \sum_{j = q_{i}+1}^{q_{i+1}} s_{j} ) - \sum_{j = 1}^{q_{1}-1} \alpha_{j} s_{j} 
		= V_{PQ}(\sigma) + \alpha_{q_{1}} ( \sum_{j = q_{1}+1}^{q_{n+1}} s_{j} ) + \sum_{x = 2}^{n} (\alpha_{q_{x}} - \alpha_{q_{x-1}}) ( \sum_{j = q_{x}+1}^{q_{n+1}} s_{j} ) - \sum_{j = 1}^{q_{1}-1} \alpha_{j} s_{j}$. 
	Since $\sum_{i = x}^{n} k_{q_{i}} \leq \sum_{j = q_{x}+1}^{m} s_{j}$ by Lemma~\ref{LMA:3.2.4} and $q_{n+1}=m$, 
	$V_{ALG}(\sigma') 
		\geq V_{PQ}(\sigma) + \alpha_{q_{1}} ( \sum_{i = 1}^{n} k_{q_{i}} ) + \sum_{x = 2}^{n} (\alpha_{q_{x}} - \alpha_{q_{x-1}}) ( \sum_{i = x}^{n} k_{q_{i}} ) - \sum_{j = 1}^{q_{1}-1} \alpha_{j} s_{j}
		= V_{PQ}(\sigma) + \sum_{i = 1}^{n} \alpha_{q_{i}} k_{q_{i}} - \sum_{j = 1}^{q_{1}-1} \alpha_{j} s_{j}
		= V_{OPT}(\sigma) - \sum_{j = 1}^{q_{1}-1} \alpha_{j} s_{j}$. 
	Therefore, 
	$\frac{V_{OPT}(\sigma')}{V_{PQ}(\sigma')} = \frac{V_{ALG}(\sigma')}{V_{PQ}(\sigma')}
		\geq \frac{V_{OPT}(\sigma) - \sum_{j = 1}^{q_{1}-1} \alpha_{j} s_{j} }{V_{PQ}(\sigma) - \sum_{j = 1}^{q_{1}-1} \alpha_{j} s_{j}}
		\geq \frac{V_{OPT}(\sigma)}{V_{PQ}(\sigma)}
	$. 
	Moreover, 
	by the definition of $\sigma'$, 
	$\sigma'$ satisfies the condition (ii) in the statement, 
	which means that ${\cal S}_1$ includes $\sigma'$. 
	\fi
\end{proof}
%
\fi
%

%

%
\ifnum \count10 > 0
この補題を踏まえて、次の様な入力の集合を定義する。
${\cal S}_{2}$は、以下が成立する様な任意の入力$\sigma (\in {\cal S}_1)$の集合を表す：
\com{（■S)}
(i)
各$i (\in [1, n])$に対して、
$k_{q_{i}} = \sum_{j = q_{i} + 1}^{q_{i+1}} s_{j}$
が成立し、
(ii)
各$j (\in [q_{1}, m])$に対して、
$s_{j} \leq B$
が成立し、
(iii)
$q_{1} - 1 \geq 1$が成立するならば、
任意の$j (\in [1, q_{1}-1])$に対して、$s_{j} = 0$が成立する。
\fi
\ifnum \count11 > 0
\com{（■英語）}
In light of the above lemma, we introduce the next set of inputs. 
${\cal S}_{2}$ denotes the set of inputs $\sigma (\in {\cal S}_1)$ satisfying the following conditions: 
(i)
for any $i (\in [1, n])$, 
$k_{q_{i}} = \sum_{j = q_{i} + 1}^{q_{i+1}} s_{j}$, 
(ii)
for any $j (\in [q_{1}, m])$, 
$s_{j} \leq B$, and 
(iii)
for any $j (\in [1, q_{1}-1])$, $s_{j} = 0$ 
if $q_{1} - 1 \geq 1$. 
\fi
%

%
\begin{LMA} \label{LMA:3.2.6}
	\ifnum \count10 > 0
	$\sigma (\in {\cal S}_{2})$を、
	ある$z (\leq n(\sigma)-1)$が存在して、
	$q_{z}(\sigma) + 1 < q_{z+1}(\sigma)$が成立する様な任意の入力とする。
	このとき、次の様な入力$\hat{\sigma} (\in {\cal S}_{2})$が存在する。
	(i) 各$i (\in [1, n(\hat{\sigma})-1])$に対して、
		$q_{i}(\hat{\sigma}) + 1 = q_{i+1}(\hat{\sigma})$と
		$k_{q_{i}(\hat{\sigma})}(\hat{\sigma}) = B$が成立し、
	(ii) $\frac{V_{OPT}(\sigma)}{V_{PQ}(\sigma)} \leq \frac{V_{OPT}(\hat{\sigma})}{V_{PQ}(\hat{\sigma})}$が成立する。
	\fi
	\ifnum \count11 > 0
	\com{（■英語）}
	Let $\sigma (\in {\cal S}_{2})$ be an input such that 
	for some $z (\leq n(\sigma)-1)$, 
	$q_{z}(\sigma) + 1 < q_{z+1}(\sigma)$. 
	Then, there exists an input $\hat{\sigma} (\in {\cal S}_{2})$ such that 
	(i) for each $i (\in [1, n(\hat{\sigma})-1])$, 
		$q_{i}(\hat{\sigma}) + 1 = q_{i+1}(\hat{\sigma})$ and 
		$k_{q_{i}(\hat{\sigma})}(\hat{\sigma}) = B$, and 
	(ii) $\frac{V_{OPT}(\sigma)}{V_{PQ}(\sigma)} \leq \frac{V_{OPT}(\hat{\sigma})}{V_{PQ}(\hat{\sigma})}$. 
	\fi
\end{LMA}
\ifnum \count14 < 1
\begin{proof}
	\ifnum \count10 > 0
	任意の$j (\in [1, m])$（ただし、$j \ne q_{z+1}-1$）に対して、
	$s'_{j} = s_{j}$
	と定義する。
	また、
	$s'_{q_{z+1}-1} = B$
	と定義する。
	（Appendix~\ref{ap.sec:3}の図~\ref{fig:L37}参照）
	次の様に$\sigma$から入力$\sigma'$を構成する。
	この構成の仕方は、補題~\ref{LMA:3.2.5}に類似している。
	まず、
	時間$(0, 1)$の間に各$Q^{(j)} \hspace{1mm} (j \in [q_{1}, m])$にpacketが到着する$s'_{j}$回のarrival eventが発生する。
	定義より、
	$s'_{j} \leq B$が成立するので、
	$PQ$はそれらの到着するpacketを全て受理する。
	ここで、
	任意の$i (\in [1, z])$に対して、
	$q'_{i} = q_{i}$と定義し、
	$q'_{z+1} = q_{z+1} - 1$と定義し、
	任意の$i (\in [z+1, n+1])$に対して、
	$q'_{i+1} = q_{i}$と定義する。
	そのとき、
	任意の$i (\in [1, n+1])$に対して、
	$a_{i} = \sum_{j = q'_{n+2-i}+1}^{q'_{n+3-i}} s'_{j}$と定義し、
	$a_{0} = 0$と定義する。
	更に、
	任意の$x (\in [0, n])$に対して、
	整数時間$t = (\sum_{j = 0}^{x} a_{j} ) + 1, \ldots, \sum_{j = 0}^{x+1} a_{j}$に、
	scheduling eventが発生し、
	時刻$t + \frac{1}{2}$において、
	$Q^{(q'_{n-x+1})}$にpacketが到着する様なarrival eventが発生する。
	時刻$(\sum_{j = 0}^{n+1} a_{j})+1$の後は、
	十分な数のscheduling eventが発生する。
	このとき、
	$PQ$は$t$に、$j (\in [q'_{n-x+1}+1, q'_{n-x+2}])$が成立する$Q^{(j)}$からpacketをtransmitする。
	また、
	あるオフラインアルゴリズム$ALG$を考える。
	$ALG$は$t$に、$Q^{(q'_{n-x+1})}$からpacketをtransmitする。
	$q'_{i}$の定義より
	任意の$i (\in [1, n+1])$に対して、
	$h_{PQ}^{(q'_{i})}(1-) = B$が成立しているので、
	$PQ$は各時刻$t + \frac{1}{2}$において、到着する全てのpacketをacceptできない。
	しかし、$ALG$は全てのpacketを受理することが出来る。
	すなわち、$ALG$は最適なオフラインアルゴリズムである。
	以上より、
	$V_{PQ}(\sigma') = V_{PQ}(\sigma) + \alpha_{q_{z+1}-1} (B - s_{q_{z+1}-1})$
	が成立する。
	また、
	任意の$j (\ne q_{z}, q_{z+1}-1)$に対して、
	$k_{j}(\sigma') = k_{j}$
	が成立し、
	$k_{q_{z}}(\sigma') = k_{q_{z}} - s_{q_{z+1}-1}$
	が成立し、
	$k_{q_{z+1}-1}(\sigma') = B$
	が成立する。
	また、
	任意の$i (\in [1, n+1])$に対して、
	$q_{i}(\sigma') = q'_{i}$
	が成立する。
	更に、
	$V_{OPT}(\sigma') = V_{ALG}(\sigma') = V_{PQ}(\sigma') + \sum_{i = 1}^{n(\sigma')} \alpha_{q_{i}(\sigma')} k_{q_{i}(\sigma')}(\sigma')$
	が成立する。
	以上の式より、
	$\sum_{i = 1}^{n(\sigma')} \alpha_{q_{i}(\sigma')} k_{q_{i}(\sigma')}(\sigma')
		= (\sum_{i = 1}^{n} \alpha_{q_{i}} k_{q_{i}}) - \alpha_{q_{z}} s_{q_{z+1}-1} + \alpha_{q_{z+1}-1} B
		\geq (\sum_{i = 1}^{n} \alpha_{q_{i}} k_{q_{i}}) + \alpha_{q_{z+1}-1} (B - s_{q_{z+1}-1})
	$
	が成立し、
	$\frac{ \sum_{i = 1}^{n(\sigma')} \alpha_{q_{i}(\sigma')} k_{q_{i}(\sigma')}(\sigma') }{ V_{PQ}(\sigma') }
		\geq \frac{ (\sum_{i = 1}^{n} \alpha_{q_{i}} k_{q_{i}}) + \alpha_{q_{z+1}-1} (B - s_{q_{z+1}-1}) }
			{ V_{PQ}(\sigma) + \alpha_{q_{z+1}-1} (B - s_{q_{z+1}-1}) }
		\geq \frac{ \sum_{i = 1}^{n} \alpha_{q_{i}} k_{q_{i}} }{ V_{PQ}(\sigma) }
	$
	が成立する。
	よって、
	$\frac{V_{OPT}(\sigma')}{V_{PQ}(\sigma')} 
		\geq \frac{ V_{PQ}(\sigma') + \sum_{i = 1}^{n(\sigma')} \alpha_{q_{i}(\sigma')} k_{q_{i}(\sigma')}(\sigma') }
					{ V_{PQ}(\sigma') }
		\geq 1 + \frac{ \sum_{i = 1}^{n} \alpha_{q_{i}} k_{q_{i}} }{ V_{PQ}(\sigma) }
		= \frac{V_{OPT}(\sigma)}{V_{PQ}(\sigma)} 
	$
	が成立する。
	上記の$\sigma'$の定義より、
	$\sigma' \in {\cal S}_{2}$が成立している。
	よって、
	以上の議論より、再帰的に
	$q_{z'} + 1 < q_{z'+1}$が成立する様な任意の$z'$に対して、
	上記の様に新しい入力を構成すると、
	題意をみたす入力を構成することができる。
	\fi
	\ifnum \count11 > 0
	\com{（■英語）}
	For any $j (\in [1, m])$ such that $j \ne q_{z+1}-1$, 
	we define $s'_{j} = s_{j}$. 
	Also, we define $s'_{q_{z+1}-1} = B$. 
	(See Figure~\ref{fig:L37}.)
	We construct $\sigma'$ from $\sigma$ in the following way. 
	This approach is similar to those in the proof of Lemma~\ref{LMA:3.2.5}. 
	First, 
	for each $j (\in [q_{1}, m])$, 
	$s'_{j}$ events at which $s'_{j}$ packets arrive at $Q^{(j)}$ occur during time $(0, 1)$. 
	Since $s'_{j} \leq B$ by definition, 
	$PQ$ accepts all these packets. 
	In addition, 
	for any $i (\in [1, z])$, 
	we define $q'_{i} = q_{i}$. 
	We define $q'_{z+1} = q_{z+1} - 1$. 
	For any $i (\in [z+1, n+1])$, 
	we define $q'_{i+1} = q_{i}$. 
	Moreover, 
	for any $i (\in [1, n+1])$, 
	we define $a_{i} = \sum_{j = q'_{n+2-i}+1}^{q'_{n+3-i}} s'_{j}$ and 
	$a_{0} = 0$. 
	For any $x (\in [0, n])$, 
	a scheduling event occurs at each integer time $t = (\sum_{j = 0}^{x} a_{j} ) + 1, \ldots, \sum_{j = 0}^{x+1} a_{j}$. 
	Also, 
	an arrival event where a packet arrives at $Q^{(q'_{n-x+1})}$ occurs at each time $t + \frac{1}{2}$. 
	After time $(\sum_{j = 0}^{n+1} a_{j}) + 1$, 
	sufficient scheduling events to transmit all the arriving packets occur. 
	Then, 
	$PQ$ transmits a packet from $Q^{(j)}$ at $t$, 
	where $j$ is an integer between $q'_{n-x+1}+1$ and $q'_{n-x+2}$. 
	Let $ALG$ be an offline algorithm which transmits a packet from $Q^{(q'_{n-x+1})}$ at $t$. 
	By the definition of $q'_{i}$, 
	for any $i (\in [1, n+1])$, 
	$h_{PQ}^{(q'_{i})}(1-) = B$. 
	Thus, 
	$PQ$ cannot accept any packet arriving at $t + \frac{1}{2}$, 
	but $ALG$ can accept all the arriving packets. 
	That is to say, $ALG$ is optimal. 
	By the above argument, 
	$V_{PQ}(\sigma') = V_{PQ}(\sigma) + \alpha_{q_{z+1}-1} (B - s_{q_{z+1}-1})$. 
	Furthermore, 
	for any $j (\ne q_{z}, q_{z+1}-1)$, 
	$k_{j}(\sigma') = k_{j}$. 
	Also, 
	$k_{q_{z}}(\sigma') = k_{q_{z}} - s_{q_{z+1}-1}$ and 
	$k_{q_{z+1}-1}(\sigma') = B$. 
	Also, 
	for any $i (\in [1, n+1])$, 
	$q_{i}(\sigma') = q'_{i}$.
	Moreover, 
	$V_{OPT}(\sigma') = V_{ALG}(\sigma') = V_{PQ}(\sigma') + \sum_{i = 1}^{n(\sigma')} \alpha_{q_{i}(\sigma')} k_{q_{i}(\sigma')}(\sigma')$. 
	By the above equalities, 
	$\sum_{i = 1}^{n(\sigma')} \alpha_{q_{i}(\sigma')} k_{q_{i}(\sigma')}(\sigma')
		= (\sum_{i = 1}^{n} \alpha_{q_{i}} k_{q_{i}}) - \alpha_{q_{z}} s_{q_{z+1}-1} + \alpha_{q_{z+1}-1} B
		\geq (\sum_{i = 1}^{n} \alpha_{q_{i}} k_{q_{i}}) + \alpha_{q_{z+1}-1} (B - s_{q_{z+1}-1})
	$. 
	Hence, 
	$\frac{ \sum_{i = 1}^{n(\sigma')} \alpha_{q_{i}(\sigma')} k_{q_{i}(\sigma')}(\sigma') }{ V_{PQ}(\sigma') }
		\geq \frac{ (\sum_{i = 1}^{n} \alpha_{q_{i}} k_{q_{i}}) + \alpha_{q_{z+1}-1} (B - s_{q_{z+1}-1}) }
			{ V_{PQ}(\sigma) + \alpha_{q_{z+1}-1} (B - s_{q_{z+1}-1}) }
		\geq \frac{ \sum_{i = 1}^{n} \alpha_{q_{i}} k_{q_{i}} }{ V_{PQ}(\sigma) }
	$. 
	Therefore, 
	$\frac{V_{OPT}(\sigma')}{V_{PQ}(\sigma')} 
		\geq \frac{ V_{PQ}(\sigma') + \sum_{i = 1}^{n(\sigma')} \alpha_{q_{i}(\sigma')} k_{q_{i}(\sigma')}(\sigma') }
					{ V_{PQ}(\sigma') }
		\geq 1 + \frac{ \sum_{i = 1}^{n} \alpha_{q_{i}} k_{q_{i}} }{ V_{PQ}(\sigma) }
		= \frac{V_{OPT}(\sigma)}{V_{PQ}(\sigma)} 
	$. 
	By the definition of $\sigma'$, 
	${\cal S}_{2}$ includes $\sigma'$. 
	By the above argument, 
	for any $z'$ such that $q_{z'} + 1 < q_{z'+1}$, 
	we recursively construct an input in the above way, and then we can obtain an input satisfying the lemma. 
	\fi
\end{proof}
\begin{figure*}[h]
 \begin{center}
  \includegraphics[width=120mm]{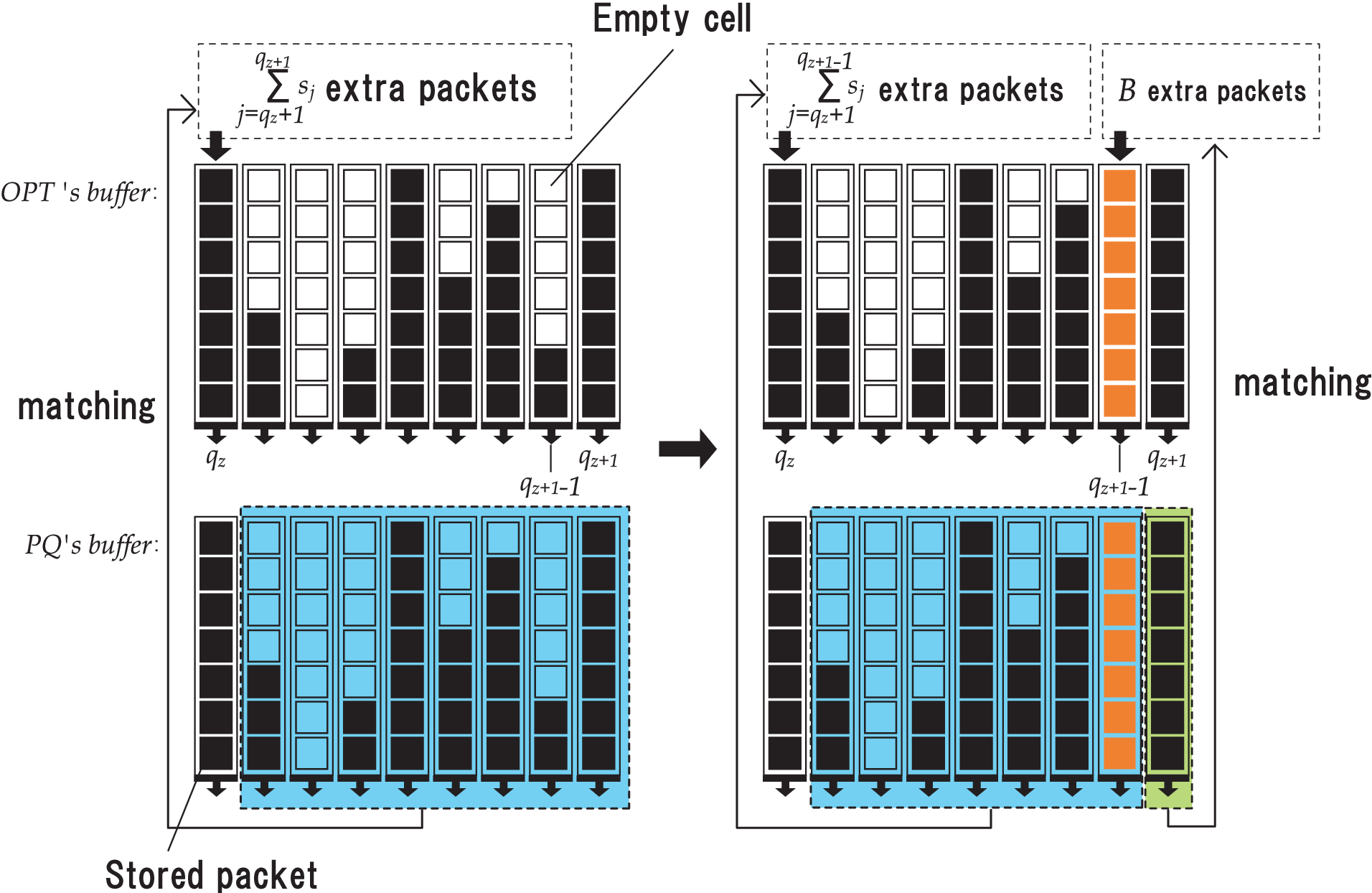}
 \end{center}
 \caption{
 			Example states of queues ($q_{z}$ through $q_{z+1}$) of $OPT$ and $PQ$ for $\sigma$ and $\sigma'$. 
			Left (Right) queues show the states for $\sigma$ ($\sigma'$). 
		}
\label{fig:L37}
 \end{figure*}
%
\fi
%

%
\ifnum \count10 > 0
入力の集合${\cal S}_{3}$を定義する。
${\cal S}_{3}$は、以下が成立する様な任意の入力$\sigma (\in {\cal S}_{2})$の集合を表す：
(i) 
各$i (\in [1, n-1])$に対して、
$q_{i} + 1 = q_{i+1}$が成立し、
(ii) 
各$i (\in [1, n-1])$に対して、
$k_{q_{i}} = B$が成立し、
(iii) 
任意の$j (\in [q_{1}, q_{n}])$に対して、
$s_{j} = B$が成立し、
\com{（■S)}
(iv)
$q_{1} - 1 \geq 1$が成立するならば、
任意の$j (\in [1, q_{1}-1])$に対して、$s_{j} = 0$が成立する。
\com{（■S)}
(v)
各$j (\in [q_{n}+1, m])$に対して、
$s_{j} \leq B$
が成立する。
(補題~\ref{LMA:3.2.1}より、
$q_{n}+1 \leq m$が成立する。)
\fi
\ifnum \count11 > 0
\com{（■英語）}
We define the set ${\cal S}_{3}$ of inputs. 
${\cal S}_{3}$ denotes the set of inputs $\sigma (\in {\cal S}_{2})$ such that 
(i) 
for each $i (\in [1, n-1])$, 
$q_{i} + 1 = q_{i+1}$, 
(ii) 
for each $i (\in [1, n-1])$, 
$k_{q_{i}} = B$, 
(iii) 
for each $j (\in [q_{1}, q_{n}])$, 
$s_{j} = B$, 
(iv)
for any $j (\in [1, q_{1}-1])$, $s_{j} = 0$ 
if $q_{1} - 1 \geq 1$, 
and 
\com{（■S)}
(v)
for each $j (\in [q_{n}+1, m])$, 
$s_{j} \leq B$. 
(By Lemma~\ref{LMA:3.2.1}, $q_{n}+1 \leq m$.)
\fi
%

%
\begin{LMA} \label{LMA:3.2.7}
	\ifnum \count10 > 0
	任意の入力$\sigma (\in {\cal S}_{3})$に対して、
	次の様な入力$\sigma' (\in {\cal S}_{3})$が存在する。
	(i) $u = \lfloor \frac{ \sum_{j = q_{n(\sigma)}(\sigma)+1}^{m} s_{j}(\sigma) }{B} \rfloor$と定義し、
		任意の$j (\in [q_{n(\sigma)}(\sigma), q_{n(\sigma)}(\sigma)+u])$
		に対して、
		$s_{j}(\sigma') = B$
		かつ
		$s_{q_{n(\sigma)}(\sigma)+u+1}(\sigma') = (\sum_{j = q_{n(\sigma)}(\sigma)+1}^{m} s_{j}(\sigma)) - u B$
		が成立する。
	(ii) $\frac{V_{OPT}(\sigma)}{V_{PQ}(\sigma)} \leq \frac{V_{OPT}(\sigma')}{V_{PQ}(\sigma')}$が成立する。
	\fi
	\ifnum \count11 > 0
	\com{（■英語）}
	For any input $\sigma (\in {\cal S}_{3})$, 
	there exists an input $\sigma' (\in {\cal S}_{3})$ such that 
	(i) 
	$s_{q_{n(\sigma)}(\sigma)+u+1}(\sigma') = (\sum_{j = q_{n(\sigma)}(\sigma)+1}^{m} s_{j}(\sigma)) - u B$, 
	where $u = \lfloor \frac{ \sum_{j = q_{n(\sigma)}(\sigma)+1}^{m} s_{j}(\sigma) }{B} \rfloor$, 
	and 
	for any $j (\in [q_{n(\sigma)}(\sigma), q_{n(\sigma)}(\sigma)+u])$, 
	$s_{j}(\sigma') = B$, and 
	(ii) 
	$\frac{V_{OPT}(\sigma)}{V_{PQ}(\sigma)} \leq \frac{V_{OPT}(\sigma')}{V_{PQ}(\sigma')}$. 
	\fi
\end{LMA}
%
\ifnum \count14 < 1
%
\begin{proof}
	\ifnum \count10 > 0
	%
	%
	%
	任意の$j (\in [1, q_{n}])$に対して、
	$s''_{j} = s_{j}$
	と定義する。
	また、
	各$j (\in [q_{n}+1, q_{n}+u])$に対して、
	$s''_{j} = B$
	かつ
	$s''_{q_{n}+u+1} = (\sum_{j = q_{n}+1}^{m} s_{j}) - u B$	
	と定義する。
	次の様に$\sigma$から入力$\sigma'$を構成する。
	この構成の仕方は、補題~\ref{LMA:3.2.5}と補題~\ref{LMA:3.2.6}に類似している。
	まず、
	時間$(0, 1)$の間に各$Q^{(j)} \hspace{1mm} (j \in [q_{1}, m])$にpacketが到着する$s''_{j}$回のarrival eventが発生する。
	定義より、
	$s''_{j} \leq B$が成立するので、
	$PQ$はそれらの到着するpacketを全て受理する。
	ここで、
	任意の$i (\in [1, n])$に対して、
	$a_{i} = \sum_{j = q_{n+1-i}+1}^{q_{n+2-i}} s''_{j}$と定義し、
	$a_{0} = 0$
	と定義する。
	任意の$x (\in [0, n-1])$に対して、
	整数時間$t = (\sum_{j = 0}^{x} a_{j} ) + 1, \ldots, \sum_{j = 0}^{x+1} a_{j}$に、
	scheduling eventが発生し、
	時刻$t + \frac{1}{2}$において、
	$Q^{(q_{n-x})}$にpacketが到着する様なarrival eventが発生する。
	時刻$(\sum_{j = 0}^{n} a_{j})+1$の後は、
	十分な数のscheduling eventが発生する。
	明らかに、
	$V_{PQ}(\sigma') \leq V_{PQ}(\sigma)$
	と
	$V_{OPT}(\sigma') = V_{OPT}(\sigma)$
	が成立する。
	また、
	$\sigma'$の定義より、
	$\sigma' \in {\cal S}_{3}$が成立し、
	$\sigma'$はステートメントの条件(i)をみたす。
	\fi
	\ifnum \count11 > 0
	\com{（■英語）}
	For any $j (\in [1, q_{n}])$, 
	we define $s''_{j} = s_{j}$. 
	Furthermore, 
	for each $j (\in [q_{n}+1, q_{n}+u])$, 
	we define $s''_{j} = B$, and 
	$s''_{q_{n}+u+1} = (\sum_{j = q_{n}+1}^{m} s_{j}) - u B$. 
	We construct $\sigma'$ from $\sigma$ in the following way. 
	This approach is similar to those in the proof of Lemmas~\ref{LMA:3.2.5} and  \ref{LMA:3.2.6}. 
	First, 
	for each $j (\in [q_{1}, m])$, 
	$s''_{j}$ events at which $s_{j}$ packets arrive at $Q^{(j)}$ occur during time $(0, 1)$. 
	Since $s''_{j} \leq B$ by definition, 
	$PQ$ accepts all these packets. 
	Then, 
	for any $i (\in [1, n])$, 
	we define $a_{i} = \sum_{j = q_{n+1-i}+1}^{q_{n+2-i}} s_{j}$, and 
	$a_{0} = 0$. 
	For any $x (\in [0, n-1])$, 
	a scheduling event occurs at each integer time $t = (\sum_{j = 0}^{x} a_{j} ) + 1, \ldots, \sum_{j = 0}^{x+1} a_{j}$. 
	Also, 
	at each time $t + \frac{1}{2}$, 
	an arrival event where a packet arrives at $Q^{(q_{n-x})}$ occurs. 
	After time $(\sum_{j = 0}^{n} a_{j}) + 1$, 
	sufficient scheduling events to transmit all the arriving packets occur. 
	It is easy to see that 
	$V_{PQ}(\sigma') \leq V_{PQ}(\sigma)$
	and 
	$V_{OPT}(\sigma') = V_{OPT}(\sigma)$. 
	Moreover, 
	by the definition of $\sigma'$, 
	$\sigma' \in {\cal S}_{3}$ holds, and $\sigma'$ satisfies the condition (i) in the statement. 
	\fi
\end{proof}
%
\fi
%

%
\ifnum \count10 > 0
入力の集合${\cal S}_{4}$を定義する。
${\cal S}_{4}$は、以下が成立する様な任意の入力$\sigma (\in {\cal S}_{3})$の集合を表す：
(i) 
各$i (\in [1, n-1])$に対して、
$q_{i} + 1 = q_{i+1}$が成立し、
(ii) 
各$i (\in [1, n-1])$に対して、
$k_{q_{i}} = B$が成立し、
(iii) 
任意の$j (\in [q_{1}, q_{n}])$に対して、
$s_{j} = B$が成立し、
(iv)
$q_{1} - 1 \geq 1$が成立するならば、
任意の$j (\in [1, q_{1}-1])$に対して、$s_{j} = 0$が成立する。
(v)
ある $u \hspace{1mm} (0 \leq u \leq m - q_{n}-1)$が存在して、
任意の$j (\in [q_{n}, q_{n}+u])$に対して、
$s_{j} = B$
かつ
$B \geq s_{q_{n}+u+1} \geq 1$
かつ
$q_{n}+u+2 \leq m$が成立する場合、
任意の$j (\in [q_{n}+u+2, m])$に対して、
$s_{j} = 0$
が成立する。
\fi
\ifnum \count11 > 0
\com{（■英語）}
We next introduce the set ${\cal S}_{4}$ of inputs. 
Let ${\cal S}_{4}$ denote the set of inputs $\sigma (\in {\cal S}_{3})$ satisfying the following five conditions: 
(i) 
for each $i (\in [1, n-1])$, 
$q_{i} + 1 = q_{i+1}$, 
(ii) 
for each $i (\in [1, n-1])$, 
$k_{q_{i}} = B$, 
(iii) 
for each $j (\in [q_{1}, q_{n}])$, 
$s_{j} = B$, 
(iv)
for any $j (\in [1, q_{1}-1])$, $s_{j} = 0$ 
if $q_{1} - 1 \geq 1$, and 
(v)
there exists some $u$ such that $0 \leq u \leq m - q_{n}-1$. 
Also, 
for any $j (\in [q_{n}, q_{n}+u])$, 
$s_{j} = B$,  
$B \geq s_{q_{n}+u+1} \geq 1$, and 
for any $j (\in [q_{n}+u+2, m])$, 
$s_{j} = 0$
if $q_{n}+u+2 \leq m$. 
\fi
%

%
\begin{LMA} \label{LMA:3.2.8}
	\ifnum \count10 > 0
	任意の入力$\sigma (\in {\cal S}_{4})$に対して、
	$q_{n(\sigma)}(\sigma) + 2 \leq m$
	かつ
	$s_{q_{n(\sigma)}(\sigma) + 1}(\sigma) = B$
	かつ
	$\sum_{ j = q_{n(\sigma)}(\sigma) + 2}^{m} s_{j}(\sigma) > 0$
	が成立するとする。
	このとき、次の様な入力$\hat{\sigma} (\in {\cal S}_{4})$が存在する。
	(i) $n(\hat{\sigma}) = n(\sigma) + 1$が成立する。
	(ii) 各$i (\in [1, n(\hat{\sigma})-1])$に対して、
		$q_{i}(\hat{\sigma}) = q_{i}(\sigma)$
		かつ
		$q_{n(\hat{\sigma})}(\hat{\sigma}) = q_{n(\sigma)}(\sigma) + 1$
		が成立する。
	(iii) $\frac{V_{OPT}(\sigma)}{V_{PQ}(\sigma)} \leq \frac{V_{OPT}(\hat{\sigma})}{V_{PQ}(\hat{\sigma})}$が成立する。
	\fi
	\ifnum \count11 > 0
	\com{（■英語）}
	Let $\sigma (\in {\cal S}_{4})$ be an input such that 
	$q_{n(\sigma)}(\sigma) + 2 \leq m$, 
	$s_{q_{n(\sigma)}(\sigma) + 1}(\sigma) = B$, and 
	$\sum_{ j = q_{n(\sigma)}(\sigma) + 2}^{m} s_{j}(\sigma) > 0$. 
	Then, 
	there exists an input $\hat{\sigma} (\in {\cal S}_{4})$ such that 
	(i) $n(\hat{\sigma}) = n(\sigma) + 1$, 
	(ii) for each $i (\in [1, n(\hat{\sigma})-1])$, 
		$q_{i}(\hat{\sigma}) = q_{i}(\sigma)$, 
		and 
		$q_{n(\hat{\sigma})}(\hat{\sigma}) = q_{n(\sigma)}(\sigma) + 1$, and 
	(iii) $\frac{V_{OPT}(\sigma)}{V_{PQ}(\sigma)} \leq \frac{V_{OPT}(\hat{\sigma})}{V_{PQ}(\hat{\sigma})}$. 
	\fi
\end{LMA}
%
\ifnum \count14 < 1
%
\begin{proof}
	\ifnum \count10 > 0
	%
	%
%
	%
	$\sigma$から次の様な入力$\sigma'$を構成する。
	まず、
	時間$(0, 1)$の間に各$Q^{(j)} \hspace{1mm} (j \in [q_{1}, m])$にpacketが到着する$s_{j}$回のarrival eventが発生する。
	${\cal S}_{4}$の定義より、
	$s_{j} \leq B$が成立するので、$PQ$はそれらの到着するpacketを全て受理する。
	ここで、
	任意の$i (\in [1, n])$に対して、
	$q'_{i} = q_{i}$と定義し、
	$q'_{n+1} = q_{n}+1$と定義し、
	$q'_{n+2} = m$と定義する。
	任意の$i (\in [1, n+1])$に対して、
	$a_{i} = \sum_{j = q'_{n+2-i}+1}^{q'_{n+3-i}} s_{j}$と定義し、
	$a_{0} = 0$と定義する。
	任意の$x (\in [0, n])$に対して、
	整数時間$t = (\sum_{j = 0}^{x} a_{j} ) + 1, \ldots, \sum_{j = 0}^{x+1} a_{j}$に、
	scheduling eventが発生する。
	更に、
	任意の$x (\in [0, n])$に対して、
	時刻$t + \frac{1}{2}$において、
	$Q^{(q'_{n+1-x})}$にpacketが到着する様なarrival eventが発生する。
	時刻$(\sum_{j = 0}^{n+1} a_{j})+1$の後は、
	十分な数のscheduling eventが発生する。
	このとき、
	$PQ$が$\sigma'$において各scheduling eventにtransmitするpacketは、
	$\sigma$と全く同じである。
	また、
	$t$に$Q^{(q'_{n+1-x})}$からpacketをtransmitする
	オフラインアルゴリズム$ALG$を考える。
	$q'_{i}$の定義より、
	任意の$i (\in [1, n+1])$に対して、
	$h_{PQ}^{(q'_{i})}(1-) = B$が成立しているので、
	$PQ$は各時刻$t + \frac{1}{2}$において、到着する全てのpacketをacceptできない。
	しかし、$ALG$は全てのpacketを受理することが出来る。
	すなわち、$ALG$は最適なオフラインアルゴリズムの1つである。
	よって、
	$n(\sigma') = n+1$
	かつ
	任意の$i (\in [1, n+1])$に対して、
	$q_{i}(\sigma') = q'_{i}$
	が成立する。
	明らかに、任意の$j (\in [1, m])$に対して、
	$s_{j}(\sigma') = s_{j}$
	が成立しており、
	$V_{PQ}(\sigma') = V_{PQ}(\sigma)$
	が成立する。
	また、
	任意の$i (\in [1, n-1])$に対して、
	$k_{q_{i}}(\sigma') = k_{q_{i}}$
	が成立し、
	$k_{q_{n}}(\sigma') = s_{q_{n}+1}$
	が成立し、
	$k_{q_{n}+1}(\sigma') = \sum_{j = q_{n}+2}^{m} s_{j}$
	\com{（■）}
	が成立する。
	よって、
	$\sigma' \in {\cal S}_{4}$
	が成立し、
	$\sigma'$はステートメントの条件(i)と(ii)をみたす。
	更に、
	$V_{OPT}(\sigma') = V_{ALG}(\sigma') 
		= V_{OPT}(\sigma) + (\alpha_{q_{n}+1} - \alpha_{q_{n}}) \sum_{j = q_{n}+2}^{m} s_{j}
		\geq V_{OPT}(\sigma)$
	が成立する。
	\fi
	\ifnum \count11 > 0
	\com{（■英語）}
	%
	%
%
	%
	We construct $\sigma'$ from $\sigma$ as follows: 
	First, 
	for each $j (\in [q_{1}, m])$, 
	$s_{j}$ events at which $s_{j}$ packets at $Q^{(j)}$ arrive occur during time $(0, 1)$. 
	Since $s_{j} \leq B$ by the definition of ${\cal S}_{4}$, 
	$PQ$ accepts all these arriving packets. 
	For any $i (\in [1, n])$, 
	we define $q'_{i} = q_{i}$,  
	$q'_{n+1} = q_{n}+1$ and 
	$q'_{n+2} = m$. 
	Moreover, 
	for any $i (\in [1, n+1])$, 
	we define $a_{i} = \sum_{j = q'_{n+2-i}+1}^{q'_{n+3-i}} s_{j}$ and 
	$a_{0} = 0$. 
	Then, 
	for any $x (\in [0, n])$, 
	a scheduling event occurs at each integer time $t = (\sum_{j = 0}^{x} a_{j} ) + 1, \ldots, \sum_{j = 0}^{x+1} a_{j}$. 
	In addition, 
	for any $x (\in [0, n])$, 
	an arrival event where a packet arrives at $Q^{(q'_{n+1-x})}$ occurs at each time $t + \frac{1}{2}$. 
	After time $(\sum_{j = 0}^{n+1} a_{j}) + 1$, 
	sufficient scheduling events to transmit all the arriving packets occur. 
	Then, 
	the packets which $PQ$ transmits at each scheduling event for $\sigma'$ are equivalent to those for $\sigma$. 
	Consider an offline algorithm $ALG$ which transmits a packet from $Q^{(q'_{n+1-x})}$ at $t$. 
	By the definition of $q'_{i}$, 
	since for any $i (\in [1, n+1])$, $h_{PQ}^{(q'_{i})}(1-) = B$, 
	$PQ$ cannot accept any packet which arrives at each time $t + \frac{1}{2}$,  
	but $ALG$ can accept all the packets, which means that  
	$ALG$ is optimal. 
	Hence, 
	$n(\sigma') = n+1$, and 
	for any $i (\in [1, n+1])$, 
	$q_{i}(\sigma') = q'_{i}$. 
	Since for any $j (\in [1, m])$, $s_{j}(\sigma') = s_{j}$, 
	$V_{PQ}(\sigma') = V_{PQ}(\sigma)$. 
	Moreover, 
	for any $i (\in [1, n-1])$, 
	$k_{q_{i}}(\sigma') = k_{q_{i}}$, 
	$k_{q_{n}}(\sigma') = s_{q_{n}+1}$, and 
	$k_{q_{n}+1}(\sigma') = \sum_{j = q_{n}+2}^{m} s_{j}$. 
	Therefore, 
	$\sigma' \in {\cal S}_{4}$ holds, 
	$\sigma'$ satisfies the conditions (i) and (ii) in the statements. 
	Also, 
	$V_{OPT}(\sigma') = V_{ALG}(\sigma') 
		= V_{OPT}(\sigma) + (\alpha_{q_{n}+1} - \alpha_{q_{n}}) \sum_{j = q_{n}+2}^{m} s_{j}
		\geq V_{OPT}(\sigma)$. 
	\fi
\end{proof}
%
\fi
%

%
\ifnum \count10 > 0
入力の集合${\cal S}_{5}$を定義する。
${\cal S}_{5}$は、以下が成立する様な任意の入力$\sigma (\in {\cal S}_{4})$の集合を表す：
(i) 
各$i (\in [1, n-1])$に対して、
$q_{i} + 1 = q_{i+1}$が成立し、
(ii) 
各$i (\in [1, n-1])$に対して、
$k_{q_{i}} = B$が成立し、
(iii) 
任意の$j (\in [q_{1}, q_{n}])$に対して、
$s_{j} = B$が成立し、
(iv)
\com{（■S)} 
$q_{1} - 1 \geq 1$が成立するならば、
任意の$j (\in [1, q_{1}-1])$に対して、$s_{j} = 0$が成立する、
(v) 
$k_{q_{n}} = s_{q_{n}+1}$
(補題~\ref{LMA:3.2.1}より、
$q_{n}+1 \leq m$が成立する。)
かつ
$1 \leq s_{q_{n}+1} \leq B$が成立する、
(vi) 
$q_{n} + 2 \leq m$が成立するならば、
任意の$j (\in [q_{n}+2, m])$に対して、$s_{j} = 0$が成立する。
\fi
\ifnum \count11 > 0
\com{（■英語）}
%
%
${\cal S}_{5}$ denotes the set of inputs $\sigma (\in {\cal S}_{4})$ satisfying the following six conditions: 
(i) 
for each $i (\in [1, n-1])$, 
$q_{i} + 1 = q_{i+1}$, 
(ii) 
for each $i (\in [1, n-1])$, 
$k_{q_{i}} = B$, 
(iii) 
for each $j (\in [q_{1}, q_{n}])$, 
$s_{j} = B$, 
(iv) 
\com{（■S)}
for any $j (\in [1, q_{1}-1])$, $s_{j} = 0$ holds 
if $q_{1} -1 \geq 1$, 
(v) 
$k_{q_{n}} = s_{q_{n}+1}$ 
(By Lemma~\ref{LMA:3.2.1}, $q_{n}+1 \leq m$.) and 
$1 \leq s_{q_{n}+1} \leq B$, and 
(vi) 
for any $j (\in [q_{n}+2, m])$, $s_{j} = 0$ holds 
if $q_{n} + 2 \leq m$. 
\fi
%

%
\begin{LMA} \label{LMA:3.2.9}
	\ifnum \count10 > 0
	任意の入力$\sigma (\in {\cal S}_{5})$に対して、
	次の様な入力$\hat{\sigma} (\in {\cal S}_{5})$が存在する。
	(i) $s_{q_{n(\hat{\sigma})}(\hat{\sigma})+1}(\hat{\sigma}) = B$が成立する。
	(ii) $\frac{V_{OPT}(\sigma)}{V_{PQ}(\sigma)} \leq \frac{V_{OPT}(\hat{\sigma})}{V_{PQ}(\hat{\sigma})}$が成立する。
	すなわち、
	ある入力$\sigma^{*} \in {\cal S}^{*}$が存在して、
	$\max_{\sigma'} \{ \frac{V_{OPT}(\sigma')}{V_{PQ}(\sigma')} \} = \frac{V_{OPT}(\sigma^{*})}{V_{PQ}(\sigma^{*})}$
	が成立する。
	\fi
	\ifnum \count11 > 0
	\com{（■英語）}
	For any input $\sigma (\in {\cal S}_{5})$, 
	there exists an input $\hat{\sigma} (\in {\cal S}_{5})$ such that 
	(i) $s_{q_{n(\hat{\sigma})}(\hat{\sigma})+1}(\hat{\sigma}) = B$, and 
	(ii) $\frac{V_{OPT}(\sigma)}{V_{PQ}(\sigma)} \leq \frac{V_{OPT}(\hat{\sigma})}{V_{PQ}(\hat{\sigma})}$. 
	That is, 
	there exists an input $\sigma^{*} \in {\cal S}^{*}$ such that 
	$\max_{\sigma'} \{ \frac{V_{OPT}(\sigma')}{V_{PQ}(\sigma')} \} = \frac{V_{OPT}(\sigma^{*})}{V_{PQ}(\sigma^{*})}$. 
	\fi
\end{LMA}
%
%
\begin{proof}
	\ifnum \count10 > 0
	%
	%
%
	%
	$\sigma \in {\cal S}_{5}$が成立するので、
	$
	\frac{V_{OPT}(\sigma)}{V_{PQ}(\sigma)} 
		= \frac{ V_{PQ}(\sigma) + \sum_{i = 1}^{n} \alpha_{q_{i}} k_{q_{i}} }{ V_{PQ}(\sigma) }
		\leq 1 + \frac{ B (\sum_{j = q_{1} }^{q_{n}-1} \alpha_{j}) + \alpha_{q_{n}} s_{q_{n}+1} }{ \sum_{j = q_{1}}^{q_{n}+1} \alpha_{j} s_{j} }
		\leq 1 + \frac{ B (\sum_{j = q_{1} }^{q_{n}-1} \alpha_{j}) + \alpha_{q_{n}} s_{q_{n}+1} }{ B (\sum_{j = q_{1} }^{q_{n}} \alpha_{j} ) + \alpha_{q_{n}+1} s_{q_{n}+1} }
	$
	（この値を$x(s_{q_{n}+1})$と定義する。）
	が成立する。
	ここで、
	$\sigma_{1}, \sigma_{2} \in {\cal S}_{5}$を次が成立する様な任意の入力とする。
	(i) 
	$n = n(\sigma_{2}) = n(\sigma_{1}) + 1$が成立し、
	(ii) 
	任意の$i (\in [1, n-1])$に対して、
	$q_{i} = q_{i}(\sigma_{1}) = q_{i}(\sigma_{2})$
	が成立し、
	(iii)
	$q_{n} = q_{n}(\sigma_{2})$
	が成立し、
	(iv) 
	$s_{q_{n-1}+1}(\sigma_{1}) = B$
	かつ
	$s_{q_{n}+1}(\sigma_{2}) = B$
	が成立する。
	このとき、
	$x(s_{q_{n}+1})$は$s_{q_{n}+1}$に関して単調増加、もしくは単調減少であるので、
	$\frac{V_{OPT}(\sigma)}{V_{PQ}(\sigma)} \leq \max\{ \frac{V_{OPT}(\sigma_{1})}{V_{PQ}(\sigma_{1})}, \frac{V_{OPT}(\sigma_{2})}{V_{PQ}(\sigma_{2})} \}$
	が成立する。
	よって、
	入力$\hat{\sigma}$を、
	$\hat{\sigma} \in \arg \max\{ \frac{V_{OPT}(\sigma_{1})}{V_{PQ}(\sigma_{1})}, \frac{V_{OPT}(\sigma_{2})}{V_{PQ}(\sigma_{2})} \}$
	が成立する様な入力とすれば、
	題意が満たされる。
	\fi
	\ifnum \count11 > 0
	\com{（■英語）}
	%
%
	%
	Since $\sigma \in {\cal S}_{5}$ holds, 
	$
	\frac{V_{OPT}(\sigma)}{V_{PQ}(\sigma)} 
		= \frac{ V_{PQ}(\sigma) + \sum_{i = 1}^{n} \alpha_{q_{i}} k_{q_{i}} }{ V_{PQ}(\sigma) }
		\leq 1 + \frac{ B (\sum_{j = q_{1} }^{q_{n}-1} \alpha_{j}) + \alpha_{q_{n}} s_{q_{n}+1} }{ \sum_{j = q_{1}}^{q_{n}+1} \alpha_{j} s_{j} }
		\leq 1 + \frac{ B (\sum_{j = q_{1} }^{q_{n}-1} \alpha_{j}) + \alpha_{q_{n}} s_{q_{n}+1} }{ B (\sum_{j = q_{1} }^{q_{n}} \alpha_{j} ) + \alpha_{q_{n}+1} s_{q_{n}+1} }
	$, 
	which we define as $x(s_{q_{n}+1})$. 
	Let $\sigma_{1}, \sigma_{2} \in {\cal S}_{5}$ be any inputs such that	%
	(i) 
	$n = n(\sigma_{2}) = n(\sigma_{1}) + 1$, 
	(ii) 
	for any $i (\in [1, n-1])$, 
	$q_{i} = q_{i}(\sigma_{1}) = q_{i}(\sigma_{2})$, 
	(iii)
	$q_{n} = q_{n}(\sigma_{2})$, and 
	(iv) 
	$s_{q_{n-1}+1}(\sigma_{1}) = B$ and 
	$s_{q_{n}+1}(\sigma_{2}) = B$. 
	Then, 
	since $x(s_{q_{n}+1})$ is monotone (increasing or decreasing) as $s_{q_{n}+1}$ increases, 
	$\frac{V_{OPT}(\sigma)}{V_{PQ}(\sigma)} \leq \max\{ \frac{V_{OPT}(\sigma_{1})}{V_{PQ}(\sigma_{1})}, \frac{V_{OPT}(\sigma_{2})}{V_{PQ}(\sigma_{2})} \}$. 
	Therefore, 
	let $\hat{\sigma}$ be the input such that 
	$\hat{\sigma} \in \arg \max\{ \frac{V_{OPT}(\sigma_{1})}{V_{PQ}(\sigma_{1})}, \frac{V_{OPT}(\sigma_{2})}{V_{PQ}(\sigma_{2})} \}$, 
	which means that the statement is true. 
	\fi
\end{proof}
%
%

%
\begin{LMA}\label{LMA:lower}
	\ifnum \count10 > 0
	$PQ$の競合比は少なくとも
	$2 - \min_{x \in [1, m-1] } \{ \frac{ \alpha_{x+1} }{ \sum_{j = 1}^{x+1} \alpha_{j} } \}$
	が成立する。
	\fi
	\ifnum \count11 > 0
	The competitive ratio of $PQ$ is at least $2 - \min_{x \in [1, m-1] } \{ \frac{ \alpha_{x+1} }{ \sum_{j = 1}^{x+1} \alpha_{j} } \}$. 
	\fi
\end{LMA}
%
%
\begin{proof}
	\ifnum \count10 > 0
	（■まだ）
	\fi
	\ifnum \count11 > 0
	\com{（■英語）}
	Consider the following input $\sigma$. 
	Define $m' \in \arg \min_{x \in [1, m-1]} \{ \frac{ \alpha_{x+1} }{ \sum_{j = 1}^{x+1} \alpha_{j} } \}$. 
	Initially, 
	$(m'+1) B$ arrival events happen such that $B$ packets arrive at $Q^{(1)}$ to $Q^{(m'+1)}$. 
	Then, 
	for $k = 1, 2, \ldots, m'$, 
	the $k$th round consists of $B$ scheduling events followed by $B$ arrival events 
	in which all the $B$ packets arrive at $Q^{(m'-k+1)}$.
	For $\sigma$, $PQ$ transmits $B$ packets from $Q^{(m'-k+2)}$ at the $k$th
	round. 
	As a result, 
	$PQ$ cannot accept arriving packets in the $(k+1)$st round. 
	Hence, 
	${V}_{PQ}(\sigma) = B \sum_{j = 1}^{m'+1} \alpha_{j}$ holds. 
	On the other hand, 
	$OPT$ transmits $B$ packets from $Q^{(m'-k+1)}$ at the $k$th round, 
	and hence can accept all the arriving packets. 
	Thus, 
	${V}_{OPT}(\sigma) = 2B \sum_{j = 1}^{m'} \alpha_{j} + B \alpha_{m'+1}$. 
	Therefore,
	$\frac{{V}_{OPT}(\sigma)}{{V}_{PQ}(\sigma)} 
		= \frac{ 2 \sum_{j = 1}^{m'} \alpha_{j} + \alpha_{m'+1}}{ \sum_{j = 1}^{m'+1} \alpha_{j} }
		= 2 - \frac{ \alpha_{m'+1} }{ \sum_{j = 1}^{m'+1} \alpha_{j} }$. 
	(It is easy to see that $\sigma \in {\cal S}_{5}$.)
	\fi
\end{proof}
%
%

%
\section{Lower Bound for Deterministic Algorithms}\label{sec:5}
\ifnum \count10 > 0
本節では、
任意の決定性アルゴリズムの下限を示す。
We consider only online algorithms that transmit a packet at a scheduling event 
whenever their buffers are not empty. 
(Such algorithms are called {\em work-conserving}. See e.g.\ \cite{YA03}.) 

\fi
\ifnum \count11 > 0
\com{（■英語）}
In this section, we show a lower bound for any deterministic algorithm. 
We make an assumption that is well-known to have no effect on the analysis of the competitive ratio. 
We consider only online algorithms that transmit a packet at a scheduling event 
whenever their buffers are not empty. 
(Such algorithms are called {\em work-conserving}. See e.g.\ \cite{YA03}.) 
\fi

\begin{theorem}\label{thm:3}
	No deterministic online algorithm can achieve a competitive ratio smaller than 
	$1 + \frac{ \alpha^3 + \alpha^2 +  \alpha }{ \alpha^4 + 4 \alpha^3 + 3 \alpha^2 + 4 \alpha + 1 }$. 
	%
	%
	%
\end{theorem}
\ifnum \count14 < 1
\begin{proof}
	\ifnum \count10 > 0
	オンラインアルゴリズム$ON$を固定する。
	次の様な入力$\sigma$を敵対者が作成する。
	$\sigma(t)$を$\sigma$の時刻$t$までの冒頭部分とする。
	$OPT$は、
	この入力において到着するパケットを全て受理し、scheduleすることが可能であることに注意せよ。
	時刻$(0, 1)$に
	$2B$回のarrival eventが発生し、
	$Q^{(1)}$と$Q^{(m)}$に$B$個ずつpacketが到着する。
	時刻$(1, 2)$に
	$B$回のscheduling eventが発生する。
	この$\sigma(2)$にたいして、
	時刻$(1, 2)$に
	online algorihtm $ON$が
	$Q^{(1)}$から$B(1-x)$個、
	$Q^{(m)}$から$Bx$個ずつpacketをscheduleするとする。
	（図~\ref{fig:L1}参照）
	時刻$2$より後の時刻では、
	$Q^{(1)}$と$Q^{(m)}$を比較して、
	より$ON$が損をするキューにパケットが到着する。
	{\bf\boldmath Case 1. $\alpha x \geq 1 - x$が成立する場合:}
	時刻$(2, 3)$に
	$B$回のarrival eventが発生し、
	$Q^{(1)}$に$B$個のpacketが到着する。
	このとき、時刻$3$までに、
	$ON$は$(\alpha + 1 + 1 - x )B$の価値のパケットを受理している。
	更に、
	時刻$(3, 4)$に
	$B$回のscheduling eventが発生する。
	この$\sigma(4)$にたいして、
	時刻$(3, 4)$に
	$ON$が
	$Q^{(1)}$から$B(1-y)$個、
	$Q^{(m)}$から$By$個ずつpacketをscheduleするとする。
	（図~\ref{fig:L2}参照）
	時刻$4$より後の時刻では、
	先の場合と同様に、
	$Q^{(1)}$と$Q^{(m)}$を比較して、
	より$ON$が損をするキューにパケットが到着する。
	{\bf\boldmath Case 1.1. $\alpha (x+y) \geq 1 - y$が成立する場合:}
	時刻$(4, 5)$に
	$B$回のarrival eventが発生し、
	$Q^{(1)}$に$B$個のpacketが到着する。
	時刻$(5, 6)$に
	$2B$回のscheduling eventが発生する。
	この入力にたいして、
	$V_{ON}(\sigma) = ( \alpha + 1 + 1 - x + 1 - y )B$、
	$V_{OPT}(\sigma) = ( \alpha + 1 + 1 + 1 )B$が成立する。
	{\bf\boldmath Case 1.2. $\alpha (x+y) < 1 - y$が成立する場合:}
	時刻$(4, 5)$に
	$B$回のarrival eventが発生し、
	$Q^{(m)}$に$B$個のpacketが到着する。
	時刻$(5, 6)$に
	$2B$回のscheduling eventが発生する。
	この入力にたいして、
	$V_{ON}(\sigma) = ( \alpha + 1 + 1 - x + \alpha(x + y) )B$、
	$V_{OPT}(\sigma) = ( \alpha + 1 + 1 + \alpha )B$
	が成立する。
	{\bf\boldmath Case 2. $\alpha x < 1 - x$が成立する場合:}
	時刻$(2, 3)$に
	$B$回のarrival eventが発生し、
	$Q^{(m)}$に$B$個のpacketが到着する。
	このとき、時刻$3$までに、
	$ON$は$( \alpha + 1 + \alpha x )B$の価値のパケットを受理している。
	時刻$(3, 4)$に
	$B$回のscheduling eventが発生する。
	この$\sigma(4)$にたいして、
	時刻$(3, 4)$に
	$ON$が
	$Q^{(1)}$から$B(1-z)$個、
	$Q^{(m)}$から$Bz$個ずつpacketをscheduleするとする。
	（図~\ref{fig:L3}参照）
	時刻$4$より後の時刻では、
	先の場合と同様に、
	$Q^{(1)}$と$Q^{(m)}$を比較して、
	より$ON$が損をするキューにパケットが到着する。
	{\bf\boldmath Case 2.1. $\alpha z \geq 1 - x + 1 - z$が成立する場合:}
	時刻$(4, 5)$に
	$B$回のarrival eventが発生し、
	$Q^{(1)}$に$B$個のpacketが到着する。
	時刻$(5, 6)$に
	$2B$回のscheduling eventが発生する。
	この入力にたいして、
	$V_{ON}(\sigma) = ( \alpha + 1 + \alpha x + 1 - x + 1 - z )B$、
	$V_{OPT}(\sigma) = ( \alpha + 1 + \alpha + 1 )B$が成立する。
	{\bf\boldmath Case 2.2. $\alpha z < 1 - x + 1 - z$が成立する場合:}
	時刻$(4, 5)$に
	$B$回のarrival eventが発生し、
	$Q^{(m)}$に$B$個のpacketが到着する。
	時刻$(5, 6)$に
	$2B$回のscheduling eventが発生する。
	この入力にたいして、
	$V_{ON}(\sigma) = ( \alpha + 1 + \alpha x + \alpha z )B$、
	$V_{OPT}(\sigma) = ( \alpha + 1 + \alpha + \alpha )B$が成立する。
	上記の議論より、
	$c_{1}(x) = \min_{y} \max\{
				\frac{ \alpha + 1 + 1 + 1 }{ \alpha + 1 + 1 - x + 1 - y }, 
				\frac{\alpha + 1 + 1 + \alpha }{\alpha + 1 + 1 - x + \alpha(x + y)}
			\}
	$
	と定義し、
	$c_{2}(x) = \min_{z} \max\{
				\frac{ \alpha + 1 + \alpha + 1 }{ \alpha + 1 + \alpha x + 1 - x + 1 - z }, 
				\frac{ \alpha + 1 + \alpha + \alpha }{ \alpha + 1 + \alpha x + \alpha z }
			\}
	$
	と定義して、
	$\frac{V_{OPT}(\sigma)}{V_{ON}(\sigma)} 
		\geq
			\min_{x} \max\{
				c_{1}(x), c_{2}(x)
			\}
	$
	が成立する。
	$c_{1}(x)$は、
	$\frac{ \alpha + 1 + 1 + 1 }{ \alpha + 1 + 1 - x + 1 - y } 
	= 
	\frac{\alpha + 1 + 1 + \alpha }{\alpha + 1 + 1 - x + \alpha(x + y)}$
	が成立するときに最小化される。
	このとき、
	$y = \frac{ \alpha (\alpha + 3) + (-\alpha^2 -4 \alpha + 1)x }
				{ \alpha^2 + 5 \alpha + 2 }$
	が成立し、
	$c_{1}(x) 
		\geq \frac{ \alpha^2 + 5 \alpha + 2 }{ \alpha^2 + 4 \alpha + 2 - x }$
	が成立する。
	$c_{2}(x)$は、
	$\frac{ \alpha + 1 + \alpha + 1 }{ \alpha + 1 + \alpha x + 1 - x + 1 - z } 
	= 
	\frac{ \alpha + 1 + \alpha + \alpha }{ \alpha + 1 + \alpha x + \alpha z }$
	が成立するときに最小化される。
	このとき、
	$z = \frac{ \alpha^2 + 6 \alpha + 1 + (\alpha^2 -4 \alpha - 1)x }
				{ 2 \alpha^2 + 5 \alpha + 1 }$
	が成立し、
	$c_{2}(x) 
		\geq \frac{ 2 \alpha^2 + 5 \alpha + 1 }{ \alpha^2 + 4 \alpha + 1 + \alpha^2 x }$
	が成立する。
	$\min_{x} \max\{	c_{1}(x), c_{2}(x) \}$は、
	$c_{1}(x) = c_{2}(x)$、すなわち、
	$\frac{ \alpha^2 + 5 \alpha + 2 }{ \alpha^2 + 4 \alpha + 2 - x }
		= \frac{ 2 \alpha^2 + 5 \alpha + 1 }{ \alpha^2 + 4 \alpha + 1 + \alpha^2 x }$
	が成立するときに最小化される。
	このとき、
	$x = \frac{ \alpha^4 + 4 \alpha^3 + 2 \alpha^2 + \alpha }{ \alpha^4 + 5 \alpha^3 + 4 \alpha^2 +5  \alpha + 1} $
	が成立し、
	$\min_{x} \max\{	c_{1}(x), c_{2}(x) \}
		\geq \frac{ \alpha^4 + 5 \alpha^3 + 4 \alpha^2 + 5 \alpha + 1 }
				{ \alpha^4 + 4 \alpha^3 + 3 \alpha^2 + 4 \alpha + 1 }
			= 1 + \frac{ \alpha^3 + \alpha^2 +  \alpha }{ \alpha^4 + 4 \alpha^3 + 3 \alpha^2 + 4 \alpha + 1 }
				$
	が成立する。
	\fi
	\ifnum \count11 > 0
	\com{（■英語）}
	Fix an online algorithm $ON$. 
	Our adversary constructs the following input $\sigma$. 
	Let $\sigma(t)$ denote the prefix of the input $\sigma$ up to time $t$. 
	$OPT$ can accept and transmit all arriving packets in this input. 
	$2B$ arrival events occur during time $(0, 1)$, and 
	$B$ packets arrive at $Q^{(1)}$ and $Q^{(m)}$, respectively. 
	In addition, 
	$B$ scheduling events occur during time $(1, 2)$. 
	For $\sigma(2)$, 
	suppose that $ON$ transmits $B(1-x)$ packets and $Bx$ ones from $Q^{(1)}$ and $Q^{(m)}$, respectively. 
	(See Figure~\ref{fig:L1}.)
	After time 2, 
	our adversary selects one queue from $Q^{(1)}$ and $Q^{(m)}$, and makes some packets arrive at the queue. 
	{\bf\boldmath Case 1: If $\alpha x \geq 1 - x$:}
	$B$ arrival events occur during time $(2, 3)$, and 
	$B$ packets arrive at $Q^{(1)}$. 
	Then, 
	the total value of packets which $ON$ accepts by time $3$ is $(\alpha + 1 + 1 - x )B$. 
	Moreover, 
	$B$ scheduling events occur during time $(3, 4)$. 
	For $\sigma(4)$, 
	suppose that $ON$ transmits $B(1-y)$ packets and $By$ packets from $Q^{(1)}$ and $Q^{(m)}$, respectively. 
	(See Figure~\ref{fig:L2}.)
	After time 4, in the same way as time 2, 
	our adversary selects one queue from $Q^{(1)}$ and $Q^{(m)}$, and makes some packets arrive at the queue. 
	{\bf\boldmath Case 1.1: If $\alpha (x+y) \geq 1 - y$:}
	$B$ arrival events occur during time $(4, 5)$, and 
	$B$ packets arrive at $Q^{(1)}$. 
	Furthermore, 
	$2B$ scheduling events occur during time $(5, 6)$. 
	For this input,
	$V_{ON}(\sigma) = ( \alpha + 1 + 1 - x + 1 - y )B$, and 
	$V_{OPT}(\sigma) = ( \alpha + 1 + 1 + 1 )B$. 
	{\bf\boldmath Case 1.2: If $\alpha (x+y) < 1 - y$:}
	$B$ arrival events occur during time $(4, 5)$, and 
	$B$ packets arrive at $Q^{(m)}$. 
	Moreover, 
	$2B$ scheduling events occur during time $(5, 6)$. 
	For this input,
	$V_{ON}(\sigma) = ( \alpha + 1 + 1 - x + \alpha(x + y) )B$, and 
	$V_{OPT}(\sigma) = ( \alpha + 1 + 1 + \alpha )B$. 
	{\bf\boldmath Case 2: If $\alpha x < 1 - x$:}
	$B$ arrival events occur during time $(2, 3)$, and 
	$B$ packets arrive at $Q^{(m)}$. 
	Then, 
	the total value of packets which $ON$ accepts by time $3$ is $( \alpha + 1 + \alpha x )B$.  
	Moreover, 
	$B$ scheduling events occur during time $(3, 4)$. 
	For $\sigma(4)$, 
	$ON$ transmits $B(1-z)$ packets and $Bz$ ones from $Q^{(1)}$ and $Q^{(m)}$, respectively during time $(3, 4)$. 
	(See Figure~\ref{fig:L3}.)
	After time 4, in the same way as the above case, 
	our adversary selects one queue from $Q^{(1)}$ and $Q^{(m)}$, and causes some packets to arrive at the queue. 
	{\bf\boldmath Case 2.1: If $\alpha z \geq 1 - x + 1 - z$:}
	$B$ arrival events occur during time $(4, 5)$, and 
	$B$ packets arrive at $Q^{(1)}$. 
	Also,
	$2B$ scheduling events occur during time $(5, 6)$. 
	For this input, 
	$V_{ON}(\sigma) = ( \alpha + 1 + \alpha x + 1 - x + 1 - z )B$, and 
	$V_{OPT}(\sigma) = ( \alpha + 1 + \alpha + 1 )B$. 
	{\bf\boldmath Case 2.2: If $\alpha z < 1 - x + 1 - z$:}
	$B$ arrival events occur during time $(4, 5)$, and 
	$B$ packets arrive at $Q^{(m)}$. 
	In addition, 
	$2B$ scheduling events occur during time $(5, 6)$. 
	For this input, 
	$V_{ON}(\sigma) = ( \alpha + 1 + \alpha x + \alpha z )B$, and 
	$V_{OPT}(\sigma) = ( \alpha + 1 + \alpha + \alpha )B$. 
	By the above argument, 
	we define $c_{1}(x) = \min_{y} \max\{
				\frac{ \alpha + 1 + 1 + 1 }{ \alpha + 1 + 1 - x + 1 - y }, 
				\frac{\alpha + 1 + 1 + \alpha }{\alpha + 1 + 1 - x + \alpha(x + y)}
			\}
	$ and 
	$c_{2}(x) = \min_{z} \max\{
				\frac{ \alpha + 1 + \alpha + 1 }{ \alpha + 1 + \alpha x + 1 - x + 1 - z }, 
				\frac{ \alpha + 1 + \alpha + \alpha }{ \alpha + 1 + \alpha x + \alpha z }
			\}
	$. 
	Then, 
	$\frac{V_{OPT}(\sigma)}{V_{ON}(\sigma)} 
		\geq
			\min_{x} \max\{
				c_{1}(x), c_{2}(x)
			\}
	$. 
	$c_{1}(x)$ is minimized when 
	$\frac{ \alpha + 1 + 1 + 1 }{ \alpha + 1 + 1 - x + 1 - y } 
	= 
	\frac{\alpha + 1 + 1 + \alpha }{\alpha + 1 + 1 - x + \alpha(x + y)}$. 
	Then, 
	$y = \frac{ \alpha (\alpha + 3) + (-\alpha^2 -4 \alpha + 1)x }
				{ \alpha^2 + 5 \alpha + 2 }$.  
	Thus, 
	$c_{1}(x) 
		\geq \frac{ \alpha^2 + 5 \alpha + 2 }{ \alpha^2 + 4 \alpha + 2 - x }$. 
	$c_{2}(x)$ is minimized when 
	$\frac{ \alpha + 1 + \alpha + 1 }{ \alpha + 1 + \alpha x + 1 - x + 1 - z } 
	= 
	\frac{ \alpha + 1 + \alpha + \alpha }{ \alpha + 1 + \alpha x + \alpha z }$. 
	Then, 
	$z = \frac{ \alpha^2 + 6 \alpha + 1 + (\alpha^2 -4 \alpha - 1)x }
				{ 2 \alpha^2 + 5 \alpha + 1 }$. 
	Hence, 
	$c_{2}(x) 
		\geq \frac{ 2 \alpha^2 + 5 \alpha + 1 }{ \alpha^2 + 4 \alpha + 1 + \alpha^2 x }$. 
	Finally, 
	$\min_{x} \max\{	c_{1}(x), c_{2}(x) \}$ is minimized when 
	$c_{1}(x) = c_{2}(x)$, that is 
	$\frac{ \alpha^2 + 5 \alpha + 2 }{ \alpha^2 + 4 \alpha + 2 - x }
		= \frac{ 2 \alpha^2 + 5 \alpha + 1 }{ \alpha^2 + 4 \alpha + 1 + \alpha^2 x }$. 
	Therefore, 
	since $x = \frac{ \alpha^4 + 4 \alpha^3 + 2 \alpha^2 + \alpha }{ \alpha^4 + 5 \alpha^3 + 4 \alpha^2 +5  \alpha + 1}$, 
	$\min_{x} \max\{	c_{1}(x), c_{2}(x) \}
		\geq \frac{ \alpha^4 + 5 \alpha^3 + 4 \alpha^2 + 5 \alpha + 1 }
				{ \alpha^4 + 4 \alpha^3 + 3 \alpha^2 + 4 \alpha + 1 }
			= 1 + \frac{ \alpha^3 + \alpha^2 +  \alpha }{ \alpha^4 + 4 \alpha^3 + 3 \alpha^2 + 4 \alpha + 1 }
				$. 
	\fi
\end{proof}
\ifnum \count12 > 0
\begin{figure*}
 \begin{center}
  \includegraphics[width=150mm]{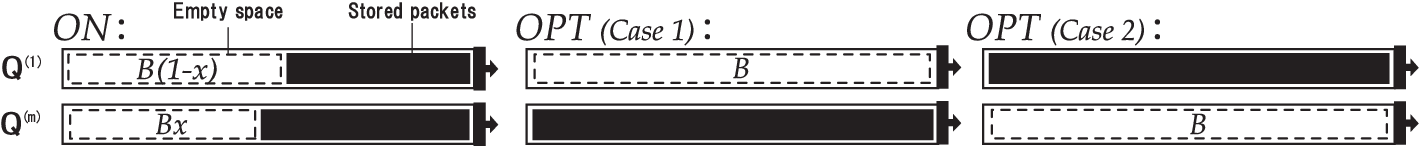}
 \end{center}
 \caption{States of queues at time 2}
\label{fig:L1}
 \end{figure*}

\begin{figure*}[t]
 \begin{center}
  \includegraphics[width=150mm]{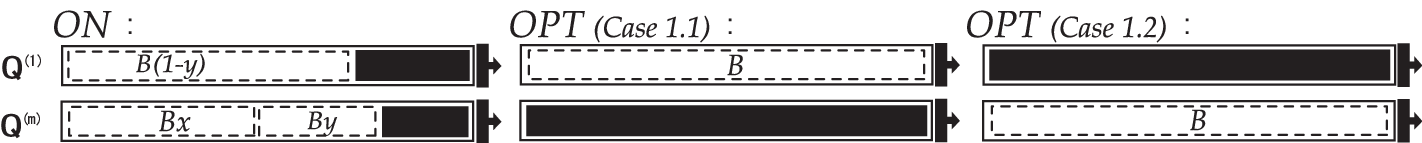}
 \end{center}
 \caption{States of queues at time 4 via Case 1}
 \label{fig:L2}
\end{figure*}
\begin{figure*}[t]
 \begin{center}
  \includegraphics[width=150mm]{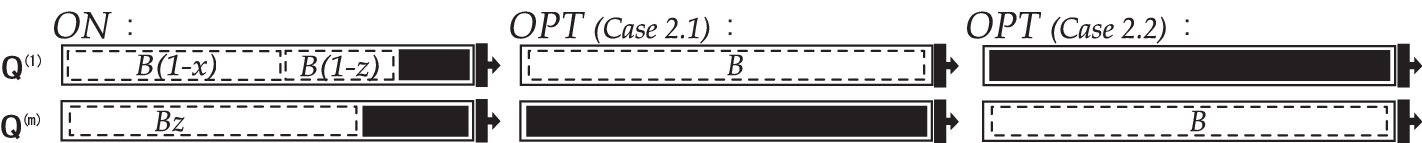}
 \end{center}
 \caption{States of queues at time 4 via Case 2}
 \label{fig:L3}
\end{figure*}
\fi
%
\fi

\section{Concluding Remarks}
A lot of packets used by multimedia applications arrive in a QoS switch at a burst, 
and managing queues to store outgoing packets (egress traffic) can become a bottleneck. 
In this paper, 
we have formulated the problem of controlling egress traffic, 
and analyzed Priority Queuing policies ($PQ$) using competitive analysis. 
We have shown that the competitive ratio of $PQ$ is exactly 
$2 - \min_{x \in [1, m-1] } \{ \frac{ \alpha_{x+1} }{ \sum_{j = 1}^{x+1} \alpha_{j} } \}$. 
Moreover, 
we have shown that there is no $1 + \frac{ \alpha^3 + \alpha^2 +  \alpha }{ \alpha^4 + 4 \alpha^3 + 3 \alpha^2 + 4 \alpha + 1 }$-competitive deterministic algorithm. 

We present some open questions as follows: 
(i) 
What is the competitive ratio of other practical policies, such as $WRR$? 
(ii) 
We consider the case where the size of each packet is one, namely fixed. 
In the setting where packets with variable sizes arrive, 
what is the competitive ratio of $PQ$ or other policies? 
(iii)
We are interested in comparing our results with experimental results using measured data in QoS switches. 
(iv)
The goal was to maximize the sum of the values of the transmitted packets in this paper, 
which is generally used for the online buffer management problems. 
However, this may not be able to evaluate the actual performance of practical scheduling algorithms correctly. 
(We showed that the worst scenario for $PQ$ is extreme in this paper.) 
What if another objective function (e.g., fairness) is used for evaluating the performance of a scheduling algorithm?
(v)
An obvious open question is to close the gap between the competitive ratio of $PQ$ and our lower bound for any deterministic algorithm. 
%


\newpage

\appendix

\section{Comparing Both Upper Bounds} \label{ap.sec:0}
\ifnum \count10 > 0
我々の結果：
$2 - \min_{ x \in [1, m-1] } \{ \frac{ \alpha_{x+1} }{ \sum_{j = 1}^{x+1} \alpha_{j} } \} 
	= 1 + \max_{ x \in [1, m-1] } \{ \frac{ \sum_{j = 1}^{x} \alpha_{j} }{ \sum_{j = 1}^{x+1} \alpha_{j} } \}$
と
Al-Bawani and Souza \cite{KA11}の結果：
$2 - \min_{ j \in [1, m-1]} \{ \frac{ \alpha_{j+1} - \alpha_{j} }{ \alpha_{j+1} } \} = 1 + \max_{j \in [1, m-1] } \{ \frac{ \alpha_{j} }{ \alpha_{j+1} }  \}$
が成立しているので、
以下では、
$\max_{ x \in [1, m-1] } \{ \frac{ \sum_{j = 1}^{x} \alpha_{j} }{ \sum_{j = 1}^{x+1} \alpha_{j} } \} < \max_{ j \in [1, m-1] } \{ \frac{ \alpha_{j} }{ \alpha_{j+1} }  \}$
を示す。
次の様に$a, b$を定義する。
$a \in \arg \max_{ j \in [1, m-1] } \{ \frac{ \alpha_{j} }{ \alpha_{j+1} } \}$、
$b \in \arg \max_{ x \in [1, m-1] } \{ \frac{ \sum_{j = 1}^{x} \alpha_{j} }{ \sum_{j = 1}^{x+1} \alpha_{j} } \}$。
このとき、
$\frac{ \alpha_{a} }{ \alpha_{a+1} } 
	\geq \frac{ \sum_{j = 1}^{b} \alpha_{j} }{ \sum_{j = 1}^{b} \alpha_{j+1} }
	> \frac{ \sum_{j = 1}^{b} \alpha_{j} }{ \alpha_{1} + \sum_{j = 1}^{b} \alpha_{j+1} }
	= \frac{ \sum_{j = 1}^{b} \alpha_{j} }{ \sum_{j = 1}^{b+1} \alpha_{j} }$
が得られる。
\fi
\ifnum \count11 > 0
\com{（■英語）}
Our upper bound is 
\[
	2 - \min_{ x \in [1, m-1] } \{ \frac{ \alpha_{x+1} }{ \sum_{j = 1}^{x+1} \alpha_{j} } \} 
		= 1 + \max_{ x \in [1, m-1] } \{ \frac{ \sum_{j = 1}^{x} \alpha_{j} }{ \sum_{j = 1}^{x+1} \alpha_{j} } \}
\]
 and 
the upper bound by Al-Bawani and Souza \cite{KA11} is 
\[
	2 - \min_{ j \in [1, m-1]} \{ \frac{ \alpha_{j+1} - \alpha_{j} }{ \alpha_{j+1} } \} = 1 + \max_{j \in [1, m-1] } \{ \frac{ \alpha_{j} }{ \alpha_{j+1} }  \}
.\] 
Now we show that 
\[
	\max_{ x \in [1, m-1] } \{ \frac{ \sum_{j = 1}^{x} \alpha_{j} }{ \sum_{j = 1}^{x+1} \alpha_{j} } \} 
		< \max_{ j \in [1, m-1] } \{ \frac{ \alpha_{j} }{ \alpha_{j+1} }  \}
.\]
Define 
$a \in \arg \max_{ j \in [1, m-1] } \{ \frac{ \alpha_{j} }{ \alpha_{j+1} } \}$ and 
$b \in \arg \max_{ x \in [1, m-1] } \{ \frac{ \sum_{j = 1}^{x} \alpha_{j} }{ \sum_{j = 1}^{x+1} \alpha_{j} } \}$. 
Then, we have that 
\[
	\frac{ \alpha_{a} }{ \alpha_{a+1} } 
		\geq \frac{ \sum_{j = 1}^{b} \alpha_{j} }{ \sum_{j = 1}^{b} \alpha_{j+1} }
		> \frac{ \sum_{j = 1}^{b} \alpha_{j} }{ \alpha_{1} + \sum_{j = 1}^{b} \alpha_{j+1} }
		= \frac{ \sum_{j = 1}^{b} \alpha_{j} }{ \sum_{j = 1}^{b+1} \alpha_{j} }
.\]
\fi
%


%
\section{Restriction of Input} \label{ap.sec:1}
\begin{LMA} \label{LMA:a.1}
	\ifnum \count10 > 0
	$\sigma$を$OPT$がパケットをrejectする様なarrival eventが発生する入力とする。
	このとき、
	次の様な入力$\sigma'$が存在する。
	$\frac{V_{OPT}(\sigma)}{V_{PQ}(\sigma)} \leq \frac{V_{OPT}(\sigma')}{V_{PQ}(\sigma')}$
	が成立し、
	$\sigma'$の全てのarrival eventにおいて、$OPT$はパケットをacceptする。
	\fi
	\ifnum \count11 > 0
	\com{（■英語）}
	Let $\sigma$ be an input such that $OPT$ rejects at least one packet at an arrival event. 
	Then, there exists an input $\sigma'$ such that 
	$\frac{V_{OPT}(\sigma)}{V_{PQ}(\sigma)} \leq \frac{V_{OPT}(\sigma')}{V_{PQ}(\sigma')}$ and 
	$OPT$ accepts all arriving packets. 
	\fi
\end{LMA}
\begin{proof}
	\ifnum \count10 > 0
	入力$\sigma$において、
	$OPT$がpacket $p$をrejectする最初のarrival event $e$が時刻$t$に発生するとする。
	$\sigma$から、$e$を取り除いて新しい入力$\sigma''$を構成する。
	このとき、$\sigma''$に対して、時刻$t$より後の時刻に、
	$PQ$は$\sigma$ではacceptできなかったpacket $q$をacceptできる可能性がある。
	そこで、
	$PQ$はscheduling eventが発生すると、
	その直前にキュー内に存在する全てのpacketに優先順位を割り当てて、
	その中で最も高い順位をもつパケットをscheduleするとする。
	$Q^{(i)}$を、$e$において$p$が到着するキューとする。
	このとき、
	$\sigma''$では、
	$t$以降のscheduling eventにおいては、
	$Q^{(j)} \hspace{1mm} (j \leq i)$に保持されているパケットの優先順位は高くなるが、
	$Q^{(j)} \hspace{1mm} (j > i)$に保持されているの優先順位は変化しない。
	すなわち、
	$\sigma$と$\sigma''$において、
	$Q^{(j)} \hspace{1mm} (j > i)$に保持されているパケットのscheduleされる時刻は変化しないので、
	$Q^{(j)} \hspace{1mm} (j > i)$に保持されているパケットの数も変化しない。
	よって、
	$\alpha_{k}$を、$q$の価値をとする。
	このとき、$i \geq k$が成立している。
	したがって、
	$V_{PQ}(\sigma'') \leq V_{PQ}(\sigma)$
	が成立する。
	一方、
	$V_{OPT}(\sigma'') = V_{OPT}(\sigma)$
	が成立するので、
	$\frac{V_{OPT}(\sigma)}{V_{PQ}(\sigma)} \leq \frac{V_{OPT}(\sigma'')}{V_{PQ}(\sigma'')}$. 
	結果として、
	$\sigma$から、$OPT$がパケットをrejectする様な全てのarrival eventを取り除いて$\sigma'$を構成すると、
	$\frac{V_{OPT}(\sigma)}{V_{PQ}(\sigma)} \leq \frac{V_{OPT}(\sigma')}{V_{PQ}(\sigma')}$
	が成立する。
	\fi
	\ifnum \count11 > 0
	\com{（■英語）}
	Let $e$ be the first arrival event where $OPT$ rejects a packet,  
	let $p$ be the arriving packet at $e$, and 
	let $t$ be the event time when $e$ happens. 
	We construct a new input $\sigma''$ by removing $e$ from a given input $\sigma$. 
	Then, 
	$PQ$ for $\sigma''$ might accept a packet $q$ which is not accepted for $\sigma$ after $t$. 
	Suppose that $PQ$ handles priorities to packets in its buffers, and transmits the packet with the highest priority at each scheduling event. 
	Let $Q^{(i)}$ be a queue at which $p$ arrives at $e$. 
	Then, at a scheduling event after $t$, 
	a priority which $PQ$ handles to a packet in $Q^{(j)} \hspace{1mm} (j \leq i)$ for $\sigma''$ is higher than that for $\sigma$. 
	However, 
	a priority which $PQ$ handles to a packet in $Q^{(j)} \hspace{1mm} (j > i)$ for $\sigma''$ is equal to that for $\sigma$. 
	Thus, 
	a time when a packet is transmitted from $Q^{(j)} \hspace{1mm} (j > i)$ in $\sigma''$ is the same as that in $\sigma$. 
	Also, 
	the number of packets which $PQ$ stores in $Q^{(j)} \hspace{1mm} (j > i)$ in $\sigma''$ is equivalent to that in $\sigma$. 
	Let $k$ be the integer such that $\alpha_{k}$ is the value of $q$. 
	Then, $i \geq k$ holds. 
	Hence, 
	$V_{PQ}(\sigma'') \leq V_{PQ}(\sigma)$. 
	On the other hand, 
	$V_{OPT}(\sigma'') = V_{OPT}(\sigma)$. 
	According  to the inequality and the equality, 
	$\frac{V_{OPT}(\sigma)}{V_{PQ}(\sigma)} \leq \frac{V_{OPT}(\sigma'')}{V_{PQ}(\sigma'')}$. 
	As a result, 
	we construct a new input $\sigma'$ by removing all arrival events at which $OPT$ rejects a packet from $\sigma$. 
	Then, 
	$\frac{V_{OPT}(\sigma)}{V_{PQ}(\sigma)} \leq \frac{V_{OPT}(\sigma')}{V_{PQ}(\sigma')}$. 
	\fi
\end{proof}
%

%
\ifnum \count14 > 0

%
\section{Proofs of Lemmas and Theorem} \label{ap.sec:3}
\subsection{Proof of Lemma~\ref{LMA:3.2.1}}
%
	%
	\ifnum \count10 > 0
	（■日本語）
	\com{（■英語と内容が違うんです）}\\
	$PQ$はその定義より、
	常に優先度の高いキューを選択する。
	よって、
	任意のnon-event time $t$に、
	$h_{PQ}^{(m)}(t) \leq h_{OPT}^{(m)}(t)$
	が成立する。
	また、
	Lemma~\ref{LMA:a.1}より、
	$OPT$はパケットをrejectしないので、
	$OPT$は$Q^{(m)}$にarriveするパケットを全て受理する。
	よって、
	$PQ$も$Q^{(m)}$にarriveするパケットを全て受理することが出来る。
	よって、
	${k}_{m} = 0$
	が成立する。
	\fi
	\ifnum \count11 > 0
	\com{（■英語）}
	By the definition of $PQ$, $PQ$ selects the non-empty queue with the	highest priority. 
	Thus, $h_{PQ}^{(m)}(t) \leq h_{OPT}^{(m)}(t)$ holds at any non-event time $t$. 
	Therefore, there is no free cell in $Q^{(m)}$ of $OPT$ at any time. 
	Since any extra packet is accepted to a free cell, 
	${k}_{m} = 0$.
	\fi
	%

%
\subsection{Proof of Lemma~\ref{LMA:3.2.3}}
%

	%
	\ifnum \count10 > 0
	$z$を、
	$s_{z}(\sigma) > B$が成立する最小の整数とする。
	このとき、次の3つの条件をみたす様な3つのevent time $t_{1}, t_{2} (> t_{1})$と$t_{3} (> t_{2})$が存在する。
	(i) $t_{2}$は、$Q^{(z)}$において$PQ$がacceptする$B+1$個目のpacketが到着するevent timeである。
	(ii) 
	時間$(t_{1}, t_{2})$の間に、$OPT$は$Q^{(z)}$からpacketをtransmitしない。
	ただし、$t_{1}$は、$OPT$が$Q^{(z)}$からpacketをtransmitするevent timeである。
		（仮定より、$OPT$は全てのpacketを受理するので、$t_{2}$より前に必ず$Q^{(z)}$からpacketをscheduleする。）
	(iii) 
	時間$(t_{2}, t_{3})$の間に、$PQ$は$Q^{(z)}$からpacketをtransmitしない。
	ただし、$t_{3}$は、$PQ$が$Q^{(z)}$からpacketをtransmitするevent timeである。
	このとき、$\sigma$から$t_{1}$と$t_{2}$に発生するeventを除いて$\sigma'$を構成する。
	もし、
	$\frac{V_{OPT}(\sigma)}{V_{PQ}(\sigma)} < \frac{V_{OPT}(\sigma')}{V_{PQ}(\sigma')}$
	が成立するならば、
	$\{ x \mid s_{x}(\sigma) > B \}$に含まれるキューの番号$j$の昇順に、
	$Q^{(j)}$に関するeventを取り除けば、
	\com{（■書き方うーむ。）}
	各$j (\in [1, m])$に対して、
	$s_{j}(\hat{\sigma}) \leq B$
	かつ
	$\frac{V_{OPT}(\sigma)}{V_{PQ}(\sigma)} < \frac{V_{OPT}(\hat{\sigma})}{V_{PQ}(\hat{\sigma})}$
	が成立する入力$\hat{\sigma}$を構成することが出来る。
	これは題意がみたされる。
	よって、以下では、
	$\frac{V_{OPT}(\sigma)}{V_{PQ}(\sigma)} < \frac{V_{OPT}(\sigma')}{V_{PQ}(\sigma')}$
	が成立することを示す。
	まず、
	$\sigma'$に対して$OPT$が得る価値を評価する。
	$ALG$を次の様なオフラインアルゴリズムとする。
	$\sigma'$の各event timeにおいて、$\sigma$において$OPT$が選ぶキューと同じキューを選ぶオフラインアルゴリズムとする。
	このとき、
	時間$(t_{1}, t_{3})$の間における$\sigma'$に対する$ALG$のバッファ内のパケット数について考える。
	任意のnon-event time $t (\in (t_{1}, t_{3}))$、任意の$y (\ne z)$に対して、
	$h_{ALG}^{(y)}(t) = h_{OPT}^{(y)}(t)$
	が成立する。
	任意のnon-event time $t (\in (t_{1}, t_{2}))$に対して、
	$h_{ALG}^{(z)}(t) = h_{OPT}^{(z)}(t) + 1$
	が成立し、また、
	任意のnon-event time $t (\in (t_{2}, t_{3}))$に対して、
	$h_{ALG}^{(z)}(t) = h_{OPT}^{(z)}(t)$
	が成立する。	
	以上より
	$V_{OPT}(\sigma') \geq V_{ALG}(\sigma') = V_{OPT}(\sigma) - \alpha_{z}$
	が成立する。
	次に、$\sigma'$に対して$PQ$が得る価値を評価する。
	簡単のため、
	$\sigma'$に対する$PQ$を$PQ'$と表記する。
	まず、
	時間$(t_{1}, t_{3})$の間に、
	$PQ$が受理するが、$PQ'$が非受理するpacketが存在しない場合を考える。
	この場合の$PQ'$の利得を評価するために
	時間$t_{1}$後の$PQ$と$PQ'$のバッファ内のpacket数について論じる。
	任意のnon-event time $t (\in (t_{1}, t_{2}))$に対して、
	$\sum_{j = 1}^{m} h_{PQ'}^{(j)}(t) = \sum_{j = 1}^{m} h_{PQ}^{(j)}(t) + 1$
	が成立する。
	任意のnon-event time $\hat{t}$に対して、
	$w(\hat{t}) = \arg \max\{ j \mid h_{PQ'}^{(j)}(\hat{t}) > 0 \}$と定義する。
	具体的には、
	$h_{PQ'}^{(w(t))}(t) = h_{PQ}^{(w(t))}(t) + 1$
	が成立する。
	（性質(a)と呼ぶ。）
	また、
	任意のnon-event time $t (\in (t_{2}, t_{3}))$に対して、
	$\sum_{j = 1}^{m} h_{PQ'}^{(j)}(t) = \sum_{j = 1}^{m} h_{PQ}^{(j)}(t)$
	が成立する。
	ただし、
	$w(t) > z$ならば、
	$h_{PQ'}^{(w(t))}(t) = h_{PQ}^{(w(t))}(t) + 1$
	が成立し、
	$h_{PQ'}^{(z)}(t) = h_{PQ}^{(z)}(t) - 1$
	が成立する。
	$w(t) = z$ならば、
	任意の$j (\in [1, m])$に対して、
	$h_{PQ'}^{(j)}(t) = h_{PQ}^{(j)}(t)$
	が成立する。
	任意のnon-event time $t (> t_{3})$、
	任意の$j (\in [1, m])$に対して、
	$h_{PQ'}^{(j)}(t) = h_{PQ}^{(j)}(t)$
	が成立する。
	以上より、
	$V_{PQ}(\sigma') = V_{PQ}(\sigma) - \alpha_{z}$
	が成立する。
	次に
	$PQ$がacceptし、$PQ'$がrejectするpacketが少なくとも1つは到着する場合を考える。
	$t'$を、
	$PQ$がacceptし、$PQ'$がrejectする様なpacket $p$が初めてarriveするevent timeとする。
	$t' \in (t_{1}, t_{2})$が成立する場合を考える。
	$z$の定義より、そのpacketは$z' \geq z$が成立する$Q^{(z')}$に到着するpacketである。
	性質(a)より、
	任意の$j (\in [1, m])$に対して、
	$h_{PQ'}^{(j)}(t'+) = h_{PQ}^{(j)}(t'+)$
	が成立する。
	よって、
	時間$(t', t_{2})$の間に、$PQ$が受理するpacketは、$PQ'$も受理することが出来る。
	$\sigma'$の定義より、
	$t_{2}$において$PQ$のみが$Q^{(z)}$に到着するpacketを受理するので、
	$h_{PQ'}^{(z)}(t_{2}+) = h_{PQ}^{(z)}(t_{2}+) - 1$
	が成立し、
	任意の$j (\in [1, m])$（ただし、$j \ne z$）に対して、
	$h_{PQ'}^{(j)}(t_{2}+) = h_{PQ}^{(j)}(t_{2}+)$
	が成立する。
	（性質(b)と呼ぶ。）
	$t_{2}$の後に、$PQ$と$PQ'$が受理するpacketが全て同じであれば、
	$V_{PQ}(\sigma') = V_{PQ}(\sigma) - \alpha_{z} - \alpha_{z'}$
	が成立する。
	次に、
	$t_{2}$の後に、
	$PQ$がrejectし、$PQ'$がacceptするpacket $p'$が存在する場合を考える。
	$PQ$の定義と性質(b)より、
	任意のnon-event time $t (> t_{2})$、任意の$z' (\geq z+1)$に対して、
	$h_{PQ'}^{(z')}(t) = h_{PQ}^{(z')}(t)$
	が成立する。
	よって、
	$p'$はある$z'' \leq z$が成立する$Q^{(z'')}$に到着するpacketである。
	$p'$が到着するevent timeを$t''$とする。
	任意の$j (\in [1, m])$に対して、
	$h_{PQ'}^{(j)}(t''+) = h_{PQ}^{(j)}(t''+)$
	が成立する。
	したがって、
	$t''$より後に$PQ$と$PQ'$が受理するpacketは一緒である。
	したがって、
	$V_{PQ}(\sigma') = V_{PQ}(\sigma) - \alpha_{z} - \alpha_{z'} + \alpha_{z''} \leq V_{PQ}(\sigma) - \alpha_{z}$
	が成立する。
	最後に
	$t' \in (t_{2}, t_{3})$が成立する場合を考える。
	上の場合と全く同じ議論が成立する。
	具体的には、
	$t'$の後に、
	$PQ$がrejectし、$PQ'$がacceptする様なpacketは高々1つしか存在しない。
	また、そのpacketはある$z''' \leq z$が成立する$Q^{(z''')}$に到着するpacketである。
	よって、
	$V_{PQ}(\sigma') = V_{PQ}(\sigma) - \alpha_{z} - \alpha_{z'} + \alpha_{z'''} \leq V_{PQ}(\sigma) - \alpha_{z}$
	が成立する。
	以上の議論より、
	$\frac{V_{OPT}(\sigma')}{V_{PQ}(\sigma')} 
		\geq \frac{V_{ALG}(\sigma')}{V_{PQ}(\sigma')} 
		\geq \frac{V_{OPT}(\sigma) - \alpha_{z}}{V_{PQ}(\sigma) - \alpha_{z}}
		> \frac{V_{OPT}(\sigma)}{V_{PQ}(\sigma)}$
	が成立する。
	\fi
	\ifnum \count11 > 0
	\com{（■英語）}
	Let $z$ be the minimum index such that $s_{z}(\sigma) > B$. 
	Then, 
	there exist the three event times $t_{1}, t_{2} (> t_{1})$ and $t_{3} (> t_{2})$ satisfying the following three conditions: 
	(i) $t_{2}$ is the arrival event time when the $(B+1)$st packet which $PQ$ accepts at $Q^{(z)}$ arrives, 
	(ii) 
	$OPT$ does not transmit any packet from $Q^{(z)}$ during time $(t_{1}, t_{2})$, 
	where $t_{1}$ is the event time when $OPT$ transmits a packet from $Q^{(z)}$, 
	(Since $OPT$ accepts any arriving packet by our assumption, $OPT$ certainly transmits at least one packet from $Q^{(z)}$ before $t_{2}$.) 
	and 
	(iii) 
	$PQ$ does not transmit any packet from $Q^{(z)}$ during time $(t_{2}, t_{3})$, 
	where $t_{3}$ is the event time when $PQ$ transmits a packet from $Q^{(z)}$. 
	We construct $\sigma'$ by removing the events at $t_{1}$ and $t_{2}$ from $\sigma$. 
	Suppose that $\frac{V_{OPT}(\sigma)}{V_{PQ}(\sigma)} < \frac{V_{OPT}(\sigma')}{V_{PQ}(\sigma')}$. 
	If we remove some events corresponding to $Q^{(j)}$ in ascending order of index $j$ in $\{ x \mid s_{x}(\sigma) > B \}$, 
	then we can construct an input $\hat{\sigma}$ such that 
	for each $j (\in [1, m])$, 
	$s_{j}(\hat{\sigma}) \leq B$, 
	and 
	$\frac{V_{OPT}(\sigma)}{V_{PQ}(\sigma)} < \frac{V_{OPT}(\hat{\sigma})}{V_{PQ}(\hat{\sigma})}$, which completes the proof. 
	Hence, 
	we next show that $\frac{V_{OPT}(\sigma)}{V_{PQ}(\sigma)} < \frac{V_{OPT}(\sigma')}{V_{PQ}(\sigma')}$. 
	First, 
	we discuss the gain of $OPT$ for $\sigma'$. 
	Let $ALG$ be the offline algorithm for $\sigma'$ such that 
	for each scheduling event $e$ in $\sigma'$, 
	$ALG$ selects the queue which $OPT$ selects at $e$ in $\sigma$. 
	We consider the number of packets in $ALG$'s buffer during time $(t_{1}, t_{3})$ for $\sigma'$. 
	For any non-event time $t (\in (t_{1}, t_{3}))$, and any $y (\ne z)$, 
	$h_{ALG}^{(y)}(t) = h_{OPT}^{(y)}(t)$. 
	For any non-event time $t (\in (t_{1}, t_{2}))$, 
	$h_{ALG}^{(z)}(t) = h_{OPT}^{(z)}(t) + 1$. 
	Also, for any non-event time $t (\in (t_{2}, t_{3}))$, 
	$h_{ALG}^{(z)}(t) = h_{OPT}^{(z)}(t)$. 
	By the above argument,
	$V_{OPT}(\sigma') \geq V_{ALG}(\sigma') = V_{OPT}(\sigma) - \alpha_{z}$. 
	Next, we evaluate the gain of $PQ$ for $\sigma'$. 
	For notational simplicity, 
	we describe $PQ$ for $\sigma'$ as $PQ'$. 
	First, 
	we consider the case where there does not exist any packet which $PQ$ accepts but $PQ'$ rejects during time $(t_{1}, t_{3})$. 
	To evaluate the gain of $PQ'$ in this case, 
	we discuss the numbers of packets which $PQ$ and $PQ'$ store in their buffers after $t_{1}$. 
	For any non-event time $t (\in (t_{1}, t_{2}))$, 
	$\sum_{j = 1}^{m} h_{PQ'}^{(j)}(t) = \sum_{j = 1}^{m} h_{PQ}^{(j)}(t) + 1$. 
	For any non-event time $\hat{t}$, 
	we define $w(\hat{t}) = \arg \max\{ j \mid h_{PQ'}^{(j)}(\hat{t}) > 0 \}$. 
	Specifically, 
	$h_{PQ'}^{(w(t))}(t) = h_{PQ}^{(w(t))}(t) + 1$. 
	(We call this fact the property (a).)
	Moreover, 
	for any non-event time $t (\in (t_{2}, t_{3}))$, 
	$\sum_{j = 1}^{m} h_{PQ'}^{(j)}(t) = \sum_{j = 1}^{m} h_{PQ}^{(j)}(t)$. 
	However, 
	if $w(t) > z$, 
	then $h_{PQ'}^{(w(t))}(t) = h_{PQ}^{(w(t))}(t) + 1$. 
	Also, 
	$h_{PQ'}^{(z)}(t) = h_{PQ}^{(z)}(t) - 1$. 
	If $w(t) = z$, 
	then for any $j (\in [1, m])$, $h_{PQ'}^{(j)}(t) = h_{PQ}^{(j)}(t)$. 
	For any non-event time $t (> t_{3})$ and 
	any $j (\in [1, m])$, 
	$h_{PQ'}^{(j)}(t) = h_{PQ}^{(j)}(t)$. 
	By the above argument, 
	$V_{PQ}(\sigma') = V_{PQ}(\sigma) - \alpha_{z}$ holds. 
	Secondly, 
	we consider the case where there exists at least one packet which $PQ$ accepts but $PQ'$ rejects. 
	Let $t'$ be the first event time when the packet $p$ which $PQ$ accepts but $PQ'$ rejects arrives. 
	Then, suppose that $t' \in (t_{1}, t_{2})$. 
	By the definition of $z$, 
	$p$ arrives at $Q^{(z')}$ such that $z' \geq z$. 
	By the property (a), 
	for $j (\in [1, m])$, 
	$h_{PQ'}^{(j)}(t'+) = h_{PQ}^{(j)}(t'+)$. 
	Thus, 
	packets accepted by $PQ$ during time $(t', t_{2})$ can be accepted by $PQ'$. 
	Only $PQ$ accepts the packet arriving at $Q^{(z)}$ at $t_{2}$ by the definition of $\sigma'$. 
	Hence, 
	$h_{PQ'}^{(z)}(t_{2}+) = h_{PQ}^{(z)}(t_{2}+) - 1$, 
	and 
	for any $j (\in [1, m])$ such that $j \ne z$, 
	$h_{PQ'}^{(j)}(t_{2}+) = h_{PQ}^{(j)}(t_{2}+)$. 
	(We call this fact the property (b).)
	If all the packets which $PQ$ accepts after $t_{2}$ are the same as those accepted by $PQ'$ after $t_{2}$, 
	$V_{PQ}(\sigma') = V_{PQ}(\sigma) - \alpha_{z} - \alpha_{z'}$. 
	Then, 
	we consider the case where there exists at least one packet $p'$ which $PQ$ rejects but $PQ'$ accepts after $t_{2}$. 
	By the greediness of $PQ$ and the property (b), 
	for any non-event time $t (> t_{2})$ and any $z' (\geq z+1)$, 
	$h_{PQ'}^{(z')}(t) = h_{PQ}^{(z')}(t)$. 
	Hence, 
	$p'$ arrives at $Q^{(z'')}$ for some $z'' (\leq z)$. 
	Let $t''$ be the event time when $p'$ arrives. 
	For any $j (\in [1, m])$, 
	$h_{PQ'}^{(j)}(t''+) = h_{PQ}^{(j)}(t''+)$, 
	which means that 
 	all the packets accepted by $PQ$ are equal to those accepted by $PQ'$ after $t''$. 
	Thus, 
	$V_{PQ}(\sigma') = V_{PQ}(\sigma) - \alpha_{z} - \alpha_{z'} + \alpha_{z''} \leq V_{PQ}(\sigma) - \alpha_{z}$. 
	Finally, 
	we consider the case where $t' \in (t_{2}, t_{3})$. 
	By the same argument as the case of $t' \in (t_{1}, t_{2})$, 
	we can prove this case. 
	Specifically, 
	the number of packets which $PQ$ rejects but $PQ'$ accepts after $t'$ is exactly one. 
	This packet arrives at $Q^{(z''')}$, where some $z''' \leq z$. 
	Therefore, 
	$V_{PQ}(\sigma') = V_{PQ}(\sigma) - \alpha_{z} - \alpha_{z'} + \alpha_{z'''} \leq V_{PQ}(\sigma) - \alpha_{z}$. 
	By the above argument, 
	$\frac{V_{OPT}(\sigma')}{V_{PQ}(\sigma')} 
		\geq \frac{V_{ALG}(\sigma')}{V_{PQ}(\sigma')} 
		\geq \frac{V_{OPT}(\sigma) - \alpha_{z}}{V_{PQ}(\sigma) - \alpha_{z}}
		> \frac{V_{OPT}(\sigma)}{V_{PQ}(\sigma)}$. 
	\fi
	%

%
\subsection{Proof of Lemma~\ref{LMA:3.2.4}}
%
	%
	\ifnum \count10 > 0
	補題~\ref{LMA:3.2.2}より、
	入力が終了した時点において、
	各extra packet $p$は$PQ$がtransmitするpacket $p'$に必ずmatchされており、
	また、
	extra packet $p$と$PQ$のpacket $p'$がmatchされている場合、
	$g(p) < g(p')$
	が成立している。
	よって、
	$k_{q_{n}} \leq \sum_{j = q_{n}+1}^{m} s_{j}$
	かつ
	$k_{q_{n-1}} \leq (\sum_{j = q_{n-1}+1}^{m} s_{j}) - k_{q_{n}}$
	かつ
	…
	$k_{q_{1}} \leq (\sum_{j = q_{1}+1}^{m} s_{j}) - \sum_{i = 2}^{n} k_{q_{i}}$
	が成立する。
	よって、
	任意の$x (\in [1, n])$に対して、
	$\sum_{i = x}^{n} k_{q_{i}} \leq \sum_{j = q_{x}+1}^{m} s_{j}$
	が成立する。
	\fi
	\ifnum \count11 > 0
	\com{（■英語）}
	By Lemma~\ref{LMA:3.2.2}, 
	each extra packet $p$ is matched with a packet $p'$ transmitted by $PQ$ at the end of the input. 
	In addition, 
	$g(p) < g(p')$ 
	if an extra packet $p$ is matched with a packet $p'$ of $PQ$. 
	Thus, 
	$k_{q_{n}} \leq \sum_{j = q_{n}+1}^{m} s_{j}$, 
	$k_{q_{n-1}} \leq (\sum_{j = q_{n-1}+1}^{m} s_{j}) - k_{q_{n}}$, $\cdots$, and 
	$k_{q_{1}} \leq (\sum_{j = q_{1}+1}^{m} s_{j}) - \sum_{i = 2}^{n} k_{q_{i}}$. 
	Therefore, 
	for any $x (\in [1, n])$, 
	$\sum_{i = x}^{n} k_{q_{i}} \leq \sum_{j = q_{x}+1}^{m} s_{j}$. 
	\fi
	%

%
\subsection{Proof of Lemma~\ref{LMA:3.2.5}}
%

	%
	\ifnum \count10 > 0
	任意の入力$\sigma \in {\cal S}_1$に対して、
	$\sigma$から次の様な入力$\sigma'$を構成する。
	まず、
	時間$(0, 1)$の間に各$Q^{(j)} \hspace{1mm} (j \in [q_{1}, m])$にpacketが到着する$s_{j}$回のarrival eventが発生する。
	${\cal S}_1$の定義より、
	$s_{j} \leq B$が成立するので、$PQ$はそれらの到着するpacketを全て受理する。
	時刻1以降に、
	$PQ$がacceptできない$\sum_{i = 1}^{n} k_{q_{i}}$個のpacketだけがarriveする。
	詳しく言うと、
	任意の$i (\in [1, n])$に対して、
	$a_{i} = \sum_{j = q_{n+1-i}+1}^{q_{n+2-i}} s_{j}$と定義し、
	$a_{0} = 0$
	と定義する。
	任意の$x (\in [0, n-1])$に対して、
	整数時間$t = (\sum_{j = 0}^{x} a_{j} ) + 1, \ldots, \sum_{j = 0}^{x+1} a_{j}$に、
	scheduling eventが発生し、
	時刻$t + \frac{1}{2}$において、
	$Q^{(q_{n-x})}$にpacketが到着する様なarrival eventが発生する。
	時刻$(\sum_{j = 0}^{n} a_{j})+1$の後は、
	十分な数のscheduling eventが発生する。
	このとき、
	$PQ$は$t$に、$j (\in [q_{n-x}+1, q_{n-x+1}])$が成立する$Q^{(j)}$からpacketをtransmitする。
	また、
	あるオフラインアルゴリズム$ALG$を考える。
	$ALG$は$t$に、$Q^{(q_{n-x})}$からpacketをtransmitする。
	このとき、
	任意の$i (\in [1, n])$に対して、
	$Q^{(q_{i})}$はextra packetが到着するので、
	$s_{q_{i}} = B$が成立している。
	よって、
	任意の$i (\in [1, n])$に対して、
	$h_{PQ}^{(q_{i})}(1-) = B$が成立しているので、
	$PQ$は各時刻$t + \frac{1}{2}$において、到着する全てのpacketをacceptできない。
	しかし、$ALG$は全てのpacketを受理することが出来る。
	すなわち、$ALG$は最適なオフラインアルゴリズムの1つである。
	このとき、
	$n(\sigma') = n$
	かつ
	任意の$i (\in [1, n])$に対して、
	$q_{i}(\sigma') = q_{i}$
	が成立する。
	以上より、
	$V_{PQ}(\sigma') = V_{PQ}(\sigma) - \sum_{j = 1}^{q_{1}-1} \alpha_{j} s_{j}$
	が成立する。
	また、各$i (\in [1, n])$に対して、
	$k_{q_{i}}(\sigma') = \sum_{j = q_{i}+1}^{q_{i+1}} s_{j}$
	が成立する。
	以上の式より、
	$V_{ALG}(\sigma') 
		= V_{PQ}(\sigma') + \sum_{i = 1}^{n} \alpha_{q_{i}} k_{q_{i}}(\sigma')
		= V_{PQ}(\sigma) + \sum_{i = 1}^{n} \alpha_{q_{i}} ( \sum_{j = q_{i}+1}^{q_{i+1}} s_{j} ) - \sum_{j = 1}^{q_{1}-1} \alpha_{j} s_{j}$
	が成立する。
	補題~\ref{LMA:3.2.4}より、
	$\sum_{i = x}^{n} k_{q_{i}} \leq \sum_{j = q_{x}+1}^{m} s_{j}$
	が成立するので、
	上の不等式から各$s_{j}$を消去すると、
	$V_{ALG}(\sigma') 
		\geq V_{PQ}(\sigma) + \sum_{i = 1}^{n} \alpha_{q_{i}} k_{q_{i}} - \sum_{j = 1}^{q_{1}-1} \alpha_{j} s_{j}
		= V_{OPT}(\sigma) - \sum_{j = 1}^{q_{1}-1} \alpha_{j} s_{j}$
	が成立する。
	以上より、
	$\frac{V_{OPT}(\sigma')}{V_{PQ}(\sigma')} = \frac{V_{ALG}(\sigma')}{V_{PQ}(\sigma')}
		\geq \frac{V_{OPT}(\sigma) - \sum_{j = 1}^{q_{1}-1} \alpha_{j} s_{j} }{V_{PQ}(\sigma) - \sum_{j = 1}^{q_{1}-1} \alpha_{j} s_{j}}
		\geq \frac{V_{OPT}(\sigma)}{V_{PQ}(\sigma)}
	$
	が成立する。
	また、$\sigma'$の定義より、
	$\sigma'$は条件(ii)を満たし、${\cal S}_1$に含まれる。
	\com{（■S)}
	\fi
	\ifnum \count11 > 0
	\com{（■英語）}
	For any input $\sigma \in {\cal S}_1$, 
	we construct $\sigma'$ from $\sigma$ according to the following steps. 
	First, 
	for each $j (\in [q_{1}, m])$, 
	$s_{j}$ events at which $s_{j}$ packets arrive at $Q^{(j)}$ occur during time $(0, 1)$. 
	Since $s_{j} \leq B$ by the definition of ${\cal S}_1$, 
	$PQ$ accepts all the packets which arrive at these events. 
	$\sum_{i = 1}^{n} k_{q_{i}}$ packets arrive after time 1, and $PQ$ cannot accept them.  
	Specifically, 
	for any $i (\in [1, n])$, 
	we define $a_{i} = \sum_{j = q_{n+1-i}+1}^{q_{n+2-i}} s_{j}$ and 
	$a_{0} = 0$. 
	Then, 
	for each $x (\in [0, n-1])$, 
	a scheduling event occurs at each integer time $t = (\sum_{j = 0}^{x} a_{j} ) + 1, \ldots, \sum_{j = 0}^{x+1} a_{j}$, and 
	an arrival event where a packet arrives at $Q^{(q_{n-x})}$ occurs at each time $t + \frac{1}{2}$. 
	After time $(\sum_{j = 0}^{n} a_{j}) + 1$, 
	sufficient scheduling events to transmit all the arriving packets occur. 
	For these scheduling events, 
	$PQ$ transmits a packet from $Q^{(j)}$ at $t$, 
	where $j$ is an integer between $q_{n-x}+1$ and $q_{n-x+1}$. 
	Also, 
	let $ALG$ be an offline algorithm. 
	$ALG$ transmits a packet from $Q^{(q_{n-x})}$ at $t$. 
	Since for any $i (\in [1, n])$, at least one extra packet arrives at $Q^{(q_{i})}$, 
	$s_{q_{i}} = B$ holds. 
	Hence, 
	since for any $i (\in [1, n])$, 
	$h_{PQ}^{(q_{i})}(1-) = B$, 
	$PQ$ cannot accept the packet which arrives at each $t + \frac{1}{2}$. 
	However, 
	$ALG$ can accept all these packets, which means that 
	$ALG$ is an optimal offline algorithm.
	Then, 
	$n(\sigma') = n$, and 
	for any $i (\in [1, n])$, 
	$q_{i}(\sigma') = q_{i}$. 
	By the above argument, 
	$V_{PQ}(\sigma') = V_{PQ}(\sigma) - \sum_{j = 1}^{q_{1}-1} \alpha_{j} s_{j}$. 
	Furthermore, for each $i (\in [1, n])$, 
	$k_{q_{i}}(\sigma') = \sum_{j = q_{i}+1}^{q_{i+1}} s_{j}$. 
	By these equalities, 
	$V_{ALG}(\sigma') 
		= V_{PQ}(\sigma') + \sum_{i = 1}^{n} \alpha_{q_{i}} k_{q_{i}}(\sigma')
		= V_{PQ}(\sigma) + \sum_{i = 1}^{n} \alpha_{q_{i}} ( \sum_{j = q_{i}+1}^{q_{i+1}} s_{j} ) - \sum_{j = 1}^{q_{1}-1} \alpha_{j} s_{j} 
		= V_{PQ}(\sigma) + \alpha_{q_{1}} ( \sum_{j = q_{1}+1}^{q_{n+1}} s_{j} ) + \sum_{x = 2}^{n} (\alpha_{q_{x}} - \alpha_{q_{x-1}}) ( \sum_{j = q_{x}+1}^{q_{n+1}} s_{j} ) - \sum_{j = 1}^{q_{1}-1} \alpha_{j} s_{j}$. 
	Since $\sum_{i = x}^{n} k_{q_{i}} \leq \sum_{j = q_{x}+1}^{m} s_{j}$ by Lemma~\ref{LMA:3.2.4} and $q_{n+1}=m$, 
	$V_{ALG}(\sigma') 
		\geq V_{PQ}(\sigma) + \alpha_{q_{1}} ( \sum_{i = 1}^{n} k_{q_{i}} ) + \sum_{x = 2}^{n} (\alpha_{q_{x}} - \alpha_{q_{x-1}}) ( \sum_{i = x}^{n} k_{q_{i}} ) - \sum_{j = 1}^{q_{1}-1} \alpha_{j} s_{j}
		= V_{PQ}(\sigma) + \sum_{i = 1}^{n} \alpha_{q_{i}} k_{q_{i}} - \sum_{j = 1}^{q_{1}-1} \alpha_{j} s_{j}
		= V_{OPT}(\sigma) - \sum_{j = 1}^{q_{1}-1} \alpha_{j} s_{j}$. 
	Therefore, 
	$\frac{V_{OPT}(\sigma')}{V_{PQ}(\sigma')} = \frac{V_{ALG}(\sigma')}{V_{PQ}(\sigma')}
		\geq \frac{V_{OPT}(\sigma) - \sum_{j = 1}^{q_{1}-1} \alpha_{j} s_{j} }{V_{PQ}(\sigma) - \sum_{j = 1}^{q_{1}-1} \alpha_{j} s_{j}}
		\geq \frac{V_{OPT}(\sigma)}{V_{PQ}(\sigma)}
	$. 
	Moreover, 
	by the definition of $\sigma'$, 
	$\sigma'$ satisfies the condition (ii) in the statement, 
	which means that ${\cal S}_1$ includes $\sigma'$. 
	\fi
	%

%
\subsection{Proof of Lemma~\ref{LMA:3.2.6}}
%

	%
	\ifnum \count10 > 0
	任意の$j (\in [1, m])$（ただし、$j \ne q_{z+1}-1$）に対して、
	$s'_{j} = s_{j}$
	と定義する。
	また、
	$s'_{q_{z+1}-1} = B$
	と定義する。
	（Appendix~\ref{ap.sec:3}の図~\ref{fig:L37}参照）
	次の様に$\sigma$から入力$\sigma'$を構成する。
	この構成の仕方は、補題~\ref{LMA:3.2.5}に類似している。
	まず、
	時間$(0, 1)$の間に各$Q^{(j)} \hspace{1mm} (j \in [q_{1}, m])$にpacketが到着する$s'_{j}$回のarrival eventが発生する。
	定義より、
	$s'_{j} \leq B$が成立するので、
	$PQ$はそれらの到着するpacketを全て受理する。
	ここで、
	任意の$i (\in [1, z])$に対して、
	$q'_{i} = q_{i}$と定義し、
	$q'_{z+1} = q_{z+1} - 1$と定義し、
	任意の$i (\in [z+1, n+1])$に対して、
	$q'_{i+1} = q_{i}$と定義する。
	そのとき、
	任意の$i (\in [1, n+1])$に対して、
	$a_{i} = \sum_{j = q'_{n+2-i}+1}^{q'_{n+3-i}} s'_{j}$と定義し、
	$a_{0} = 0$と定義する。
	更に、
	任意の$x (\in [0, n])$に対して、
	整数時間$t = (\sum_{j = 0}^{x} a_{j} ) + 1, \ldots, \sum_{j = 0}^{x+1} a_{j}$に、
	scheduling eventが発生し、
	時刻$t + \frac{1}{2}$において、
	$Q^{(q'_{n-x+1})}$にpacketが到着する様なarrival eventが発生する。
	時刻$(\sum_{j = 0}^{n+1} a_{j})+1$の後は、
	十分な数のscheduling eventが発生する。
	このとき、
	$PQ$は$t$に、$j (\in [q'_{n-x+1}+1, q'_{n-x+2}])$が成立する$Q^{(j)}$からpacketをtransmitする。
	また、
	あるオフラインアルゴリズム$ALG$を考える。
	$ALG$は$t$に、$Q^{(q'_{n-x+1})}$からpacketをtransmitする。
	$q'_{i}$の定義より
	任意の$i (\in [1, n+1])$に対して、
	$h_{PQ}^{(q'_{i})}(1-) = B$が成立しているので、
	$PQ$は各時刻$t + \frac{1}{2}$において、到着する全てのpacketをacceptできない。
	しかし、$ALG$は全てのpacketを受理することが出来る。
	すなわち、$ALG$は最適なオフラインアルゴリズムである。
	以上より、
	$V_{PQ}(\sigma') = V_{PQ}(\sigma) + \alpha_{q_{z+1}-1} (B - s_{q_{z+1}-1})$
	が成立する。
	また、
	任意の$j (\ne q_{z}, q_{z+1}-1)$に対して、
	$k_{j}(\sigma') = k_{j}$
	が成立し、
	$k_{q_{z}}(\sigma') = k_{q_{z}} - s_{q_{z+1}-1}$
	が成立し、
	$k_{q_{z+1}-1}(\sigma') = B$
	が成立する。
	また、
	任意の$i (\in [1, n+1])$に対して、
	$q_{i}(\sigma') = q'_{i}$
	が成立する。
	更に、
	$V_{OPT}(\sigma') = V_{ALG}(\sigma') = V_{PQ}(\sigma') + \sum_{i = 1}^{n(\sigma')} \alpha_{q_{i}(\sigma')} k_{q_{i}(\sigma')}(\sigma')$
	が成立する。
	以上の式より、
	$\sum_{i = 1}^{n(\sigma')} \alpha_{q_{i}(\sigma')} k_{q_{i}(\sigma')}(\sigma')
		= (\sum_{i = 1}^{n} \alpha_{q_{i}} k_{q_{i}}) - \alpha_{q_{z}} s_{q_{z+1}-1} + \alpha_{q_{z+1}-1} B
		\geq (\sum_{i = 1}^{n} \alpha_{q_{i}} k_{q_{i}}) + \alpha_{q_{z+1}-1} (B - s_{q_{z+1}-1})
	$
	が成立し、
	$\frac{ \sum_{i = 1}^{n(\sigma')} \alpha_{q_{i}(\sigma')} k_{q_{i}(\sigma')}(\sigma') }{ V_{PQ}(\sigma') }
		\geq \frac{ (\sum_{i = 1}^{n} \alpha_{q_{i}} k_{q_{i}}) + \alpha_{q_{z+1}-1} (B - s_{q_{z+1}-1}) }
			{ V_{PQ}(\sigma) + \alpha_{q_{z+1}-1} (B - s_{q_{z+1}-1}) }
		\geq \frac{ \sum_{i = 1}^{n} \alpha_{q_{i}} k_{q_{i}} }{ V_{PQ}(\sigma) }
	$
	が成立する。
	よって、
	$\frac{V_{OPT}(\sigma')}{V_{PQ}(\sigma')} 
		\geq \frac{ V_{PQ}(\sigma') + \sum_{i = 1}^{n(\sigma')} \alpha_{q_{i}(\sigma')} k_{q_{i}(\sigma')}(\sigma') }
					{ V_{PQ}(\sigma') }
		\geq 1 + \frac{ \sum_{i = 1}^{n} \alpha_{q_{i}} k_{q_{i}} }{ V_{PQ}(\sigma) }
		= \frac{V_{OPT}(\sigma)}{V_{PQ}(\sigma)} 
	$
	が成立する。
	上記の$\sigma'$の定義より、
	$\sigma' \in {\cal S}_{2}$が成立している。
	よって、
	以上の議論より、再帰的に
	$q_{z'} + 1 < q_{z'+1}$が成立する様な任意の$z'$に対して、
	上記の様に新しい入力を構成すると、
	題意をみたす入力を構成することができる。
	\fi
	\ifnum \count11 > 0
	\com{（■英語）}
	For any $j (\in [1, m])$ such that $j \ne q_{z+1}-1$, 
	we define $s'_{j} = s_{j}$. 
	Also, we define $s'_{q_{z+1}-1} = B$. 
	(See Figure~\ref{fig:L37}.)
	We construct $\sigma'$ from $\sigma$ in the following way. 
	This approach is similar to those in the proof of Lemma~\ref{LMA:3.2.5}. 
	First, 
	for each $j (\in [q_{1}, m])$, 
	$s'_{j}$ events at which $s'_{j}$ packets arrive at $Q^{(j)}$ occur during time $(0, 1)$. 
	Since $s'_{j} \leq B$ by definition, 
	$PQ$ accepts all these packets. 
	In addition, 
	for any $i (\in [1, z])$, 
	we define $q'_{i} = q_{i}$. 
	We define $q'_{z+1} = q_{z+1} - 1$. 
	For any $i (\in [z+1, n+1])$, 
	we define $q'_{i+1} = q_{i}$. 
	Moreover, 
	for any $i (\in [1, n+1])$, 
	we define $a_{i} = \sum_{j = q'_{n+2-i}+1}^{q'_{n+3-i}} s'_{j}$ and 
	$a_{0} = 0$. 
	For any $x (\in [0, n])$, 
	a scheduling event occurs at each integer time $t = (\sum_{j = 0}^{x} a_{j} ) + 1, \ldots, \sum_{j = 0}^{x+1} a_{j}$. 
	Also, 
	an arrival event where a packet arrives at $Q^{(q'_{n-x+1})}$ occurs at each time $t + \frac{1}{2}$. 
	After time $(\sum_{j = 0}^{n+1} a_{j}) + 1$, 
	sufficient scheduling events to transmit all the arriving packets occur. 
	Then, 
	$PQ$ transmits a packet from $Q^{(j)}$ at $t$, 
	where $j$ is an integer between $q'_{n-x+1}+1$ and $q'_{n-x+2}$. 
	Let $ALG$ be an offline algorithm which transmits a packet from $Q^{(q'_{n-x+1})}$ at $t$. 
	By the definition of $q'_{i}$, 
	for any $i (\in [1, n+1])$, 
	$h_{PQ}^{(q'_{i})}(1-) = B$. 
	Thus, 
	$PQ$ cannot accept any packet arriving at $t + \frac{1}{2}$, 
	but $ALG$ can accept all the arriving packets. 
	That is to say, $ALG$ is optimal. 
	By the above argument, 
	$V_{PQ}(\sigma') = V_{PQ}(\sigma) + \alpha_{q_{z+1}-1} (B - s_{q_{z+1}-1})$. 
	Furthermore, 
	for any $j (\ne q_{z}, q_{z+1}-1)$, 
	$k_{j}(\sigma') = k_{j}$. 
	Also, 
	$k_{q_{z}}(\sigma') = k_{q_{z}} - s_{q_{z+1}-1}$ and 
	$k_{q_{z+1}-1}(\sigma') = B$. 
	Also, 
	for any $i (\in [1, n+1])$, 
	$q_{i}(\sigma') = q'_{i}$.
	Moreover, 
	$V_{OPT}(\sigma') = V_{ALG}(\sigma') = V_{PQ}(\sigma') + \sum_{i = 1}^{n(\sigma')} \alpha_{q_{i}(\sigma')} k_{q_{i}(\sigma')}(\sigma')$. 
	By the above equalities, 
	$\sum_{i = 1}^{n(\sigma')} \alpha_{q_{i}(\sigma')} k_{q_{i}(\sigma')}(\sigma')
		= (\sum_{i = 1}^{n} \alpha_{q_{i}} k_{q_{i}}) - \alpha_{q_{z}} s_{q_{z+1}-1} + \alpha_{q_{z+1}-1} B
		\geq (\sum_{i = 1}^{n} \alpha_{q_{i}} k_{q_{i}}) + \alpha_{q_{z+1}-1} (B - s_{q_{z+1}-1})
	$. 
	Hence, 
	$\frac{ \sum_{i = 1}^{n(\sigma')} \alpha_{q_{i}(\sigma')} k_{q_{i}(\sigma')}(\sigma') }{ V_{PQ}(\sigma') }
		\geq \frac{ (\sum_{i = 1}^{n} \alpha_{q_{i}} k_{q_{i}}) + \alpha_{q_{z+1}-1} (B - s_{q_{z+1}-1}) }
			{ V_{PQ}(\sigma) + \alpha_{q_{z+1}-1} (B - s_{q_{z+1}-1}) }
		\geq \frac{ \sum_{i = 1}^{n} \alpha_{q_{i}} k_{q_{i}} }{ V_{PQ}(\sigma) }
	$. 
	Therefore, 
	$\frac{V_{OPT}(\sigma')}{V_{PQ}(\sigma')} 
		\geq \frac{ V_{PQ}(\sigma') + \sum_{i = 1}^{n(\sigma')} \alpha_{q_{i}(\sigma')} k_{q_{i}(\sigma')}(\sigma') }
					{ V_{PQ}(\sigma') }
		\geq 1 + \frac{ \sum_{i = 1}^{n} \alpha_{q_{i}} k_{q_{i}} }{ V_{PQ}(\sigma) }
		= \frac{V_{OPT}(\sigma)}{V_{PQ}(\sigma)} 
	$. 
	By the definition of $\sigma'$, 
	${\cal S}_{2}$ includes $\sigma'$. 
	By the above argument, 
	for any $z'$ such that $q_{z'} + 1 < q_{z'+1}$, 
	we recursively construct an input in the above way, and then we can obtain an input satisfying the lemma. 
	\fi
	%
%
\ifnum \count12 > 0
\begin{figure*}[h]
 \begin{center}
  \includegraphics[width=120mm]{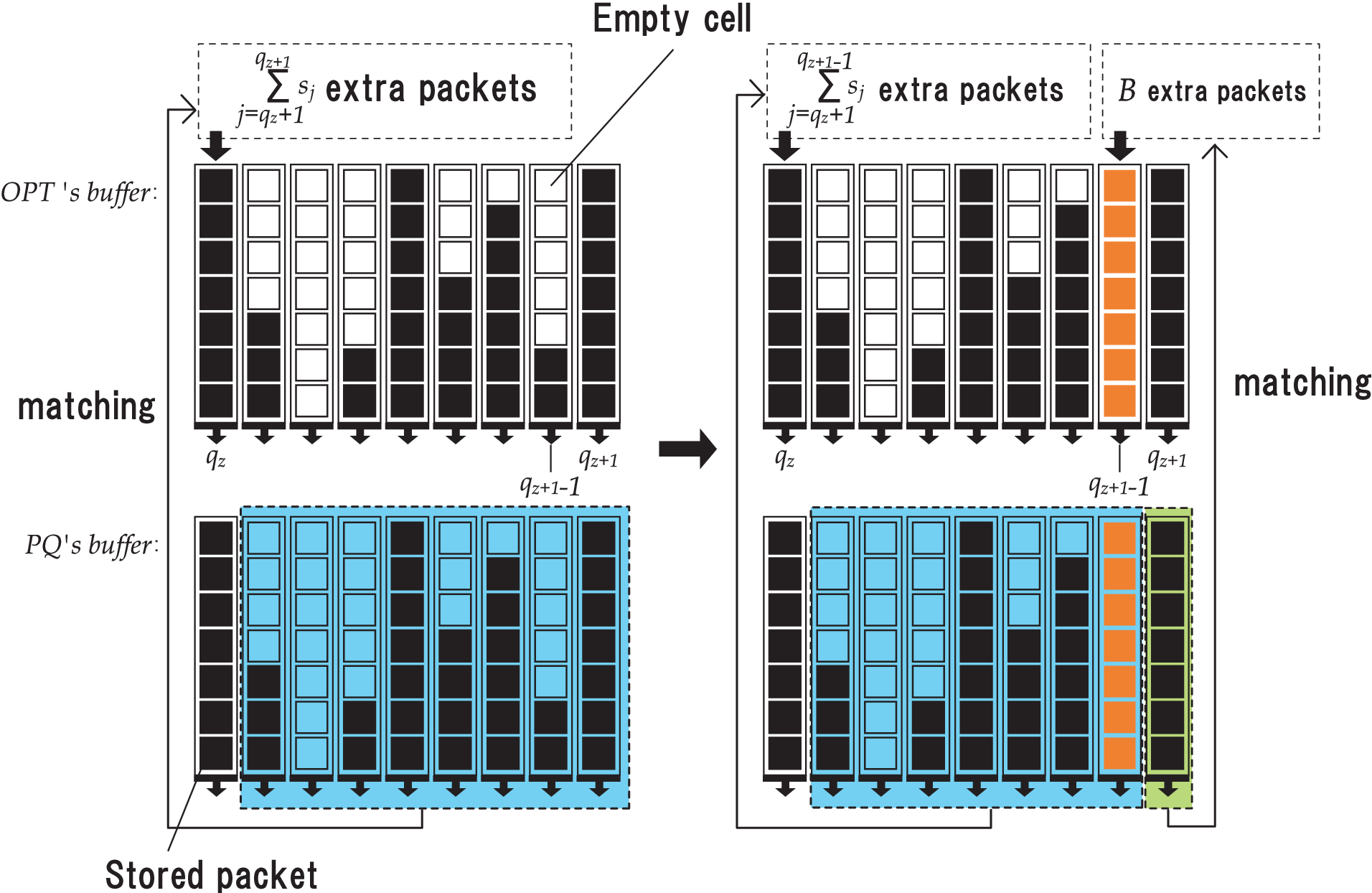}
 \end{center}
 \caption{
 			Example states of queues ($q_{z}$ through $q_{z+1}$) of $OPT$ and $PQ$ for $\sigma$ and $\sigma'$. 
			Left (Right) queues show the states for $\sigma$ ($\sigma'$). 
		}
\label{fig:L37}
 \end{figure*}
\fi
\subsection{Proof of Lemma~\ref{LMA:3.2.7}}
%

	%
	\ifnum \count10 > 0
	%
	%
	%
	任意の$j (\in [1, q_{n}])$に対して、
	$s''_{j} = s_{j}$
	と定義する。
	また、
	各$j (\in [q_{n}+1, q_{n}+u])$に対して、
	$s''_{j} = B$
	かつ
	$s''_{q_{n}+u+1} = (\sum_{j = q_{n}+1}^{m} s_{j}) - u B$	
	と定義する。
	次の様に$\sigma$から入力$\sigma'$を構成する。
	この構成の仕方は、補題~\ref{LMA:3.2.5}と補題~\ref{LMA:3.2.6}に類似している。
	まず、
	時間$(0, 1)$の間に各$Q^{(j)} \hspace{1mm} (j \in [q_{1}, m])$にpacketが到着する$s''_{j}$回のarrival eventが発生する。
	定義より、
	$s''_{j} \leq B$が成立するので、
	$PQ$はそれらの到着するpacketを全て受理する。
	ここで、
	任意の$i (\in [1, n])$に対して、
	$a_{i} = \sum_{j = q_{n+1-i}+1}^{q_{n+2-i}} s''_{j}$と定義し、
	$a_{0} = 0$
	と定義する。
	任意の$x (\in [0, n-1])$に対して、
	整数時間$t = (\sum_{j = 0}^{x} a_{j} ) + 1, \ldots, \sum_{j = 0}^{x+1} a_{j}$に、
	scheduling eventが発生し、
	時刻$t + \frac{1}{2}$において、
	$Q^{(q_{n-x})}$にpacketが到着する様なarrival eventが発生する。
	時刻$(\sum_{j = 0}^{n} a_{j})+1$の後は、
	十分な数のscheduling eventが発生する。
	明らかに、
	$V_{PQ}(\sigma') \leq V_{PQ}(\sigma)$
	と
	$V_{OPT}(\sigma') = V_{OPT}(\sigma)$
	が成立する。
	また、
	$\sigma'$の定義より、
	$\sigma' \in {\cal S}_{3}$が成立し、
	$\sigma'$はステートメントの条件(i)をみたす。
	\fi
	\ifnum \count11 > 0
	\com{（■英語）}
	For any $j (\in [1, q_{n}])$, 
	we define $s''_{j} = s_{j}$. 
	Furthermore, 
	for each $j (\in [q_{n}+1, q_{n}+u])$, 
	we define $s''_{j} = B$, and 
	$s''_{q_{n}+u+1} = (\sum_{j = q_{n}+1}^{m} s_{j}) - u B$. 
	We construct $\sigma'$ from $\sigma$ in the following way. 
	This approach is similar to those in the proof of Lemmas~\ref{LMA:3.2.5} and  \ref{LMA:3.2.6}. 
	First, 
	for each $j (\in [q_{1}, m])$, 
	$s''_{j}$ events at which $s_{j}$ packets arrive at $Q^{(j)}$ occur during time $(0, 1)$. 
	Since $s''_{j} \leq B$ by definition, 
	$PQ$ accepts all these packets. 
	Then, 
	for any $i (\in [1, n])$, 
	we define $a_{i} = \sum_{j = q_{n+1-i}+1}^{q_{n+2-i}} s_{j}$, and 
	$a_{0} = 0$. 
	For any $x (\in [0, n-1])$, 
	a scheduling event occurs at each integer time $t = (\sum_{j = 0}^{x} a_{j} ) + 1, \ldots, \sum_{j = 0}^{x+1} a_{j}$. 
	Also, 
	at each time $t + \frac{1}{2}$, 
	an arrival event where a packet arrives at $Q^{(q_{n-x})}$ occurs. 
	After time $(\sum_{j = 0}^{n} a_{j}) + 1$, 
	sufficient scheduling events to transmit all the arriving packets occur. 
	It is easy to see that 
	$V_{PQ}(\sigma') \leq V_{PQ}(\sigma)$
	and 
	$V_{OPT}(\sigma') = V_{OPT}(\sigma)$. 
	Moreover, 
	by the definition of $\sigma'$, 
	$\sigma' \in {\cal S}_{3}$ holds, and $\sigma'$ satisfies the condition (i) in the statement. 
	\fi
	%

%
\subsection{Proof of Lemma~\ref{LMA:3.2.8}}
%

	%
	\ifnum \count10 > 0
	%
	%
%
	%
	$\sigma$から次の様な入力$\sigma'$を構成する。
	まず、
	時間$(0, 1)$の間に各$Q^{(j)} \hspace{1mm} (j \in [q_{1}, m])$にpacketが到着する$s_{j}$回のarrival eventが発生する。
	${\cal S}_{4}$の定義より、
	$s_{j} \leq B$が成立するので、$PQ$はそれらの到着するpacketを全て受理する。
	ここで、
	任意の$i (\in [1, n])$に対して、
	$q'_{i} = q_{i}$と定義し、
	$q'_{n+1} = q_{n}+1$と定義し、
	$q'_{n+2} = m$と定義する。
	任意の$i (\in [1, n+1])$に対して、
	$a_{i} = \sum_{j = q'_{n+2-i}+1}^{q'_{n+3-i}} s_{j}$と定義し、
	$a_{0} = 0$と定義する。
	任意の$x (\in [0, n])$に対して、
	整数時間$t = (\sum_{j = 0}^{x} a_{j} ) + 1, \ldots, \sum_{j = 0}^{x+1} a_{j}$に、
	scheduling eventが発生する。
	更に、
	任意の$x (\in [0, n])$に対して、
	時刻$t + \frac{1}{2}$において、
	$Q^{(q'_{n+1-x})}$にpacketが到着する様なarrival eventが発生する。
	時刻$(\sum_{j = 0}^{n+1} a_{j})+1$の後は、
	十分な数のscheduling eventが発生する。
	このとき、
	$PQ$が$\sigma'$において各scheduling eventにtransmitするpacketは、
	$\sigma$と全く同じである。
	また、
	$t$に$Q^{(q'_{n+1-x})}$からpacketをtransmitする
	オフラインアルゴリズム$ALG$を考える。
	$q'_{i}$の定義より、
	任意の$i (\in [1, n+1])$に対して、
	$h_{PQ}^{(q'_{i})}(1-) = B$が成立しているので、
	$PQ$は各時刻$t + \frac{1}{2}$において、到着する全てのpacketをacceptできない。
	しかし、$ALG$は全てのpacketを受理することが出来る。
	すなわち、$ALG$は最適なオフラインアルゴリズムの1つである。
	よって、
	$n(\sigma') = n+1$
	かつ
	任意の$i (\in [1, n+1])$に対して、
	$q_{i}(\sigma') = q'_{i}$
	が成立する。
	明らかに、任意の$j (\in [1, m])$に対して、
	$s_{j}(\sigma') = s_{j}$
	が成立しており、
	$V_{PQ}(\sigma') = V_{PQ}(\sigma)$
	が成立する。
	また、
	任意の$i (\in [1, n-1])$に対して、
	$k_{q_{i}}(\sigma') = k_{q_{i}}$
	が成立し、
	$k_{q_{n}}(\sigma') = s_{q_{n}+1}$
	が成立し、
	$k_{q_{n}+1}(\sigma') = \sum_{j = q_{n}+2}^{m} s_{j}$
	\com{（■）}
	が成立する。
	よって、
	$\sigma' \in {\cal S}_{4}$
	が成立し、
	$\sigma'$はステートメントの条件(i)と(ii)をみたす。
	更に、
	$V_{OPT}(\sigma') = V_{ALG}(\sigma') 
		= V_{OPT}(\sigma) + (\alpha_{q_{n}+1} - \alpha_{q_{n}}) \sum_{j = q_{n}+2}^{m} s_{j}
		\geq V_{OPT}(\sigma)$
	が成立する。
	\fi
	\ifnum \count11 > 0
	\com{（■英語）}
	%
	%
%
	%
	We construct $\sigma'$ from $\sigma$ as follows: 
	First, 
	for each $j (\in [q_{1}, m])$, 
	$s_{j}$ events at which $s_{j}$ packets arrive at $Q^{(j)}$ occur during time $(0, 1)$. 
	Since $s_{j} \leq B$ by the definition of ${\cal S}_{4}$, 
	$PQ$ accepts all these arriving packets. 
	For any $i (\in [1, n])$, 
	we define $q'_{i} = q_{i}$,  
	$q'_{n+1} = q_{n}+1$ and 
	$q'_{n+2} = m$. 
	Moreover, 
	for any $i (\in [1, n+1])$, 
	we define $a_{i} = \sum_{j = q'_{n+2-i}+1}^{q'_{n+3-i}} s_{j}$ and 
	$a_{0} = 0$. 
	Then, 
	for any $x (\in [0, n])$, 
	a scheduling event occurs at each integer time $t = (\sum_{j = 0}^{x} a_{j} ) + 1, \ldots, \sum_{j = 0}^{x+1} a_{j}$. 
	In addition, 
	for any $x (\in [0, n])$, 
	an arrival event where a packet arrives at $Q^{(q'_{n+1-x})}$ occurs at each time $t + \frac{1}{2}$. 
	After time $(\sum_{j = 0}^{n+1} a_{j}) + 1$, 
	sufficient scheduling events to transmit all the arriving packets occur. 
	Then, 
	the packets which $PQ$ transmits at each scheduling event for $\sigma'$ are equivalent to those for $\sigma$. 
	Consider an offline algorithm $ALG$ which transmits a packet from $Q^{(q'_{n+1-x})}$ at $t$. 
	By the definition of $q'_{i}$, 
	since for any $i (\in [1, n+1])$, $h_{PQ}^{(q'_{i})}(1-) = B$, 
	$PQ$ cannot accept any packet which arrives at each time $t + \frac{1}{2}$,  
	but $ALG$ can accept all the packets, which means that  
	$ALG$ is optimal. 
	Hence, 
	$n(\sigma') = n+1$, and 
	for any $i (\in [1, n+1])$, 
	$q_{i}(\sigma') = q'_{i}$. 
	Since for any $j (\in [1, m])$, $s_{j}(\sigma') = s_{j}$, 
	$V_{PQ}(\sigma') = V_{PQ}(\sigma)$. 
	Moreover, 
	for any $i (\in [1, n-1])$, 
	$k_{q_{i}}(\sigma') = k_{q_{i}}$, 
	$k_{q_{n}}(\sigma') = s_{q_{n}+1}$, and 
	$k_{q_{n}+1}(\sigma') = \sum_{j = q_{n}+2}^{m} s_{j}$. 
	Therefore, 
	$\sigma' \in {\cal S}_{4}$ holds, and 
	$\sigma'$ satisfies the conditions (i) and (ii) in the statements. 
	Also, 
	$V_{OPT}(\sigma') = V_{ALG}(\sigma') 
		= V_{OPT}(\sigma) + (\alpha_{q_{n}+1} - \alpha_{q_{n}}) \sum_{j = q_{n}+2}^{m} s_{j}
		\geq V_{OPT}(\sigma)$. 
	\fi
	%
%

%
\subsection{Proof of Lemma~\ref{LMA:lower}}
%

%
	%
	\ifnum \count10 > 0
	（■まだ）
	\fi
	\ifnum \count11 > 0
	\com{（■英語）}
	Consider the following input $\sigma$. 
	Define $m' \in \arg \min_{x \in [1, m-1]} \{ \frac{ \alpha_{x+1} }{ \sum_{j = 1}^{x+1} \alpha_{j} } \}$. 
	Initially, 
	$(m'+1) B$ arrival events happen such that $B$ packets arrive at $Q^{(1)}$ to $Q^{(m'+1)}$. 
	Then, 
	for $k = 1, 2, \ldots, m'$, 
	the $k$th round consists of $B$ scheduling events followed by $B$ arrival events 
	in which all the $B$ packets arrive at $Q^{(m'-k+1)}$.
	For $\sigma$, $PQ$ transmits $B$ packets from $Q^{(m'-k+2)}$ at the $k$th
	round. 
	As a result, 
	$PQ$ cannot accept arriving packets in the $(k+1)$st round. 
	Hence, 
	${V}_{PQ}(\sigma) = B \sum_{j = 1}^{m'+1} \alpha_{j}$ holds. 
	On the other hand, 
	$OPT$ transmits $B$ packets from $Q^{(m'-k+1)}$ at the $k$th round, 
	and hence can accept all the arriving packets. 
	Thus, 
	${V}_{OPT}(\sigma) = 2B \sum_{j = 1}^{m'} \alpha_{j} + B \alpha_{m'+1}$. 
	Therefore,
	$\frac{{V}_{OPT}(\sigma)}{{V}_{PQ}(\sigma)} 
		= \frac{ 2 \sum_{j = 1}^{m'} \alpha_{j} + \alpha_{m'+1}}{ \sum_{j = 1}^{m'+1} \alpha_{j} }
		= 2 - \frac{ \alpha_{m'+1} }{ \sum_{j = 1}^{m'+1} \alpha_{j} }$. 
	(It is easy to see that $\sigma \in {\cal S}_{5}$.)
	\fi
	%
%

%
\subsection{Proof of Theorem~\ref{thm:3}}
%

	%
	\ifnum \count10 > 0
	オンラインアルゴリズム$ON$を固定する。
	次の様な入力$\sigma$を敵対者が作成する。
	$\sigma(t)$を$\sigma$の時刻$t$までの冒頭部分とする。
	$OPT$は、
	この入力において到着するパケットを全て受理し、scheduleすることが可能であることに注意せよ。
	時刻$(0, 1)$に
	$2B$回のarrival eventが発生し、
	$Q^{(1)}$と$Q^{(m)}$に$B$個ずつpacketが到着する。
	時刻$(1, 2)$に
	$B$回のscheduling eventが発生する。
	この$\sigma(2)$にたいして、
	時刻$(1, 2)$に
	online algorihtm $ON$が
	$Q^{(1)}$から$B(1-x)$個、
	$Q^{(m)}$から$Bx$個ずつpacketをscheduleするとする。
	（図~\ref{fig:L1}参照）
	時刻$2$より後の時刻では、
	$Q^{(1)}$と$Q^{(m)}$を比較して、
	より$ON$が損をするキューにパケットが到着する。
	{\bf\boldmath Case 1. $\alpha x \geq 1 - x$が成立する場合:}
	時刻$(2, 3)$に
	$B$回のarrival eventが発生し、
	$Q^{(1)}$に$B$個のpacketが到着する。
	このとき、時刻$3$までに、
	$ON$は$(\alpha + 1 + 1 - x )B$の価値のパケットを受理している。
	更に、
	時刻$(3, 4)$に
	$B$回のscheduling eventが発生する。
	この$\sigma(4)$にたいして、
	時刻$(3, 4)$に
	$ON$が
	$Q^{(1)}$から$B(1-y)$個、
	$Q^{(m)}$から$By$個ずつpacketをscheduleするとする。
	（図~\ref{fig:L2}参照）
	時刻$4$より後の時刻では、
	先の場合と同様に、
	$Q^{(1)}$と$Q^{(m)}$を比較して、
	より$ON$が損をするキューにパケットが到着する。
	{\bf\boldmath Case 1.1. $\alpha (x+y) \geq 1 - y$が成立する場合:}
	時刻$(4, 5)$に
	$B$回のarrival eventが発生し、
	$Q^{(1)}$に$B$個のpacketが到着する。
	時刻$(5, 6)$に
	$2B$回のscheduling eventが発生する。
	この入力にたいして、
	$V_{ON}(\sigma) = ( \alpha + 1 + 1 - x + 1 - y )B$、
	$V_{OPT}(\sigma) = ( \alpha + 1 + 1 + 1 )B$が成立する。
	{\bf\boldmath Case 1.2. $\alpha (x+y) < 1 - y$が成立する場合:}
	時刻$(4, 5)$に
	$B$回のarrival eventが発生し、
	$Q^{(m)}$に$B$個のpacketが到着する。
	時刻$(5, 6)$に
	$2B$回のscheduling eventが発生する。
	この入力にたいして、
	$V_{ON}(\sigma) = ( \alpha + 1 + 1 - x + \alpha(x + y) )B$、
	$V_{OPT}(\sigma) = ( \alpha + 1 + 1 + \alpha )B$
	が成立する。
	{\bf\boldmath Case 2. $\alpha x < 1 - x$が成立する場合:}
	時刻$(2, 3)$に
	$B$回のarrival eventが発生し、
	$Q^{(m)}$に$B$個のpacketが到着する。
	このとき、時刻$3$までに、
	$ON$は$( \alpha + 1 + \alpha x )B$の価値のパケットを受理している。
	時刻$(3, 4)$に
	$B$回のscheduling eventが発生する。
	この$\sigma(4)$にたいして、
	時刻$(3, 4)$に
	$ON$が
	$Q^{(1)}$から$B(1-z)$個、
	$Q^{(m)}$から$Bz$個ずつpacketをscheduleするとする。
	（図~\ref{fig:L3}参照）
	時刻$4$より後の時刻では、
	先の場合と同様に、
	$Q^{(1)}$と$Q^{(m)}$を比較して、
	より$ON$が損をするキューにパケットが到着する。
	{\bf\boldmath Case 2.1. $\alpha z \geq 1 - x + 1 - z$が成立する場合:}
	時刻$(4, 5)$に
	$B$回のarrival eventが発生し、
	$Q^{(1)}$に$B$個のpacketが到着する。
	時刻$(5, 6)$に
	$2B$回のscheduling eventが発生する。
	この入力にたいして、
	$V_{ON}(\sigma) = ( \alpha + 1 + \alpha x + 1 - x + 1 - z )B$、
	$V_{OPT}(\sigma) = ( \alpha + 1 + \alpha + 1 )B$が成立する。
	{\bf\boldmath Case 2.2. $\alpha z < 1 - x + 1 - z$が成立する場合:}
	時刻$(4, 5)$に
	$B$回のarrival eventが発生し、
	$Q^{(m)}$に$B$個のpacketが到着する。
	時刻$(5, 6)$に
	$2B$回のscheduling eventが発生する。
	この入力にたいして、
	$V_{ON}(\sigma) = ( \alpha + 1 + \alpha x + \alpha z )B$、
	$V_{OPT}(\sigma) = ( \alpha + 1 + \alpha + \alpha )B$が成立する。
	上記の議論より、
	$c_{1}(x) = \min_{y} \max\{
				\frac{ \alpha + 1 + 1 + 1 }{ \alpha + 1 + 1 - x + 1 - y }, 
				\frac{\alpha + 1 + 1 + \alpha }{\alpha + 1 + 1 - x + \alpha(x + y)}
			\}
	$
	と定義し、
	$c_{2}(x) = \min_{z} \max\{
				\frac{ \alpha + 1 + \alpha + 1 }{ \alpha + 1 + \alpha x + 1 - x + 1 - z }, 
				\frac{ \alpha + 1 + \alpha + \alpha }{ \alpha + 1 + \alpha x + \alpha z }
			\}
	$
	と定義して、
	$\frac{V_{OPT}(\sigma)}{V_{ON}(\sigma)} 
		\geq
			\min_{x} \max\{
				c_{1}(x), c_{2}(x)
			\}
	$
	が成立する。
	$c_{1}(x)$は、
	$\frac{ \alpha + 1 + 1 + 1 }{ \alpha + 1 + 1 - x + 1 - y } 
	= 
	\frac{\alpha + 1 + 1 + \alpha }{\alpha + 1 + 1 - x + \alpha(x + y)}$
	が成立するときに最小化される。
	このとき、
	$y = \frac{ \alpha (\alpha + 3) + (-\alpha^2 -4 \alpha + 1)x }
				{ \alpha^2 + 5 \alpha + 2 }$
	が成立し、
	$c_{1}(x) 
		\geq \frac{ \alpha^2 + 5 \alpha + 2 }{ \alpha^2 + 4 \alpha + 2 - x }$
	が成立する。
	$c_{2}(x)$は、
	$\frac{ \alpha + 1 + \alpha + 1 }{ \alpha + 1 + \alpha x + 1 - x + 1 - z } 
	= 
	\frac{ \alpha + 1 + \alpha + \alpha }{ \alpha + 1 + \alpha x + \alpha z }$
	が成立するときに最小化される。
	このとき、
	$z = \frac{ \alpha^2 + 6 \alpha + 1 + (\alpha^2 -4 \alpha - 1)x }
				{ 2 \alpha^2 + 5 \alpha + 1 }$
	が成立し、
	$c_{2}(x) 
		\geq \frac{ 2 \alpha^2 + 5 \alpha + 1 }{ \alpha^2 + 4 \alpha + 1 + \alpha^2 x }$
	が成立する。
	$\min_{x} \max\{	c_{1}(x), c_{2}(x) \}$は、
	$c_{1}(x) = c_{2}(x)$、すなわち、
	$\frac{ \alpha^2 + 5 \alpha + 2 }{ \alpha^2 + 4 \alpha + 2 - x }
		= \frac{ 2 \alpha^2 + 5 \alpha + 1 }{ \alpha^2 + 4 \alpha + 1 + \alpha^2 x }$
	が成立するときに最小化される。
	このとき、
	$x = \frac{ \alpha^4 + 4 \alpha^3 + 2 \alpha^2 + \alpha }{ \alpha^4 + 5 \alpha^3 + 4 \alpha^2 +5  \alpha + 1} $
	が成立し、
	$\min_{x} \max\{	c_{1}(x), c_{2}(x) \}
		\geq \frac{ \alpha^4 + 5 \alpha^3 + 4 \alpha^2 + 5 \alpha + 1 }
				{ \alpha^4 + 4 \alpha^3 + 3 \alpha^2 + 4 \alpha + 1 }
			= 1 + \frac{ \alpha^3 + \alpha^2 +  \alpha }{ \alpha^4 + 4 \alpha^3 + 3 \alpha^2 + 4 \alpha + 1 }
				$
	が成立する。
	\fi
	\ifnum \count11 > 0
	\com{（■英語）}
	Fix an online algorithm $ON$. 
	Our adversary constructs the following input $\sigma$. 
	Let $\sigma(t)$ denote the prefix of the input $\sigma$ up to time $t$. 
	$OPT$ can accept and transmit all arriving packets in this input. 
	$2B$ arrival events occur during time $(0, 1)$, and 
	$B$ packets arrive at $Q^{(1)}$ and $Q^{(m)}$, respectively. 
	In addition, 
	$B$ scheduling events occur during time $(1, 2)$. 
	For $\sigma(2)$, 
	suppose that $ON$ transmits $B(1-x)$ packets and $Bx$ ones from $Q^{(1)}$ and $Q^{(m)}$, respectively. 
	(See Figure~\ref{fig:L1}.)
	After time 2, 
	our adversary selects one queue from $Q^{(1)}$ and $Q^{(m)}$, and makes some packets arrive at the queue. 
	{\bf\boldmath Case 1: If $\alpha x \geq 1 - x$:}
	$B$ arrival events occur during time $(2, 3)$, and 
	$B$ packets arrive at $Q^{(1)}$. 
	Then, 
	the total value of packets which $ON$ accepts by time $3$ is $(\alpha + 1 + 1 - x )B$. 
	Moreover, 
	$B$ scheduling events occur during time $(3, 4)$. 
	For $\sigma(4)$, 
	suppose that $ON$ transmits $B(1-y)$ packets and $By$ packets from $Q^{(1)}$ and $Q^{(m)}$, respectively. 
	(See Figure~\ref{fig:L2}.)
	After time 4, in the same way as time 2, 
	our adversary selects one queue from $Q^{(1)}$ and $Q^{(m)}$, and makes some packets arrive at the queue. 
	{\bf\boldmath Case 1.1: If $\alpha (x+y) \geq 1 - y$:}
	$B$ arrival events occur during time $(4, 5)$, and 
	$B$ packets arrive at $Q^{(1)}$. 
	Furthermore, 
	$2B$ scheduling events occur during time $(5, 6)$. 
	For this input,
	$V_{ON}(\sigma) = ( \alpha + 1 + 1 - x + 1 - y )B$, and 
	$V_{OPT}(\sigma) = ( \alpha + 1 + 1 + 1 )B$. 
	{\bf\boldmath Case 1.2: If $\alpha (x+y) < 1 - y$:}
	$B$ arrival events occur during time $(4, 5)$, and 
	$B$ packets arrive at $Q^{(m)}$. 
	Moreover, 
	$2B$ scheduling events occur during time $(5, 6)$. 
	For this input,
	$V_{ON}(\sigma) = ( \alpha + 1 + 1 - x + \alpha(x + y) )B$, and 
	$V_{OPT}(\sigma) = ( \alpha + 1 + 1 + \alpha )B$. 
	{\bf\boldmath Case 2: If $\alpha x < 1 - x$:}
	$B$ arrival events occur during time $(2, 3)$, and 
	$B$ packets arrive at $Q^{(m)}$. 
	Then, 
	the total value of packets which $ON$ accepts by time $3$ is $( \alpha + 1 + \alpha x )B$.  
	Moreover, 
	$B$ scheduling events occur during time $(3, 4)$. 
	For $\sigma(4)$, 
	$ON$ transmits $B(1-z)$ packets and $Bz$ ones from $Q^{(1)}$ and $Q^{(m)}$, respectively during time $(3, 4)$. 
	(See Figure~\ref{fig:L3}.)
	After time 4, in the same way as the above case, 
	our adversary  selects one queue from $Q^{(1)}$ and $Q^{(m)}$, and causes some packets to arrive at the queue. 
	{\bf\boldmath Case 2.1: If $\alpha z \geq 1 - x + 1 - z$:}
	$B$ arrival events occur during time $(4, 5)$, and 
	$B$ packets arrive at $Q^{(1)}$. 
	Also,
	$2B$ scheduling events occur during time $(5, 6)$. 
	For this input, 
	$V_{ON}(\sigma) = ( \alpha + 1 + \alpha x + 1 - x + 1 - z )B$, and 
	$V_{OPT}(\sigma) = ( \alpha + 1 + \alpha + 1 )B$. 
	{\bf\boldmath Case 2.2: If $\alpha z < 1 - x + 1 - z$:}
	$B$ arrival events occur during time $(4, 5)$, and 
	$B$ packets arrive at $Q^{(m)}$. 
	In addition, 
	$2B$ scheduling events occur during time $(5, 6)$. 
	For this input, 
	$V_{ON}(\sigma) = ( \alpha + 1 + \alpha x + \alpha z )B$, and 
	$V_{OPT}(\sigma) = ( \alpha + 1 + \alpha + \alpha )B$. 
	By the above argument, 
	we define $c_{1}(x) = \min_{y} \max\{
				\frac{ \alpha + 1 + 1 + 1 }{ \alpha + 1 + 1 - x + 1 - y }, 
				\frac{\alpha + 1 + 1 + \alpha }{\alpha + 1 + 1 - x + \alpha(x + y)}
			\}
	$ and 
	$c_{2}(x) = \min_{z} \max\{
				\frac{ \alpha + 1 + \alpha + 1 }{ \alpha + 1 + \alpha x + 1 - x + 1 - z }, 
				\frac{ \alpha + 1 + \alpha + \alpha }{ \alpha + 1 + \alpha x + \alpha z }
			\}
	$. 
	Then, 
	$\frac{V_{OPT}(\sigma)}{V_{ON}(\sigma)} 
		\geq
			\min_{x} \max\{
				c_{1}(x), c_{2}(x)
			\}
	$. 
	$c_{1}(x)$ is minimized when 
	$\frac{ \alpha + 1 + 1 + 1 }{ \alpha + 1 + 1 - x + 1 - y } 
	= 
	\frac{\alpha + 1 + 1 + \alpha }{\alpha + 1 + 1 - x + \alpha(x + y)}$. 
	Then, 
	$y = \frac{ \alpha (\alpha + 3) + (-\alpha^2 -4 \alpha + 1)x }
				{ \alpha^2 + 5 \alpha + 2 }$.  
	Thus, 
	$c_{1}(x) 
		\geq \frac{ \alpha^2 + 5 \alpha + 2 }{ \alpha^2 + 4 \alpha + 2 - x }$. 
	$c_{2}(x)$ is minimized when 
	$\frac{ \alpha + 1 + \alpha + 1 }{ \alpha + 1 + \alpha x + 1 - x + 1 - z } 
	= 
	\frac{ \alpha + 1 + \alpha + \alpha }{ \alpha + 1 + \alpha x + \alpha z }$. 
	Then, 
	$z = \frac{ \alpha^2 + 6 \alpha + 1 + (\alpha^2 -4 \alpha - 1)x }
				{ 2 \alpha^2 + 5 \alpha + 1 }$. 
	Hence, 
	$c_{2}(x) 
		\geq \frac{ 2 \alpha^2 + 5 \alpha + 1 }{ \alpha^2 + 4 \alpha + 1 + \alpha^2 x }$. 
	Finally, 
	$\min_{x} \max\{	c_{1}(x), c_{2}(x) \}$ is minimized when 
	$c_{1}(x) = c_{2}(x)$, that is 
	$\frac{ \alpha^2 + 5 \alpha + 2 }{ \alpha^2 + 4 \alpha + 2 - x }
		= \frac{ 2 \alpha^2 + 5 \alpha + 1 }{ \alpha^2 + 4 \alpha + 1 + \alpha^2 x }$. 
	Therefore, 
	since $x = \frac{ \alpha^4 + 4 \alpha^3 + 2 \alpha^2 + \alpha }{ \alpha^4 + 5 \alpha^3 + 4 \alpha^2 +5  \alpha + 1}$, 
	$\min_{x} \max\{	c_{1}(x), c_{2}(x) \}
		\geq \frac{ \alpha^4 + 5 \alpha^3 + 4 \alpha^2 + 5 \alpha + 1 }
				{ \alpha^4 + 4 \alpha^3 + 3 \alpha^2 + 4 \alpha + 1 }
			= 1 + \frac{ \alpha^3 + \alpha^2 +  \alpha }{ \alpha^4 + 4 \alpha^3 + 3 \alpha^2 + 4 \alpha + 1 }
				$. 
	\fi
	%
%

%
\ifnum \count12 > 0
\begin{figure*}[h]
 \begin{center}
  \includegraphics[width=120mm]{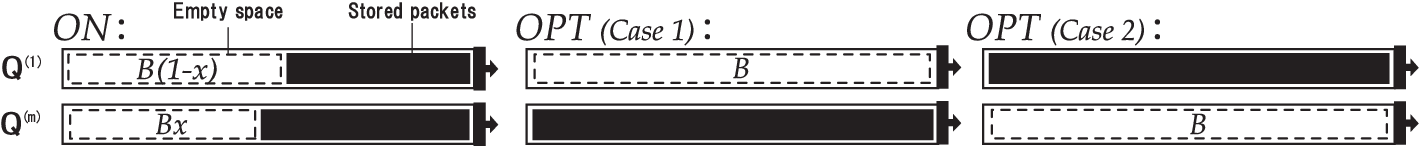}
 \end{center}
 \caption{States of queues at time 2}
\label{fig:L1}
 \end{figure*}

\begin{figure*}[h]
 \begin{center}
  \includegraphics[width=120mm]{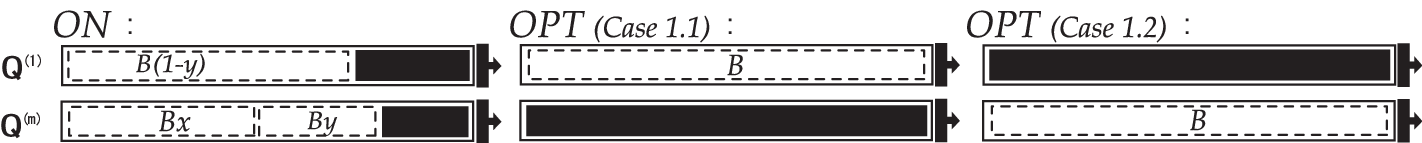}
 \end{center}
 \caption{States of queues at time 4 via Case 1}
 \label{fig:L2}
\end{figure*}
\begin{figure*}[h]
 \begin{center}
  \includegraphics[width=120mm]{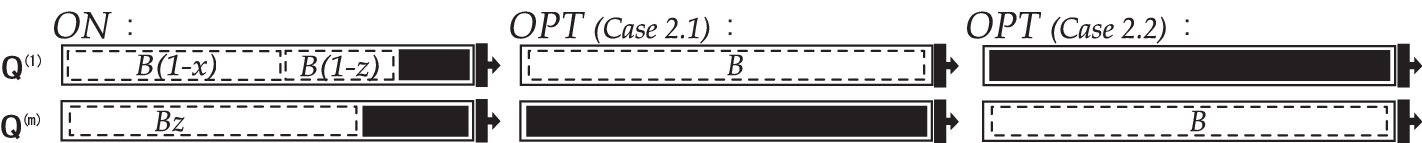}
 \end{center}
 \caption{States of queues at time 4 via Case 2}
 \label{fig:L3}
\end{figure*}
\fi
%

\fi


\newpage

\begin{thebibliography}{99}


\bibitem{WA00}
W. Aiello, Y. Mansour, S. Rajagopolan, and A. Ros\'{e}n,
``Competitive queue policies for differentiated services,''
{\it Journal of Algorithms},
 Vol. 55, No. 2, pp. 113--141,
2005. 


\bibitem{KA11}
K. Al-Bawani, and A. Souza, 
``Buffer overflow management with class segregation,''
{\it Information Processing Letters},
 Vol. 113, No. 4, pp. 145--150,
2013. 


\bibitem{SA07}
S. Albers and T. Jacobs, 
``An experimental study of new and known online packet buffering algorithms,''
{\it Algorithmica}, Vol.57, No.4, 
pp, 725--746, 2010.

\bibitem{SA04}
S. Albers and M. Schmidt, 
``On the performance of greedy algorithms in packet buffering,''
{\it SIAM Journal on Computing}, Vol. 35, No. 2,
pp. 278--304, 2005.

\bibitem{NA05}
N. Andelman,
``Randomized queue management for DiffServ,'' 
{\it In Proc. of the 17th ACM Symposium on Parallel Algorithms and Architectures},
pp. 1--10, 2005.


\bibitem{NACM03}
N. Andelman and Y. Mansour,
``Competitive management of non-preemptive queues with multiple values,'' 
{\it Distributed Computing}, pp. 166--180, 2003.

\bibitem{NACQ03}
N. Andelman, Y. Mansour and A. Zhu,
``Competitive queueing policies for QoS switches,''
{\it In Proc. of the 14th ACM-SIAM Symposium on Discrete Algorithms}, 
pp. 761--770, 2003.






\bibitem{YAM04}
Y. Azar and A. Litichevskey, 
``Maximizing throughput in multi-queue switches,''
{\it Algorithmica}, Vol.45, No. 1, pp, 69--90, 2006,


\bibitem{YA03}
Y. Azar and Y. Richter,
``Management of multi-queue switches in QoS networks,''
{\it Algorithmica}, Vol.43, No. 1-2, pp, 81--96, 2005,



\bibitem{YA06}
Y. Azar and Y. Richter,
``An improved algorithm for CIOQ switches,''
{\it ACM Transactions on Algorithms},
 Vol. 2, No. 2, pp. 282--295,
2006,

\bibitem{YAZ04}
Y. Azar and Y. Richter,
``The zero-one principle for switching networks,''
{\it In Proc. of the 36th ACM Symposium on Theory of Computing},
pp. 64--71, 2004.


%

\bibitem{AB03}
A. Bar-Noy, A. Freund, S. Landa and J. Naor, 
``Competitive on-line switching policies,''
{\it Algorithmica}, 
 Vol. 36, No. 3, pp. 225--247, 2003,

 


\bibitem{MB08}
M. Bienkowski and A. Madry,
``Geometric aspects of online packet buffering: an optimal randomized algorithm for two buffers,''
{\it In Proc. of the 8th Latin American Theoretical Informatics},
pp. 252--263, 2008. 



\bibitem{MB10}
M. Bienkowski,
``An optimal lower bound for buffer management in multi-queue switches,''
{\it Algorithmica}, Vol.68, No.2, 
pp, 426--447, 2014.


\bibitem{SB98}
S. Blanke, D. Black, M. Carlson, E. Davies, Z. Wang, and W. Weiss
``An architecture for differentiated services''
{\it RFC2475, IETF},
December 1998.


\bibitem{AB98}
A. Borodin and R. El-Yaniv,
``Online computation and competitive analysis,''
{\it Cambridge University Press},
1998.





\bibitem{QoSSwitch}
Cisco Systems, Inc,``Campus QoS Design'', 
http://www.cisco.com/en/US/docs/solutions/\\Enterprise/WAN\_and\_MAN/QoS\_SRND/QoSDesign.html, 2014

\bibitem{cat2955}
Cisco Systems, Inc,``Cisco Catalyst 2955 series switches data sheets'', 
http://www.cisco.com/en/US/products/hw/switches/ps628/products\_data\_\\sheets\_list.html, 2014


\bibitem{cat6500}
Cisco Systems, Inc,``Cisco Catalyst 6500 series switches data sheets'', 
http://www.cisco.com/en/US/products/hw/switches/ps708/products\_data\_\\sheets\_list.html, 2014


\bibitem{AD90}
A. Demers, S. Keshav, and S. Shenker,
``Analysis and simulation of a fair queueing algorithm'' 
{\it journal of Internetworking Research and Experience}, Vol.1, No.1,
pp. 3--26, 1990.


\bibitem{ME06}
M. Englert and M. Westermann,
``Lower and upper bounds on FIFO buffer management in QoS switches,''
{\it Algorithmica}, Vol.53, No.4, 
pp, 523--548, 2009.




\bibitem{MG10}
M. Goldwasser,
``A survey of buffer management policies for packet switches,''
ACM SIGACT News,
Vol.41, No. 1, pp.100--128,
2010.

\bibitem{EH01}
E. Hahne, A. Kesselman and Y. Mansour,
``Competitive buffer management for shared-memory switches,''
{\it In Proc. of the 13th ACM Symposium on Parallel Algorithms and Architectures},
pp. 53--58, 2001.







\bibitem{MK91}
M. Katevenis, S. Sidiropopulos, and C. Courcoubetis,
``Weighted round-robin cell multiplexing in a general-purpose ATM switch chip'' 
{\it IEEE Journal on Selected Area in Communications}, Vol. 9, No. 8,
pp. 1265--1279, October 1991.


\bibitem{AKB01}
A. Kesselman, Z. Lotker, Y. Mansour, B. Patt-Shamir, B. Schieber, and M. Sviridenko,
``Buffer overflow management in QoS switches,''
{\it SIAM Journal on Computing}, Vol. 33, No. 3,
pp. 563--583, 2004.



\bibitem{AKH01}
A. Kesselman and Y. Mansour,
``Harmonic buffer management policy for shared memory switches,''
{\it Theoretical Computer Science}, Vol. 324, No. 2-3,
pp. 161--182, 2004.


\bibitem{AKI02}
A. Kesselman, Y. Mansour and R. van Stee,
``Improved competitive guarantees for QoS buffering,''
{\it Algorithmica}, Vol.43, No.1-2,
pp. 63--80, 2005.


\bibitem{AKS03}
A. Kesselman and A. Ros\'{e}n,
``Scheduling policies for CIOQ switches,''
{\it Journal of Algorithms},
 Vol. 60, No. 1, pp. 60--83,
2006,

\bibitem{AKC08}
A. Kesselman and A. Ros\'{e}n,
``Controlling CIOQ switches with priority queuing and in multistage interconnection networks,''
{\it Journal of Interconnection Networks},
 Vol. 9, No. 1/2, pp. 53--72,
2008,


\bibitem{AKP08}
A. Kesselman, K. Kogan and M. Segal,
``Packet mode and QoS algorithms for buffered crossbar switches with FIFO queuing,''
{\it Distributed Computing}, Vol.23, No.3, pp. 163--175, 2010.


\bibitem{AKB08}
A. Kesselman, K. Kogan and M. Segal,
``Best effort and priority queuing policies for buffered crossbar switches,''
{\it Chicago Journal of Theoretical Science}, pp. 1--14, 2012,


\bibitem{AKI08}
A. Kesselman, K. Kogan and M. Segal,
``Improved competitive performance bounds for CIOQ switches,''
{\it Algorithmica}, Vol.63, No.1-2, pp, 411--424, 2012. 


\bibitem{KK07}
K. Kobayashi, S. Miyazaki and Y. Okabe,
``A tight bound on online buffer management for two-port shared-memory switches,''
{\it In Proc. of the 19th ACM Symposium on Parallel Algorithms and Architectures},
pp. 358--364, 2007.

\bibitem{KK08}
K. Kobayashi, S. Miyazaki and Y. Okabe,
``A tight upper bound on online buffer management for multi-queue switches with bicodal buffers,''
{\it IEICE TRANSACTIONS on Fundamentals of Electronics, Communications and Computer Sciences}, Vol. E91-D, No. 12,
pp. 2757--2769, 2008.

\bibitem{KK09}
K. Kobayashi, S. Miyazaki and Y. Okabe,
``Competitive buffer management for multi-queue switches in QoS networks using packet buffering algorithms,''
{\it In Proc. of the 21st ACM Symposium on Parallel Algorithms and Architectures},
pp. 328--336, 2009.

\bibitem{KK13}
K. Kogan, A. Lopez-Ortiz, S. Nikolenko, and A. Sirotkin,
``Multi-queued network processors for packets with heterogeneous processing requirements,''
{\it In Proc. of the 5th International Conference on Communication Systems and Networks},
pp. 1--10, 2013.


\bibitem{HK04}
R. Fleischer and H. Koga,
``Balanced scheduling toward loss-free packet queuing and delay fairness,''
{\it Algorithmica}, 
 Vol. 38, No. 2, pp. 363--376, 2004,
 







%




%



\bibitem{DS85}
D. Sleator and R. Tarjan,
``Amortized efficiency of list update and paging rules,''
{\it Communications of the ACM}, Vol. 28, No. 2,
pp. 202--208, 1985.

\bibitem{MS01}
M. Sviridenko,
``A lower bound for on-line algorithms in the FIFO model,''
unpublished manuscript, 2001.


\end{thebibliography}
\end{document}